\documentclass{article}

\usepackage{arxiv}

\usepackage{chngcntr}
\usepackage[UKenglish]{babel}
\usepackage{oldgerm}
\usepackage{bm}
\usepackage{extarrows}
\usepackage{amsfonts}
\usepackage{mathabx}
\usepackage{amsmath, amssymb, amsthm}
\usepackage{mathrsfs}
\usepackage{geometry}
\usepackage{graphicx}
\usepackage{dsfont}
\usepackage{wasysym}
\usepackage[T3,T1]{fontenc}
\usepackage{inputenc}
\usepackage{pstricks-add}
\usepackage[arrow, matrix, curve]{xy}
\usepackage{enumitem}
\usepackage{epigraph}
\usepackage{setspace}
\usepackage[all]{nowidow}
\usepackage[colorlinks,pdfpagelabels,pdfstartview=FitH,bookmarksopen=false,bookmarksnumbered=true,linkcolor=black,plainpages=false,hypertexnames=false,citecolor=black]{hyperref}

\newtheorem{Def}{Definition}[section]
\newtheorem{Prop}[Def]{Proposition}
\newtheorem{Lem}[Def]{Lemma}
\newtheorem{Thm}[Def]{Theorem}

\newtheorem{Rmk}[Def]{Remark}
\newtheorem{Cor}[Def]{Corollary}

\newcommand{\N}{\mathbb{N}}

\newcommand{\R}{\mathbb{R}}
\newcommand{\C}{\mathbb{C}}

\newcommand{\A}{\mathcal{A}}

\newcommand{\D}{\mathscr{D}}

\newcommand{\E}{\mathcal{E}}

\newcommand{\T}{\mathcal{T}}

\newcommand{\KK}{\mathcal{K}}

\newcommand{\QQ}{\mathcal{Q}}
\newcommand{\Riesz}{\mathscr{R}}

\newcommand{\LL}{\mathcal{L}}

\renewcommand{\S}{\mathcal{S}}
\newcommand{\e}{\varepsilon}
\newcommand{\f}{\varphi}
\newcommand{\sLL}{\mathcal{L}^\uparrow_+}

\renewcommand{\O}{\mathcal{O}}

\newcommand{\HR}{\mathscr{H}}
\renewcommand{\L}{\mathscr{L}}
\newcommand{\RP}{\mathring{\R}}

\newcommand{\CGerm}{\mathscr{C}^\infty}
\newcommand{\CC}{C^\infty}

\newcommand{\CCsc}{C^\infty_{sc}}

\newcommand{\Obar}{\overline{O}}

\newcommand{\grad}{\text{grad}}

\newcommand{\df}{\, \text{d}}
\newcommand{\id}{\text{id}}
\newcommand{\ran}{\mathrm{ran}\,}
\newcommand{\RT}[1]{\text{Re}\left(#1\right)}
\newcommand{\IT}[1]{\text{Im}\left(#1\right)}

\newcommand{\supp}{\text{supp} \,}
\newcommand{\ssupp}{\text{sing} \, \supp}
\newcommand{\WF}{\text{WF}}
\newcommand{\Char}{\mathrm{Char}\,}

\newcommand{\Hom}{\textup{Hom}}

\newcommand{\sgn}{\text{sgn}}

\newcommand{\vol}{\text{vol}}

\newcommand{\dom}{\mathrm{dom}\,}

\newcommand{\res}{\mathrm{res}}

\newcommand{\lm}[2]{\mathop{\lim}\limits_{#1 \rightarrow #2}}

\newcommand{\Res}[2]{\mathop{\mathrm{Res}}\limits_{#1 = #2}}

\newcommand{\pd}[2]{\frac{\partial #1}{\partial #2}}

\newcommand{\td}[2]{\frac{\df #1}{\df #2}}
\newcommand{\tds}[3]{\left.\frac{\df #1}{\df #2}\right|_{#3}}

\newcommand{\boxx}{\raisebox{-0mm}{\text{\Large{$\Box$}}}}

\newcommand{\cat}[1]{\textup{\textsf{#1}}}
\newcommand{\CCR}{\mathrm{CCR}}

\newcommand{\braket}[2]{\left\langle #1,#2\right\rangle}
\newcommand{\bbraket}[2]{\left\langle\!\left\langle #1,#2\right\rangle\!\right\rangle}

\newcommand{\Braket}[2]{\big\langle #1,#2\big\rangle}

\numberwithin{equation}{section}

\title{Hadamard states for bosonic quantum field theory on globally hyperbolic spacetimes}

\author{
 Max Lewandowski \\
 Institute of Mathematics, Potsdam University\\
 Karl-Liebknecht-Str. 24-25, 14476 Potsdam, Germany\\
 E-mail: mlewando@uni-potsdam.de}

\begin{document}

\begin{spacing}{1.1}

\maketitle

\begin{abstract}
 According to Radzikowski's celebrated results, bisolutions of a wave operator on a globally hyperbolic spacetime are of Hadamard form iff they are given by a linear combination of distinguished parametrices $\frac{i}{2}\big(\widetilde{G}_{aF} - \widetilde{G}_{F} + \widetilde{G}_A - \widetilde{G}_R\big)$ in the sense of Duistermaat-H{\" o}rmander \cite{R1996, DH1972}. Inspired by the construction of the corresponding advanced and retarded Green operator $G_A,G_R$ as done in \cite{BGP2007}, we construct the remaining two Green operators $G_F, G_{aF}$ locally in terms of Hadamard series. Afterwards, we provide the global construction of $\frac{i}{2}\big(\widetilde{G}_{aF} - \widetilde{G}_{F}\big)$, which relies on new techniques like a well-posed Cauchy problem for bisolutions and a patching argument using \v{C}ech cohomology. This leads to global bisolutions of Hadamard form, each of which can be chosen to be a Hadamard two-point-function, i.e. the smooth part can be adapted such that, additionally, the symmetry and the positivity condition are exactly satisfied.

\end{abstract}

\section{Introduction}

\subsection{Quantum field theory on curved space times}

Quantum field theory on curved spacetimes is a semiclassical theory, which investigates the coupling of a quantum field with classical gravitation. It already predicts remarkable effects such as particle creation by the curved spacetime itself, a phenomenon most prominently represented by Hawking's evaporation of black holes \cite{H1975} and the Unruh effect \cite{F1973,D1975,U1976}. Due to their rather small, generally trivial, isometry group, curved spacetimes lack an invariant concept of energy, so a distinct vacuum and consequently the notion of particles turn out to be non-sensible in general curved spacetimes \cite{D1984, W1994}. This situation is best addressed by the algebraic approach to quantum field theory \cite{HK1964, D1980}, where first of all observables are introduced as elements of rather abstract algebras $\A(O)\subset\A(M)$ associated to spacetime regions $O\subset M$ in a local and covariant manner, and states will join the game only later as certain functionals on these algebras.\\
The categorical framework of locally covariant quantum field theory \cite{BFV2003} has established as a suitable generalization of the principles of AQFT to curved spacetimes. Instead of fixing a spacetime and its symmetries from the beginning, a whole category of spacetimes is considered with arrows given by certain isometric embeddings. The actual QFT is then represented by a covariant functor to the category of $C^*$-algebras and injective $C^*$-homomorphisms fulfilling adapted Haag-Kastler-axioms. For the category of spacetimes, we will adopt the setting of \cite{BG2011}:

\begin{samepage}
\begin{Def} \label{GlobHypGreen}
 The category \cat{GlobHypGreen} consists of objects $(M,E,P)$ and morphisms $(f,F)$, where
  \begin{itemize}
   \item $M$ is a globally hyperbolic Lorentzian manifold,
   \item $E$ is a finite-dimensional, real vector bundle over $M$ with a non-degenerate inner product,
   \item $P\colon \CC(M,E)\rightarrow \CC(M,E)$ is a formally self-adjoint, Green hyperbolic operator,
   \item $f\colon M_1 \rightarrow M_2$ time-orientation preserving, isometric embedding with $f(M_1)\subset M_2$ open and causally compatible,
   \item $F$ is a fiberwise isometric vector bundle isomorphism over $f$ such that $P_1,P_2$ are related via \mbox{$P_1\circ\res = \res\circ P_2$}, where $\res(\f):= F^{-1} \circ \f \circ f$ the restriction of $\f\in\CC(M_2,E_2)$ to $M_1$.
  \end{itemize}
\end{Def}
\end{samepage}

Because of their good causal and analytic properties (no causal loops, foliation by Cauchy hypersurfaces, well-posed Cauchy problem), globally hyperbolic Lorentzian manifolds have proven to be a reasonable model for curved spacetimes. As further data, we consider real vector bundles, which excludes for instance charged fields, and the large class of Green hyperbolic operators implying the existence of an advanced and retarded Green operator $G_A,G_R$ (see \cite{Baer2015} for a thorough discussion of these operators). For $G := G_A - G_R$, Theorem 3.5 of \cite{BG2011} provides the exact sequence
\begin{align} \label{ExactSequence}
 \{0\} \longrightarrow \D(M,E) \stackrel{P}{\longrightarrow} \D(M,E) \stackrel{G}{\longrightarrow} \CCsc(M,E) \stackrel{P}{\longrightarrow} \CCsc(M,E),
\end{align}

and hence, it leads to a covariant functor into the category of symplectic vector spaces with objects $(V,\sigma)$ essentially given by the solution space of the field equation
\begin{align} \label{SymplecticVectorSpace}
 V:= \D(M,E) \slash \ker G \cong \ker P\big|_{\CCsc}, \qquad \sigma\big([\f],[\psi]\big) := (G\f,\psi)_M.
\end{align}

$\D(M,E)$ represents test sections in $E$, more precisely smooth sections with compact support, and $\CCsc(M,E)$ those with only spacelike compact support, meaning that it is contained in the causal future and past of some compact subset of $M$. The $L^2$-product $(\cdot,\cdot)_M$ of test sections is induced by the non-degenerate inner product on $E$.\\
For a bosonic quantum field theory, we take the $\CCR$-representation of $(V,\sigma)$, i.e.\ a pair $(w,A)$ consisting of a $C^*$-algebra $A$ and a map $w$ from $V$ into the unitary elements of $A$ such that $A$ is generated as a $C^*$-algebra by $\{w(x)\}_{x\in V}$ and the Weyl relations hold:
\begin{align} \label{WeylRelations}
 w(x)w(y) = e^{-\frac{i}{2}\sigma(x,y)} w(x+y), \qquad x,y\in V.
\end{align}
This construction goes back to \cite{M1968} (see also section 4.2 of \cite{BGP2007} and 5.2.2.2 of \cite{BR2002}) and it is unique in an appropriate sense. Therefore, altogether, $(M,E,P) \mapsto \CCR\big(\D(M,E) \slash \ker G\big)$ provides the desired functor and Haag-Kastler's axioms are satisfied (Theorem 3.10 of \cite{BG2011})

\subsection{States, quasifree states and Hadamard states}

We introduce states in the theory as normed and positive functionals $\tau$ on $\CCR(V)$, where $\tau(a)$ can be thought of as the expectation value of the observable $a$ in the state $\tau$. The induced GNS-representation $(\pi_\tau, \HR_\tau,\Omega_\tau)$ provides the familiar framework of a state space $\HR_\tau$ with observables as bounded operators $\pi_\tau(a)$ and a cyclic vector $\Omega_\tau$ (see section 2.3 of \cite{BR2002} for details). Hence, the selection of a distinct vacuum is shifted to that of an algebraic state $\tau$. Particularly adapted to free quantum fields are the so-called quasifree states generated by 
\begin{align*} 
 \tau\big(w(x)\big) = e^{-\frac{1}{2}\eta(x,x)}, \qquad x\in V,
\end{align*}

for some scalar product $\eta$ on $V$, by which $(\pi_\tau, \HR_\tau,\Omega_\tau)$ is determined up to unitary equivalence. For these states, the unitary operators $\big\{\pi_\tau\big(w(tx)\big)\big\}_{t\in\R}$ constitute a strongly continuous one-parameter group and thus, field operators $\Phi_\tau(x)$ are given by the self-adjoint generators due to Stone's theorem. Furthermore, there is a dense domain $D_\tau\subset \HR_\tau$ such that $\ran\Phi_\tau(x) \subset D_\tau \subset \dom\Phi_\tau(x)$ for all $x$, so polynomials of field operators are well-defined on $D_\tau$. Hence, the Weyl relations (\ref{WeylRelations}) imply the familiar canonical commutator relations $\big[\Phi_\tau(x),\Phi_\tau(y)\big] = i\sigma(x,y)\, \id_{\HR_\tau}$ and for all $n\in\N$, the $n$-point-function of the state $\tau_n(x_1,\hdots,x_n) := \Braket{\Phi_\tau(x_1)\hdots\Phi_\tau(x_n)\Omega_\tau}{\Omega_\tau}_{\HR_\tau}$ represents a well-defined distribution (see section 4.2 of \cite{BG2011} for precise definitions and proofs). In particular, the two-point-function is of the form
\begin{align} \label{TwoPointFunction}
 \tau_2(x,y) = \eta(x,y) + \frac{i}{2}\sigma(x,y), \qquad x,y\in V,
\end{align}

so it reproduces $\eta$ and hence $\tau$. Indeed, we have $\tau_n=0$ for odd $n$, and for $n$ even, it is given by some polynomial in the elements of $\{\tau_2(x_i,x_j)\}_{i,j=1,\hdots,n}$. Thinking of $\tau_n$ as the propagation of the state of the field, the focus on quasifree states corresponds to the perception that this propagation is essentially given by independent one-particle-propagations, which legitimates them as the natural objects to look at when dealing with free quantum fields (see chapter 17 of \cite{DG2013} for an overview of quasifree states). Going back from $(V,\sigma)$ to $(M,E,P)$, the scalar product $\eta$ corresponds to a bidistribution $S\colon  \D(M,E)\times\D(M,E) \rightarrow \R$ with
\begin{align} \label{PropertiesTwoPointFunction}
 S[P\psi_1,\psi_2] = 0 = S[\psi_1,P\psi_2], \qquad S[\psi_1,\psi_2] = S[\psi_2,\psi_1], \qquad S[\psi,\psi]\geq 0.
\end{align}

With regard to (\ref{SymplecticVectorSpace}) and (\ref{TwoPointFunction}), a quasifree state is therefore determined by $S$ and $G$.\\
Despite all physically motivated restrictions on the field so far, there is still a huge variety of possible states, so we need to look for constraints also on this level. A reasonable demand would be renormalizability of $\tau_2$, most prominently represented by the expectation value of the energy momentum tensor, since products of distributions are in general ill-defined. In flat quantum field theory, this would be carried out by subtracting the vacuum expectation value setting us back to the problem of non-existence of a distinct vacuum. On the other hand, this procedure merely requires regularity of differences of expectation values, and indeed, J. Hadamard's theory of second order hyperbolic equations \cite{H1923} led to a family of bisolutions with fixed singular part in the sense that the difference of any two such bisolutions is smooth (see \cite{W1994} and, more recently, \cite{H2016} for details, as well as \cite{DF2008} for the concrete renormalization). Accordingly, a state is called a Hadamard state if its two-point-function has the Hadamard singularity structure, which, by now, has been shown to be invariant under Cauchy evolution \cite{FSW1978}. Futhermore, any globally hyperbolic spacetime admits a large class of pure Hadamard states \cite{FNW1981, SV2001}.\\ 
However, the first mathematically precise definition of the Hadamard singularity structure has been specified only in \cite{KW1991}, in which the authors also show that for a wide class of spacetimes the Hadamard property singles out an invariant quasifree state. Moreover, in any spatially compact spacetime ("closed universes"), all Hadamard states, more specifically their GNS representations, comprise one unitary equivalence class, which, for general spacetimes, suggests a certain "local equivalence" of all possible notions of a vacuum state \cite{V1994}. In addition, Hadamard states yield finite fluctuations for all Wick polynomials \cite{BF2000}, which makes them relevant also for the perturbative construction of interacting fields (see also \cite{HW2015, R2016, D2019} and references therein). Consequently, Hadamard states are by now considered a reasonable counterpart of Minkowski finite energy states and the Hadamard condition an appropriate generalization of the energy condition for Minkowski quantum field theory. Note that the replacement of a distinct vacuum state by a whole class of states somehow reflects the essence of general relativity: Just like there is no preferred coordinate system, the concept of vacuum and particles as absolute quantities has to be re-evaluated and eventually downgraded to one choice among many. \\
It was Radzikowski who showed that for the massive scalar field the global Hadamard condition is equivalent to a certain requirement on the wave front set of the two-point-function \cite{R1992, R1996}, namely
\begin{align} \label{HConditionWF}
 \WF(\tau_2) = \big\{(p,-\xi;q,\zeta)\in \big(T^*M\times T^*M\big)\backslash\{0\}~\big|~(p,\xi)\sim(q,\zeta), ~ \xi \text{ is future-directed}\big\},
\end{align}

where 
\begin{align} 
 \begin{split} \label{LightlikeRelation}
  (p,\xi)\sim(q,\zeta) \qquad \Longleftrightarrow \qquad \begin{array}{cl} & \exists \text{ lightlike geodesic } c\colon I\rightarrow M \enspace \text{and} \enspace t,t'\in I\colon \\[2mm] & c(t)=p, \enspace c(t')=q, \enspace \dot{c}(t)=\xi^\sharp, \enspace \dot{c}(t')=\zeta^\sharp. \end{array}
 \end{split}
\end{align}

Note that, unlike the criterion given in \cite{KW1991}, this is a local condition, which Sahlmann and Verch generalized to sections in general vector bundles \cite{SV2001}. It includes Hadamard states of the Dirac field in the sense of \cite{K1995, K2000, H2001}, which have been used, for instance, for a mathematical rigorous description of the chiral anomaly \cite{BS2016}. In addition, \cite{SVW2002} proposed an even more elegant characterization of the Hadamard property in terms of Hilbert space valued distributions $\f\mapsto \Phi_\tau[\f]\Omega_\tau\in\HR_\tau$, involving the GNS-representation induced by $\tau$. Also for non-quasifree states, one can formulate (\ref{HConditionWF}) as a constraint on the whole $n$-point-function, which is compatible with the special case of quasifree states \cite{S2010}. Moreover, in analytic spacetimes, this generalized Hadamard condition can be sharpened to a condition on the analytic wave front set, thereby implying the Reeh-Schlieder-property \cite{SVW2002}. Likewise, for non-globally hyperbolic spacetimes, there is a formulation of the Hadamard condition via restriction to globally hyperbolic subregions. Hadamard states have therefore been studied in connection with the Casimir effect and on anti-de Sitter spacetime (see \cite{DNP2014, DFM2018} and references therein). By using the weaker concept of Sobolev wave front sets, a definition of adiabatic states on globally hyperbolic spacetimes similar to (\ref{HConditionWF}) is given in \cite{JS2002}, thus implying that Hadamard states are adiabatic.\\
However, most importantly for the purpose of this work, the Hadamard condition in the form (\ref{HConditionWF}) allows us to employ the techniques of microlocal analysis provided by Duistermaat and H{\" o}rmander \cite{DH1972}. Soon after Radzikowski's work, Junker derived pure Hadamard states for the massive scalar field on spatially compact globally hyperbolic spacetimes, using a factorization of the Klein-Gordon operator by pseudo-differential operators \cite{J1996,J2002}. G\'erard, Wrochna et al. generalized this construction to a large class of spacetimes \cite{GOW2017} and even gauge fields \cite{GW2015}. Furthermore, they proved the existence of (not necessarily pure) Hadamard states \cite{GW2014} in a much more concrete manner than \cite{FNW1981}. See \cite{G2019} for a recent review of these techniques. \\ 
On the other hand, there have been further proposals for physically reasonable states like the Sorkin-Johnston-states \cite{AAS2012}, which in general lack the Hadamard property \cite{FV2012}. Nevertheless, a modification of their construction produces Hadamard states \cite{BF2014}. For a contemporary synopsis concerning preferred vacuum states on general spacetimes, the nature of the Hadamard property and this construction in particular, see also \cite{F2018}. Apart from these rather general prescriptions, Hadamard states have been constructed explicitly for a large variety of spacetimes with special (asymptotic) symmetries, and furthermore, well-established states have been tested for the Hadamard property (see the introduction sections of \cite{GW2014, GOW2017} and the references therein as well as section 8.4 of \cite{FV2015} and 2.4 of \cite{H2016} for an overview). 

\subsection{This work}

In his seminal work \cite{R1996}, Radzikowski already realized that a bidistribution $\widetilde{H}$ satisfies the Hadamard condition if and only if it is of the form
\begin{align} \label{TwoPointFunctionHadState}
 \widetilde{H}=\frac{i}{2}\big(\widetilde{G}_{aF} - \widetilde{G}_{F} + \widetilde{G}_A - \widetilde{G}_R\big)
\end{align}

with $\widetilde{G}_{aF}, \widetilde{G}_{F}, \widetilde{G}_A, \widetilde{G}_R$ the distinguished parametrices in the sense of Duistermaat-H{\" o}rmander (Theorem 6.5.3 in \cite{DH1972}). With $\Delta':= \big\{(p,\xi;p,-\xi)\big\}$ the primed diagonal, they are characterized by
\begin{align}
 \begin{split} \label{WFDistParametrices}
  \WF(\widetilde{G}_A)    = \Delta' \cup \big\{(p,\xi)\sim(q,-\zeta), ~ q\in J_+(p)\big\}, & \quad  \WF(\widetilde{G}_R) = \Delta' \cup \big\{(p,\xi)\sim(q,-\zeta), ~ q\in J_-(p)\big\}, \\[2mm]
  \WF(\widetilde{G}_F)    = \Delta' \cup \big\{(p,\xi)\sim(q,\zeta), ~ t>t' \big\}, & \quad  \WF(\widetilde{G}_{aF}) = \Delta' \cup \big\{(p,\xi)\sim(q,\zeta), ~ t<t' \big\}.
 \end{split}
\end{align}

However, he remarked that it is not clear how to prove that one may choose the smooth part such that (\ref{PropertiesTwoPointFunction}) is exactly satisfied \cite{R1992}, an issue closely related to the question, whether $\widetilde{G}_{aF}, \widetilde{G}_{F}$ can be chosen as actual Green operators, which has been already addressed in section 6.6 of \cite{DH1972}. Clearly, microlocal analysis proved to be a powerful tool for the investigation of singularities and indispensable for the results listed in the former paragraph. Nevertheless, for these remaining questions, the non-singular part of $\widetilde{H}$ is primarily concerned, so different techniques are required.\\
This work is dedicated to provide such techniques and resolves the question for wave operators on globally hyperbolic spacetimes. It therefore gives a further and more constructive existence proof of, not necessarily pure, Hadamard states than \cite{FNW1981}. By avoiding any kind of deformation argument, it covers situations involving, for example, analytic spacetimes or constraint equations as in general relativity, where such an argument is usually not applicable.\\
The starting point in chapter 3 is the local construction of $\widetilde{G}_{aF}, \widetilde{G}_{F}$ in terms of Hadamard series very much inspired by \cite{BGP2007}. In chapter 4, these local objects are globalized and finally patched together by introducing a well-posed Cauchy problem for bisolutions and techniques from \v{C}ech cohomology theory. In the final chapter 5, we check existence of a "positive and symmetric choice", that is the smooth part of any Hadamard bisolution can be chosen such that it obtains the properties of a two-point-function. The necessary preparation is given in chapter 2 as well as a proof for the symmetry of the Hadamard coefficients in the vector valued case alternative to \cite{K2019} in the Appendix.

\section{Preliminaries}

In this section, we provide basic notations and prove certain theorems needed in the later constructions. For any $d$-dimensional vector space with non-degenerate inner product $\bbraket{\cdot}{\cdot}$ of index $1$, we adopt the notations and conventions of \cite{BGP2007}, that is, for instance, the signature $(-,+,\hdots,+)$ and the squared distance $\gamma(x):=-\bbraket{x}{x}$. The two connected components $I_\pm$ of the set of timelike vectors $I:=\{\gamma(x)>0\}$ then determine a time-orientation, where we define the elements of $I_+ (I_-)$ to be future (past) directed. Correspondingly, we set $C_\pm := \partial I_\pm,~ J_\pm := \overline{I_\pm}$, whose non-zero elements we call "lightlike" and "causal", respectively. Leaving out "$\pm$" means the union of both components, i.e. $I:= I_+\cup I_-$ and similarly $C$ and $J$. Non-causal vectors are referred to as "spacelike".\\
For $(M,g)$ a $d$-dimensional time-oriented Lorentzian manifold and $p\in M$, we write $I_\pm^M(p),C_\pm^M(p)$ and $J^M_\pm(p)$ for the corresponding chronological/lightlike/causal future/past of $p$. These sets comprise all points that can be reached from $p$ via some timelike/lightlike/causal future/past directed differentiable curve, that is a curve with tangent vectors of the respective type at each point. For subsets $A\subset M$, we define $I_\pm^M(A):=\bigcup_{p\in A}I^M_\pm(p)$ and similarly $J^M_\pm(A)$. For the definitions of different types of subsets of $M$ like future/past compact, geodesically starshaped, convex, causally compatible, causal, Cauchy hypersurface etc., we refer to section 1.3 of \cite{BGP2007}. In this work, Cauchy hypersurfaces of $M$ are always assumed to be spacelike. \\
For $E$ some real or complex finite-dimensional vector bundle over $M$, the spaces of $C^k$-, $\CC$-, $\D$-sections in $E$ as well as distributional sections $\D(M,E,W)'$ with values in some finite-dimensional space $W$, including their (singular) support, convergence, order etc., are defined as in section 1.1 of \cite{BGP2007}. For $\df V$ the volume density induced by the Lorentzian metric, $F$ another vector bundle over $M$ and $P\colon \CC(M,E)\rightarrow\CC(M,F)$ some linear differential operator, the formally transposed operator $P^t\colon\CC(M,F^*) \rightarrow \CC(M,E^*)$ of $P$ is given by 
$$(P^t\f)[\psi] := \f[P\psi] = \int_M \f\big(P\psi\big) \df V, \qquad \psi\in\D(M,E),~\f\in\D(M,F^*).$$ 
If $E$ is equipped with a non-degenerate inner product $\braket{\cdot}{\cdot}$, which induces the $L^2$-product $(\cdot,\cdot)_M$ and the isomorphism $\Theta\colon E \rightarrow E^*$, we call $P$ formally self-adjoint if $(P\psi_1,\psi_2)_M = (\psi_1,P\psi_2)_M$ for all $\psi_1,\psi_2$, that is, $P=\Theta^{-1}P^t\Theta$. Furthermore, $P$ is a wave operator if its principal symbol is given by $\xi\mapsto g(\xi^\sharp,\xi^\sharp) \cdot \id_E,$ on $T^*M$ implying that wave operators are of second order (see section 1.5 of \cite{BGP2007} for details).\\
L. Schwartz' celebrated kernel theorem establishes a one-to-one-correspondence between sequentially continuous operators $\KK\colon\D(M,E^*)\rightarrow\D(M,E^*)'$, i.e. $\KK\f_j\rightarrow \KK\f$ if $\f_j\rightarrow\f$, and bidistributions $K\colon\D(M,E) \times\D(M,E^*) \rightarrow\R$ given by $K[\psi,\f] = (\KK\f)[\psi]$ and called Schwartz kernel of $\KK$. It is represented by a distributional section in the bundle $E^*\boxtimes E$ over $M\times M$, whose fibers we identify via
\begin{align} \label{FibersBoxBundle}
 (E^*\boxtimes E)_{(p,q)} = E^*_p\otimes E_q\cong \Hom(E^*_q,E^*_p), \qquad (p,q)\in M\times M.
\end{align}

\begin{Def} \label{DefParametrices}
 Let $P\colon\CC(M,E)\rightarrow\CC(M,E)$ be a linear differential operator. A linear and sequentially continuous operator $\QQ\colon \D(M,E^*)\rightarrow\CC(M,E^*)$ is called
 \begin{itemize}
  \item left parametrix for $P^t$ if $\QQ P^t\big|_\D-\id$ is smoothing,
  \item right parametrix for $P^t$ if $P^t\QQ-\id$ is smoothing,
  \item two-sided parametrix or just parametrix for $P^t$ if $\QQ$ is left and right parametrix for $P^t$,
  \item Green operator for $P^t$ if $\QQ P^t\big|_\D=P^t\QQ=\id$.
 \end{itemize}
 A bidistribution $Q\colon\D(M,E)\times\D(M,E^*)\rightarrow\R$ with $p\mapsto Q(p)[\f]\in\CC(M,E^*)$ for all $\f$ is called
 \begin{itemize}
  \item parametrix for $P$ at $p\in M$ if $P_{(2)}Q(p) - \delta_p \in \CC(M,E)$,
  \item fundamental solution for $P$ at $p\in M$ if $P_{(2)}Q(p) = \delta_p$.
 \end{itemize}
\end{Def}

Note that $\QQ$ is a parametrix for $P^t$ if and only if its Schwartz kernel is a parametrix for $P^t$ at all $p\in M$. Especially, it is a Green operator for $P^t$ if and only if $Q(p)$ is a fundamental solution for $P$ and $P^t_{(1)}\big(Q(\cdot)[\f]\big) = \f$ for all $p$ and $\f$.

\subsection{Cauchy problems}

A crucial feature of globally hyperbolic Lorentzian manifolds $M$ is a well-posed Cauchy problem for wave operators on smooth sections, i.e. for any Cauchy hypersurface $\Sigma$ with normal field $\nu$, which is timelike, and Cauchy data $f\in\CC(M,E), u_0,u_1 \in \CC(\Sigma,E)$, the Cauchy problem 
\begin{align} \label{CauchyProblem}
 \left\{\begin{array}{cl} Pu & = f, \\[2mm] u\big|_\Sigma & = u_0,  \\[2mm] \nabla_\nu u\big|_\Sigma & = u_1, \end{array} \right.
\end{align}

has a unique solution $u\in\CC(M,E)$, which is supported in $J\big(\supp u_0 \cup \supp u_1 \cup \supp f\big)$ and depends continuously on the data (see section 3.2 of \cite{BGP2007} and chapter 3 of \cite{BF2009}). For $F$ a vector bundle over some further globally hyperbolic Lorentzian manifold $N$, recall (\ref{FibersBoxBundle}) for the definition of the vector bundle $E\boxtimes F$ over $M\times N$. 

\begin{Thm} \label{CauchyProblemBisection}
 Let $M,N$ be globally hyperbolic Lorentzian manifolds with Cauchy hypersurfaces $\Sigma,\Xi$ and unit normal fields $\mu,\nu$. Furthermore, let $P,Q$ denote linear differential operators of second order acting on smooth sections in vector bundles $E,F$ over $M,N$, which admit well-posed Cauchy problems and only lightlike characteristic directions. Then, for all $u_i\in\CC(\Sigma\times\Xi,E\boxtimes F),~i=1,...,4,$ and \mbox{$f,g\in\CC(M\times N,E\boxtimes F)$} with $Qf=Pg$, there is a unique section $u\in\CC(M\times N,E\boxtimes F)$ solving
 \begin{align} \label{CauchyProblemBisolution}
  \left\{\begin{array}{cl} Pu & = f, \\[2mm]
                           Qu & = g, \\[2mm]
                           u\big|_{\Sigma\times\Xi} & = u_1, \\[2mm]
                           \nabla_\mu u\big|_{\Sigma\times\Xi} & = u_2, \\[2mm]
                           \nabla_\nu u\big|_{\Sigma\times\Xi} & = u_3, \\[2mm]
                           \nabla_\nu\nabla_\mu u\big|_{\Sigma\times\Xi} & = u_4. \end{array}\right.
 \end{align}
 Let $Z:=\big(\oplus^2\CC(M\times N,E\boxtimes F)\big)\oplus\big(\oplus^4\CC(\Sigma\times\Xi,E\boxtimes F)\big)$ and $X$ denote the subset of elements $(f,g,u_1,u_2,u_3,u_4)$ satisfying $Qf=Pg$. Then the map $(f,g,u_1,u_2,u_3,u_4) \longmapsto u$, which sends the Cauchy data to the unique solution $u$ of (\ref{CauchyProblemBisolution}), is a linear and continuous operator $X \longrightarrow \CC(M\times N,E\boxtimes F)$.
\end{Thm}

\begin{proof}
 For all $q\in N$ and $h_0, h_1\in\CC(\Sigma\times N,E\boxtimes F)$, the Cauchy problem 
 \begin{align}\label{CPu1}
  \left\{\begin{array}{cl} Pu_q & = f(\cdot,q), \\[2mm] u_q\big|_\Sigma & = h_0(\cdot,q), \\[2mm] \nabla_\mu u_q\big|_\Sigma & = h_1(\cdot,q), \end{array} \right.
 \end{align}
 
 has a unique solution $u_q\in\CC(M,E\otimes F_q)$ depending smoothly on the data by well-posedness of (\ref{CauchyProblem}). Thus, it remains to determine $h_0, h_1$ from $u_1,u_2,u_3,u_4,g$ and to show that then $Qu=g$ is automatically fulfilled. For all \mbox{$\sigma\in\Sigma,\xi\in\Xi$}, we define smooth sections $h_0(\sigma,\cdot), h_1(\sigma,\cdot) \in \CC(N,E_\sigma\otimes F)$ and $h_0(\cdot,\xi),h_2(\cdot,\xi) \in \CC(M,E\otimes F_\xi)$ as solutions of
 \begin{align} 
  \begin{split} \label{CPInitialData}
   \left\{\begin{array}{cl} Ph_0(\cdot,\xi) & = f(\cdot,\xi), \\[2mm] h_0(\cdot,\xi)\big|_\Sigma & = u_1(\cdot,\xi), \\[2mm] \nabla_\mu h_0(\cdot,\xi)\big|_\Sigma & = u_2(\cdot,\xi), \end{array} \right.
   \qquad & \qquad 
   \left\{\begin{array}{cl} Ph_2(\cdot,\xi) & = \big(\nabla_\nu f\big)(\cdot,\xi), \\[2mm] h_2(\cdot,\xi)\big|_\Sigma & = u_3(\cdot,\xi), \\[2mm] \nabla_\mu h_2(\cdot,\xi)\big|_\Sigma & = u_4(\cdot,\xi), \end{array} \right. \\[2mm]
   \left\{\begin{array}{cl} Qh_0(\sigma,\cdot) & = g(\sigma,\cdot), \\[2mm] h_0(\sigma,\cdot)\big|_\Xi & = u_1(\sigma,\cdot), \\[2mm] \nabla_\nu h_0(\sigma,\cdot)\big|_\Xi & = u_3(\sigma,\cdot), \end{array} \right.
   \qquad & \qquad
   \left\{\begin{array}{cl} Qh_1(\sigma,\cdot) & = \big(\nabla_\mu g\big)(\sigma,\cdot), \\[2mm] h_1(\sigma,\cdot)\big|_\Xi & = u_2(\sigma,\cdot), \\[2mm] \nabla_\nu h_1(\sigma,\cdot)\big|_\Xi & = u_4(\sigma,\cdot). \end{array} \right.
  \end{split}
 \end{align}
 
 By adapting the proof of Proposition A.1 of \cite{FNW1981}, we obtain smooth sections $h_0,h_1,h_2$ in $E\boxtimes F$ over \mbox{$(M\times\Xi) \cup (\Sigma \times N), \Sigma \times N$} and $M\times\Xi$, respectively, and, following the same lines, $u(\cdot,q) := u_q$ depends smoothly on $q$. Hence, we found some $u\in\CC(M\times N,E\boxtimes F)$ solving (\ref{CPu1}), which yields the initial data of a solution of (\ref{CauchyProblemBisolution}):
 \begin{align*}
  u\big|_{\Sigma\times\Xi} = h_0\big|_{\Sigma\times\Xi} = u_1, \qquad & \nabla_\mu u\big|_{\Sigma\times\Xi} = h_1\big|_{\Sigma\times\Xi} = u_2, \\[2mm]
  \nabla_\nu u\big|_{\Sigma\times\Xi} = \nabla_\nu h_0\big|_{\Sigma\times\Xi} = u_3, \qquad & \nabla_\nu\nabla_\mu u\big|_{\Sigma\times\Xi} = \nabla_\nu h_1\big|_{\Sigma\times\Xi} = u_4.
 \end{align*} 
 Note that $P$ and $Q$ commute because they act on different factors of $M\times N$. Therefore, (\ref{CPu1}) and (\ref{CPInitialData}) imply that $Qu$ and $g$ satisfy the same Cauchy problem:
 $$\left\{\begin{array}{cl} & PQu = QPu = Qf = Pg, \\[2mm] & Qu\big|_{\Sigma\times N} = Qh_0 = g\big|_{\Sigma\times N}, \\[2mm] & \nabla_\mu Qu\big|_{\Sigma\times N} = Q \nabla_\mu u\big|_{\Sigma\times N} = Qh_1 = \nabla_\mu g\big|_{\Sigma\times N}, \end{array} \right.$$
 that is $Qu=g$ due to well-posedness. Clearly, this solution $u$ is unique since trivial Cauchy data in (\ref{CauchyProblemBisolution}) lead to trivial data in (\ref{CPInitialData}) and therefore in (\ref{CPu1}), which implies $u_q=0$ for all $q$ and hence $u=0$. The proof of stability follows the same lines as for (\ref{CauchyProblem}), that is Theorem 3.2.12 of \cite{BGP2007}: Obviously, the map
 \begin{align*}
  \Phi\colon \qquad \CC(M\times N,E\boxtimes F) & \longrightarrow Z \\[2mm]
                                         u & \longmapsto \big(Pu,Qu,u\big|_{\Sigma\times\Xi},\nabla_\mu u\big|_{\Sigma\times\Xi}, \nabla_\nu u\big|_{\Sigma\times\Xi}, \nabla_\nu\nabla_\mu u\big|_{\Sigma\times\Xi}\big)
 \end{align*}%
 is linear, injective and continuous, and since (\ref{CauchyProblemBisolution}) admits a solution $u$ for each data in $X$, the closed subset $X\subset Z$ is contained in $\ran \Phi$. Due to continuity of differential operators, the subspace $\Phi^{-1}(X)\subset \CC(M\times N, E\boxtimes F)$ is also closed, so we obtain a continuous and bijective map $\Phi\colon \Phi^{-1}(X) \rightarrow X$ between Fr\'echet spaces, whose inverse is continuous by the open mapping theorem.
\end{proof}

It seems that this procedure directly generalizes to operators of order $k$, for which derivatives up to order $k^2$ have to be provided as data. Furthermore, symmetry of the data is inherited by the solution:

\begin{Cor} \label{SymmetricBisolution}
 With regard to the assumptions of Theorem \ref{CauchyProblemBisection}, let $(N,\Xi,\nu)=(M, \Sigma, \mu),~F=E^*$ and $E$ be equipped with a non-degenerate inner product. Let $Q=P^t$, and for all $(p,q)\in M\times M$, assume
 \begin{align} \label{SymmetryInitialData}
  \begin{split}
   f(p,q) = \Theta_p^{-1} g(q,p)^t \Theta_q, \qquad u_2(\sigma_1,\sigma_2) = \Theta^{-1}_{\sigma_1} u_3(\sigma_2,\sigma_1)^t \Theta_{\sigma_2},\\[2mm]
   u_1(\sigma_1,\sigma_2) = \Theta^{-1}_{\sigma_1} u_1(\sigma_2,\sigma_1)^t \Theta_{\sigma_2}, \qquad u_4(\sigma_1,\sigma_2) = \Theta^{-1}_{\sigma_1} u_4(\sigma_2,\sigma_1)^t \Theta_{\sigma_2}
  \end{split}
 \end{align}
 with fiberwise transposition $~^t\colon \Hom(E_q, E_p)\rightarrow \Hom(E_p^*, E_q^*)$. Then the solution of (\ref{CauchyProblemBisolution}) satisfies 
 \begin{align} \label{SymmetrySolution}
  u(p,q)=\Theta_p^{-1} u(q,p)^t \Theta_q, \qquad p,q\in M.
 \end{align}
\end{Cor}

Additionally, we investigate the propagation of a family of singular solutions from a neighborhood of $\Sigma$ to all of $M$ by applying the well-posed Cauchy problem for singular sections treated in \cite{BTW2015}. For that, we just have to ensure the existence of the restriction to $\Sigma$ by checking H{\" o}rmander's criterion.

\begin{Thm} \label{CauchyProblemSingularSolution}
 Let $M$ be a globally hyperbolic Lorentzian manifold, \mbox{$\Sigma\subset M$} a Cauchy hypersurface, $\pi\colon E\rightarrow M$ a real or complex vector bundle over $M$ and $P\colon \CC(M,E)\rightarrow \CC(M,E)$ a wave operator. Furthermore, let $O\subset M$ be relatively compact, and for all $p\in O$, let \mbox{$v(p)\in\D(M,E,E_p^*)'$} have spacelike compact support and only lightlike singular directions. Moreover, we assume \mbox{$p\mapsto v(p)[\f]\in\CC(M,E^*)$} for fixed $\f\in\D(M,E^*)$. Then the Cauchy problem
 \begin{align*} 
  \left\{\begin{array}{cl} Pu(p) & = 0, \\[2mm] u(p)\big|_\Sigma & = v(p)\big|_\Sigma, \\[2mm] \nabla_\nu u(p)\big|_\Sigma & = \nabla_\nu v(p)\big|_\Sigma, \end{array} \right.
 \end{align*}
 has a unique solution $u(p)\in\D(M,E,E^*_p)'$, which has spacelike compact support and provides a smooth section $p\mapsto u(p)[\f]$ for each \mbox{$\f\in\D(M,E^*)$}.
\end{Thm}

\begin{proof}
 Let $t\colon M \rightarrow \R$ be a Cauchy time function on $M$ such that $\Sigma=t^{-1}(0)$ (Theorem 1.3.13 of \cite{BGP2007}). Therefore, the normal directions of $\Sigma$ are timelike and do not match the singular directions of $v$, so $v(p)\big|_\Sigma$ and $\nabla_\nu v(p)\big|_\Sigma$ are well-defined and compactly supported distributions on $\Sigma$ for all $p$ due to H{\" o}rmander's criterion \big((8.2.3) of \cite{H1990}\big). Recall that any compactly supported distribution lies in some Sobolev space $H_c^k$ (see e.g.\ (31.6) of \cite{T1967}), and hence, $v(p)\big|_\Sigma\in H_c^k(\Sigma,E_p^*\otimes E)$ and $\nabla_\nu v(p)\big|_\Sigma\in H^{k-1}_c(\Sigma,E_p^*\otimes E)$ for some $k\in\R$. Thus, for all $p$, Corollary 14 of \cite{BTW2015} provides a unique solution 
 $$u(p)\in C_{sc}^0\big(t(M),H^k(\Sigma_\cdot);E_p^*\otimes E\big) \cap C_{sc}^1\big(t(M),H^{k-1}(\Sigma_\cdot);E_p^*\otimes E\big),$$ 
 where this intersection is commonly referred to as the space of finite $k$-energy sections (see section 1.7 of \cite{BTW2015} for details about them). Moreover, the mapping of initial data to the solution is a linear homeomorphism, i.e. continuity of the restriction $v(p)\mapsto \big(v(p)\big|_\Sigma, \nabla_\nu v(p)\big|_\Sigma\big)$ implies continuity of the map of distributions $T\colon v(p)\mapsto u(p)$ for all $p$. \\ 
 For $D$ a differential operator, let $\big(D_{(1)}v\big)(p)$ denote the distribution $\f\mapsto \big(D(v(\cdot)[\f])\big)(p)$. Since $P$ and $D$ act on different factors, $\big(D_{(1)}v\big)(p)$ is linearly and continuously mapped to $\big(D_{(1)}u\big)(p)$, that is, $T$ commutes with $D_{(1)}$ (see the proof of Proposition A.1 in \cite{FNW1981}). In particular, the map 
 $$p\longmapsto\big(D_{(1)}v\big)(p)[\f] \longmapsto \big(D_{(1)}u\big)(p)[\f] = \big(D(u(\cdot)[\f])\big)(p), \qquad \f\in\D(M,E^*),$$ 
 is continuous due to smoothness of the first arrow. This holds for all differential operators $D$, which provides smoothness of $p\mapsto u(p)[\f]$ for fixed $\f$.
\end{proof}

\subsection{The prototype}

The Hadamard condition in the form (\ref{HConditionWF}) is a local condition and, moreover, the singularity structure of a bisolution is related to the corresponding differential operator essentially via its principal symbol. On these grounds, we start with the prototype setting $P=\Box$ on $M=\R^d_{\mathrm{Mink}},~d\geq 3$, since, from the viewpoint of the singularity structure of the solutions, this already incorporates the characteristic properties of the solutions for the general setting of wave operators on curved spacetimes. In order to obtain a decomposition like (\ref{TwoPointFunctionHadState}), we study fundamental solutions for $\Box$. It is not hard to check that these objects are invariant under the special orthochronous Lorentz group $\sLL$ and their singular support is cointained in the light cone $C$. This makes them directly comparable outside $C$ meaning that we need to understand $\sLL$-invariant distributions supported on $C$. \\
For $\RP^d := \R^d\backslash\{0\}$, we consider the submersions $\gamma_\pm := \gamma\big|_{J_\pm^c}$, where $J_\pm^c := \RP^d\backslash J_\pm$. Then the pullback $\gamma^*_\pm$ maps distributions on $\R$ to $\sLL$-invariant distributions on $J_\pm^c$ and moreover establishes a close connection to the well-known classification of distributions on $\R$ supported in $\{0\}$:

\begin{Thm}[Th\'eor\`eme 2 of \cite{M1954}] \label{T&Tpm}
 For any pair $T_\pm \in \D(\R)'$ with $T_+|_{\R_{<0}} = T_-|_{\R_{<0}}$, there is a \mbox{$\sLL$-invariant} distribution $T \in \D(\RP^d)'$ given by $T|_{J_\mp^c} := \gamma_\pm^*T_\pm$. Conversely, for any $\sLL$-invariant $T \in \D(\RP^d)'$, we find a pair $T_\pm \in \D(\R)'$ with $T_+|_{\R_{<0}} = T_-|_{\R_{<0}}$ such that $T|_{J_\mp^c} = \gamma_\pm^*T_\pm$.
\end{Thm}

For that, it is important to note that the proof of Theorem \ref{T&Tpm} particularly demonstrates surjectivity of $(\gamma_\pm)_*$.

\begin{Thm}[Th\'eor\`eme 1 of \cite{M1954}] \label{LorentzInvariantPointDistribution}
 Any $\sLL$-invariant $T\in\D(\R^d)'$ with $\supp(T) \subset \{0\}$ is of the form \mbox{$T=\sum_{k=0}^\infty b_k \cdot \Box^k\delta_0$} with $b_k\neq0$ for only finitely many $k$.
\end{Thm}

There are two immediate consequences: First, every $\sLL$-invariant distribution $T$ supported in $C$ has to be of the form
\begin{align} \label{LorentzInvariantDistribution}
 T = \sum_{k=0}^\infty \left(a^+_k\cdot \gamma_+^*\delta_0^{(k)} + a^-_k\cdot \gamma_-^*\delta_0^{(k)} + b_k\cdot\Box^k\delta_0\right),
\end{align}

where only finitely many coefficients $a_k^\pm,b_k$ are non-zero. Recall that $\gamma_\pm^*\delta_0^{(k)}[\f] = \big((\gamma_\pm)_*\f\big)^{(k)}(0)$ and the push-forward along a submersion is given by integration along the fibers. This leads to the second consequence, which is that every $\sLL$-invariant measure supported on $C_\pm$ is of the form $a\df\Omega^0_\pm + b\delta_0$ (section IX.8 of \cite{RS1975}), where \vspace{-5mm}

\begin{align} \label{OmegaMeasure}
 \df\Omega^0_\pm[\f] = \int_{\R^{d-1}} \frac{\f\big(\pm\|\hat{x}\|,\hat{x}\big)}{2\|\hat{x}\|} ~ \df\hat{x}, \qquad \f\in\D(\R^d).
\end{align}

The most known examples of $\sLL$-invariant supported on $C_\pm$ are certain Riesz distributions (not all of them) and in fact, there are no other: For all $\alpha\in\C$ with $\RT{\alpha}>d$, they are defined as continuous functions via
\begin{align} \label{RieszDef}
 R_\pm^\alpha(x) = \left\{ \begin{array}{cl} 2C(\alpha,d)\cdot \gamma(x)^{\frac{\alpha-d}{2}}, \qquad & x \in J_\pm, \\[2mm] 0, & \text{otherwise}, \end{array} \right., \qquad C(\alpha,d) = \frac{2^{-\alpha}\pi^{\frac{2-d}{2}}}{\Gamma\left(\frac{\alpha}{2}\right) \Gamma\left(\frac{\alpha-d}{2}+1\right)}.
\end{align}

One directly checks holomorphicity in $\alpha$ and calculates $\Box R^{\alpha+2}_\pm = R^\alpha_\pm$, which provides an extension as distributions to all of $\C$ via $R^\alpha_\pm := \Box^k R^{\alpha+2k}_\pm, ~k>\frac{d}{2}-\RT{\alpha}$. Moreover, as shown in section 13.2 of \cite{DK2010}, we have $R^0_\pm=\delta_0$ and \vspace{-3mm}
\begin{align} \label{RieszDelta}
 \gamma_\pm^*\delta_0^{(k)} = \frac{\Gamma\left(\frac{d}{2}-(k+1)\right)}{2^{2k+3-d}\pi^{\frac{d-2}{2}}} \cdot R^{d-2(k+1)}_\pm\Big|_{J_\mp^c}, \qquad k\in\N_0.
\end{align}

Therefore, due to (\ref{LorentzInvariantDistribution}), every $\sLL$-invariant distribution $T$ supported on $C$ is given by a linear combination of Riesz distributions:
\begin{align} \label{LorentzInvariantDistributionRiesz}
 T = \sum_{k=0}^\infty \left(\lambda_k^+\cdot R_+^{d-2(k+1)} + \lambda_k^-\cdot R_-^{d-2(k+1)} + \beta_k\cdot R_\pm^{-2k}\right).
\end{align}

The properties of the Riesz distributions combined with the uniqueness of $G_A,G_R$ reveal $R_\pm^2$ as the advanced and retarded fundamental solution for $\Box$ on $\R^d_{\mathrm{Mink}}$. In order to identify $G_F,G_{aF}$, we introduce the functions $(\gamma\pm i0)^\alpha$ for $\RT{\alpha}>0$, where $\pm i0$ stands for the direction from which the branch cut along $\R_{<0}$ is approached. A meromorphic extension to all of $\C$ as distributions is given by
\begin{align} \label{gammapmi0}
 (\gamma\pm i0)^\alpha = \frac{\Box^k (\gamma\pm i0)^{\alpha+k}}{4^k \prod_{j=1}^k (\alpha+j)\big(\alpha+j+\frac{d-2}{2}\big)},
\end{align}

which is in fact holomorphic on all of $\big\{\RT{\alpha}>-\frac{d}{2}\big\}$, and we calculate $\Res{\alpha}{-\frac{d}{2}} (\gamma\pm i0)^\alpha = (\mp i)^{d-1} \frac{\pi^{\frac{d}{2}}}{\Gamma\left(\frac{d}{2}\right)} \cdot \delta_0$ (see section III.3 of \cite{GS1967} for proofs). This already ensures all crucial properties:

\begin{Prop} \label{GammaPMi0Family} 
 The distributions (\ref{gammapmi0}) are symmetric and $\sLL$-invariant. Moreover, for all natural numbers $m<\frac{d}{2}$, we have
 \begin{align} \label{GammaLog}
  (\gamma\pm i0)^{-m} = \frac{(-1)^{m-1} ~ \Gamma\left(\frac{d}{2}-m\right)}{4^m ~ \Gamma(m) \Gamma\left(\frac{d}{2}\right)}\cdot\Box^m \log(\gamma\pm i0),
 \end{align}
 
 and for $d>2$, fundamental solutions for $\Box$ are given by
 \begin{align} \label{Spm} 
  S_\pm := (\pm i)^{d+1}\frac{\Gamma\left(\frac{d-2}{2}\right)}{4\pi^{\frac{d}{2}}} \cdot \left(\gamma \pm i0\right)^{\frac{2-d}{2}}.
 \end{align}
\end{Prop}

\begin{samepage}
\begin{proof}
 Symmetry and $\sLL$-invariance is obvious for $\RT{\alpha}>0$ and thus on all of $\C$, it follows from meromorphicity and the identity theorem. Holomorphicity on $\big\{\RT{\alpha}>-\frac{d}{2}\big\}$ ensures $(\gamma\pm i0)^{-m} = \lm{\alpha}{-m} (\gamma\pm i0)^\alpha$, which leads to
 \begin{align*}
   (\gamma\pm i0)^{-m} = \frac{\Box^m\tds{}{\alpha}{\alpha=-m}(\gamma\pm i0)^{\alpha+m}}{4^m \left(\frac{d}{2}-m\right) \hdots \frac{d-2}{2} ~ (-1)^{m-1}\prod^{m-1}_{j=1}(m-j)} = \frac{(-1)^{m-1} ~ \Gamma\left(\frac{d}{2}-m\right)}{4^m ~ \Gamma(m) \Gamma\left(\frac{d}{2}\right)}\cdot\Box^m \log(\gamma\pm i0).
 \end{align*}
 Particularly at $\alpha=\frac{2-d}{2}$, it implies
 \begin{gather*}
  \Box\left(\gamma \pm i0\right)^{\frac{2-d}{2}} = \lm{\alpha}{-\frac{d}{2}}\Box\left(\gamma \pm i0\right)^{\alpha+1} = 4\cdot\frac{2-d}{2} \underbrace{\lm{\alpha}{-\frac{d}{2}} \left(\alpha+\frac{d}{2}\right) \left(\gamma \pm i0\right)^\alpha}_{=\Res{\alpha}{-\frac{d}{2}} (\gamma \pm i0)^\alpha} = (\mp i)^{d+1} \frac{4\pi^{\frac{d}{2}}}{\Gamma\left(\frac{d-2}{2}\right)} \cdot \delta_0. \qedhere
 \end{gather*}%
\end{proof}
\end{samepage}

\pagebreak

Let us now come to the two-point-function in the prototype case, which is characterized by Wightman's axioms and d'Alembert's equation. In particular, the Hadamard condition corresponds to the spectral condition given by some constraint on the Fourier transform. Due to the symmetry conditions implied by these axioms, the two-point-function is completely determined by some $\sLL$-invariant distribution $W$ on $\RP^d$, which solves $\Box W=0$. The spectral condition then leads to the general form $W=a\df\widecheck{\Omega}_0^+,~a\in\R,$ (\cite{RS1975, S2000}) and we choose $a=(2\pi)^{\frac{2-d}{2}}$. Here $\widehat{f}$ denotes the Fourier transformation with $\bbraket{\cdot}{\cdot}$ used in the exponential, i.e. ordinary Fourier transformation in space and inverse Fourier transformation in time direction $e_0$, and $\widecheck{f}$ the corresponding inverse operation. The maps
\begin{align} \label{DeltaAnalyt}
 \Delta^\pm: \qquad \T_\pm \longrightarrow \C, \qquad z \longmapsto \pm \frac{i}{(2\pi)^{d-1}} \int_{C_\mp} e^{i\bbraket{z}{p}} \df\Omega^\mp_0(p)
\end{align}

are analytic in the complex forward/backward tube $\T_\pm := \left\{z\in\C^d \, \big| \, \IT{z}\in I_\pm\right\}$, they are $\sLL$-invariant and related via $\Delta^+(z) = -\Delta^-(-z),~z\in T_+$. Hence, for each $z\in \T_\pm$, we obtain $\Delta^\pm(z)=\Delta^\pm(x\pm i\e e_0)=:\Delta_\e(x)$ for some $x\in\R^d,~\e>0$, and one directly calculates $W=i\lm{\e}{0}\Delta_\e^-$ in the distributional sense, that is
\begin{align} \label{IntFormelW}
 W(x) = \frac{1}{(2\pi)^{d-1}}\lm{\e}{0}\int_{C_+} e^{i\bbraket{x-i\e e_0}{p}} \df\Omega^+_0(p).
\end{align}

Due to the inverse Cauchy Schwarz inequality, $\gamma$ maps $\T_\pm$ to $\C\backslash\R_{\geq0}$, and hence $\sqrt{\gamma(z)}\in\C\backslash\R$ for all $z\in \T_\pm$. Let $\sigma\colon \T_\pm \rightarrow \mathbb{H}_\pm:=\{\pm\IT{z}>0\}$ denote the square root of $\gamma$ with appropriately chosen sign such that again $\sLL$-invariance yields $\Delta^\pm(z) = \Delta^\pm\big(\sigma(z)e_0\big), ~ z\in \T_\pm$. In particular, $z=x\pm i\e e_0$ is mapped to
\begin{align} \label{GammaEpsilon}
 \sigma(x\pm i\e e_0) = \sgn(x_0) \sqrt{\gamma_\e^\pm(x)}, \qquad\qquad \gamma_\e^\pm(x) & := \gamma(x \pm i\e e_0) = \gamma(x) - \e^2 \pm 2i\e x_0,
\end{align}

which lets us evaluate the integrals $\Delta^\pm_\e$ and thus derive a local expression for $W$:
 \begin{align} \label{Deltapm}
  \Delta^\pm_\e = \pm\frac{i\Gamma\left(\frac{d-2}{2}\right)}{4\pi^{\frac{d}{2}}} \cdot \big(-\gamma^\pm_\e\big)^{\frac{2-d}{2}} \qquad \Longrightarrow \qquad W = \frac{\Gamma\left(\frac{d-2}{2}\right)}{4\pi^{\frac{d}{2}}} \cdot \lm{\e}{0}\big(-\gamma^-_\e\big)^{\frac{2-d}{2}}.
 \end{align}

It is a well-known procedure to extract the advanced and retarded fundamental solution $\Delta^A$ and $\Delta^R$: \\
We extract the causal propagator $\Delta^C := -2iW^a$ from the antisymmetric part of $W$, which turns out to be supported only in $J = \big\{\gamma(x)\geq 0\big\}$ and is a solution of d'Alembert's equation as well since $W^a = \Delta^+ + \Delta^-$. H\"ormander's criterion allows us to multiply $\Delta^C$ with the step function $H_0$ with respect to the time coordinate and integration by parts reveals $\Box\big(H_0\cdot\Delta^C\big)=-\delta_0$. Therefore, the aforesaid fundamental solutions are given by
\begin{align} \label{AvRetPropBox}
 \Delta^A := \left(1-H_0\right) \cdot \Delta^C, \qquad \Delta^R := -H_0 \cdot \Delta^C.
\end{align}

Due to their support properties, they are unique, so we directly conclude $\Delta^A = R^2_-$ and $\Delta^R = R^2_+$.\\
Following Radzikowski's results, the remaining two fundamental solutions are given by 
\begin{align} \label{ExtractFaF}
  \Delta^{F}  = iW + \Delta^A, \qquad \qquad \Delta^{aF} = -iW + \Delta^R,
\end{align}

for which one directly calculates $\Delta^F=S_-$ and $\Delta^{aF}=S_+$ outside of $C$. It follows that the distributions $\Delta^F-S_-$ and $\Delta^{aF}-S_+$ are symmetric and $\sLL$-invariant solutions of d'Alembert's equation. Moreover, their support is contained in $C$ and they are homogeneous of degree $2-d$, so they have to vanish also on the light cone according to our prior characterization (\ref{LorentzInvariantDistributionRiesz}). Combining both equations (\ref{ExtractFaF}) leads to the following form of the Wightman distribution, that is the Hadamard two-point-function in the prototype case:
\begin{align} \label{HadamardPrototype}
 W = \frac{i}{2}\big(S_+ - S_- + R^2_- - R^2_+\big).
\end{align}

\section{Local Construction}

\subsection{Families of Riesz-like distributions}

We proceed with the local construction of Hadamard bidistributions in a setting $(M,E,P)$ as in Definition \ref{GlobHypGreen} with $P$ a wave operator. Inspired by the local construction of the advanced and retarded parametrices for $P$ in \cite{G1988, BGP2007}, we introduce families of distributions similar to the Riesz distributions (\ref{RieszDef}) but containing $S_\pm$ instead, so for $\RT{\alpha}>0$, we introduce the distributions
\begin{align} \label{DefL}
 L^\alpha_\pm := C(\alpha,d) \cdot (\gamma \pm i0)^{\frac{\alpha-d}{2}}, \qquad\qquad C(\alpha,d) := \frac{2^{-\alpha} \pi^{\frac{2-d}{2}}}{\Gamma\left(\frac{\alpha}{2}\right)\Gamma\left(\frac{\alpha-d}{2}+1\right)}.
\end{align}

Very similar to $R^\alpha_\pm$, we have $\Box L^{\alpha+2}_\pm=L^\alpha_\pm$ leading to a holomorphic extension to all of $\C$ as distributions, for which then, moreover, analogous relations hold:

\begin{Prop} \label{PropertiesL}
 For all $\alpha\in\C$, we have
 \begin{itemize}
  \item[(1)] $\gamma \cdot L^\alpha_\pm = \alpha(\alpha-d+2) L^{\alpha+2}_\pm$,
  \item[(2)] $\grad \gamma \cdot L^\alpha_\pm = 2\alpha ~ \grad L^{\alpha+2}_\pm$,
  \item[(3)] $\Box L^{\alpha+2}_\pm = L_\pm^\alpha$,
  \item[(4)] $L^{d-2n}_\pm=0, \enspace n\in \N$,
  \item[(5)] $L^{d+2n}_+ = L^{d+2n}_-, \enspace n\in\N_0$,
  \item[(6)] if $\RT{\alpha}>0$, then $L^\alpha_\pm$ are distributions of order at most $\kappa_d:=2\cdot \left\lceil\frac{d}{2}\right\rceil$.
 \end{itemize}
\end{Prop}

\begin{proof}
 The proofs for (1) - (3) and (6) are similar to Proposition 1.2.4 of \cite{BGP2007}.
 (4): Due to (3), integration by parts yields
 $$L^{d-2n}_\pm[\f] = L^d_\pm\big[\Box^n\f\big] = C(d,d) \int_{\R^d}\big(\Box^n\f\big)(x)\df x =0, \qquad \f \in \D(\R^d),~n\in\N.$$
 (5): Follows from $(\gamma\pm i0)^n=\gamma^n$ for all $n\in\N_0$.
\end{proof}

\vspace{-2mm} With regard to (5), we omit the "$\pm$" if $\alpha=d+2n$. A crucial property of the Riesz distributions is $R_\pm^0=\delta_0$, which, due to (4), fails to be true for $L_\pm^\alpha$ in the even-dimensional case. However, it holds for odd $d$ because in that case, $\Gamma\left(\frac{d-2}{2}\right)\cdot\Gamma\left(\frac{4-d}{2}\right) = \frac{\pi}{\sin\left(\frac{d-2}{2}\pi\right)} = (-1)^{\frac{d+1}{2}}\pi$ leads to $L^2_\pm=S_\pm$. In even dimensions, we need a further family that is "more singular at even numbers", so for $\alpha\in\C\backslash\{\hdots,d-2,d,d+2,\hdots\}$, we introduce
\begin{align} \label{DefLTildeL}
 \widetilde{L}^\alpha_\pm := \frac{(\pm i)^{d-\alpha-1}}{\sin\left(\frac{d-\alpha}{2}\pi\right)} \cdot L^\alpha_\pm = \widetilde{C}(\alpha,d) \cdot (\gamma \pm i0)^{\frac{\alpha-d}{2}}, \qquad \widetilde{C}(\alpha,d):=\frac{(\pm i)^{d-\alpha-1}\Gamma\left(\frac{d-\alpha}{2}\right)}{2^\alpha~\pi^{\frac{d}{2}} ~\Gamma\left(\frac{\alpha}{2}\right)}.
\end{align}

Indeed, the zeros of $L^\alpha_\pm$ and the poles of the prefactor in (\ref{DefLTildeL}) compensate, and hence, $\widetilde{L}^\alpha_\pm$ exist as distributions for all $\alpha=d-2n,~n\in\N$, i.e. $\alpha\mapsto\widetilde{L}^\alpha_\pm[\f]$ are meromorphic functions with simple poles at $\alpha=d,d+2,\hdots$ for fixed $\f\in\D(\R^d)$. From the definition follows that (1), (2), (3), (6) of Proposition \ref{PropertiesL} also hold for $\widetilde{L}^\alpha_\pm$ and in addition, we have $\widetilde{L}^2_\pm = S_\pm$ and at all non-pole-integers
\begin{align} \label{LTildeInteger}
 \widetilde{L}^\alpha_\pm = \left\{\begin{array}{cl} L^\alpha_\pm, \qquad & d-\alpha \text{ odd}, \\[2mm] \pm\frac{i}{\pi}L^d \Box^n\log(\gamma\pm i0), \qquad & \alpha= d-2n,~n\in\N. \end{array}\right.
\end{align}

\begin{Rmk} \label{FundSolDimTwo}
 Also for $d=1$ and $d=2$, (\ref{DefL}) - (\ref{LTildeInteger}) provide the respective fundamental solutions: 
 $$L^2(x)=\frac{|x|}{2}, \qquad \qquad \widetilde{L}^2_\pm = \pm\frac{i}{2\pi}\log(\gamma\pm i0),$$
 such that $\widetilde{L}^0_\pm=\delta_0$ also in these cases.
\end{Rmk}

Following section 1.4 of \cite{BGP2007}, we now transfer the families $\{L_\pm^\alpha\}_{\alpha\in\C}, \{\widetilde{L}_\pm^\alpha\}_{\alpha\in\C\backslash\{d,d+2,\hdots\}}$ locally to $M$. For $p\in M$ and $\Omega\subset M$ geodesically starshaped with respect to $p$, let 
$$\Gamma_p(q) := \gamma\big(\exp_p^{-1}(q)\big), \qquad \mu_p(q) := \big|\det\big(\df\exp_p\big)\big|_{\exp_p^{-1}(q)}\big|, \qquad q\in\Omega,$$ 
denote the squared Lorentz distance to $p$ and the distortion function (section 1.3 of \cite{BGP2007}), which is related to the van-Vleck-Morette-determinant via $\Delta=\frac{1}{\mu}$. On $\Omega$, we define
\begin{align} \label{DefLPatch}
 L_\pm^\Omega(\alpha,p)[\f] := L_\pm^\alpha\left[(\mu_p \f) \circ \exp_p\right], \qquad \f\in\D(\Omega),
\end{align}

which provides holomorphic maps $\alpha\mapsto L_\pm^\Omega(\alpha,p)[\f]$ and, analogously, meromorphic maps \mbox{$\alpha\mapsto \widetilde{L}_\pm^\Omega(\alpha,p)[\f]$} with simple poles at $\alpha=d,d+2,\hdots$.

\begin{Prop} \label{PropLDomain}
 Let $p\in M$ and $\Omega\subset M$ be geodesically starshaped with respect to $p$. Then, for all $\alpha\in\C$, we have:
 \begin{itemize}
  \item[(1)] For $\RT{\alpha}>d$, the maps $p\mapsto L^\Omega_\pm(\alpha,p)$ are continuous on $\Omega$ and given by
  \begin{align} \label{ContLPatch}
   L^\Omega_\pm(\alpha,p) = C(\alpha,d) \left(\Gamma_p \pm i0\right)^{\frac{\alpha-d}{2}}.
  \end{align}
  \item[(2)] $\Gamma_p \cdot L^\Omega_\pm(\alpha,p) = \alpha(\alpha-d+2) \cdot L^\Omega_\pm(\alpha+2,p)$,
  \item[(3)] $\grad \,\Gamma_p \cdot L^\Omega_\pm(\alpha,p) = 2\alpha ~ \grad L^\Omega_\pm(\alpha+2,p)$,
  \item[(4)] $\Box L^\Omega_\pm(\alpha+2,p) = \left(\frac{\Box\Gamma_p -2d}{2\alpha}+1\right) \cdot L^\Omega_\pm(\alpha,p), \quad \alpha \neq 0$.
  \item[(5)] For $\RT{\alpha}>0$, (\ref{DefLPatch}) yield distributions of order at most $\kappa_d$. Moreover, there is an open neighborhood $U$ of $p$ and some $C>0$ such that
             $$\left|L^\Omega_\pm(\alpha,q)[\f]\right| \leq C \cdot \|\f\|_{C^{\kappa_d}(\Omega)}, \qquad q\in U, ~ \f\in\D(\Omega).$$
  \item[(6)] Let $U \subset \Omega$ be an open neighborhood of $p$ such that $\Omega$ is geodesically starshaped with respect to all $q\in U$. Furthermore, let $\RT{\alpha}>0$ and $V \in C^{\kappa_d+k}(U \times \Omega)$ such that $\supp V(q,\cdot)\subset\Omega$ is compact for all $q\in U$. Then $q \mapsto L^\Omega_\pm(\alpha,q)[V(q,\cdot)] \in C^k(U)$.
  \item[(7)] For all $\f \in C^k_c(\Omega)$, the map $\alpha \mapsto L^\Omega_\pm(\alpha,p)[\f]$ is holomorphic on $\left\{\RT{\alpha} > d-2\lfloor\frac{k}{2}\rfloor\right\}$.
 \end{itemize}
 On the domain of holomorphicity of $\widetilde{L}^\Omega_\pm(\alpha,p)$, the statements (1) - (7) remain true, when we replace $L^\Omega_\pm(\alpha,p)$ and $C(\alpha,d)$ by $\widetilde{L}^\Omega_\pm(\alpha,p)$ and $\widetilde{C}(\alpha,d)$.
 \begin{itemize}
  \item[(8)] For $d-\alpha$ an odd integer, we have $L^\Omega_\pm(\alpha,p)= \widetilde{L}^\Omega_\pm(\alpha,p)$.
  \item[(9)] $\widetilde{L}_\pm^\Omega(0,p)=\delta_p$.
  \item[(10)] For all $n\in \N$, we have $L^\Omega_\pm(d-2n,p)=0$ and $L^\Omega_+(d+2n-2,p)=L^\Omega_-(d+2n-2,p)$.
 \end{itemize}
\end{Prop}

\begin{proof}
 The proofs of (1) - (7) are similar to those of Proposition 1.4.2 of \cite{BGP2007} and by employing the respective properties of $\widetilde{L}^\alpha_\pm$, they also apply here. Moreover, (8) follows directly from (\ref{LTildeInteger}), (10) from \ref{PropertiesL} and (9) from $\widetilde{L}^0_\pm=\delta_0$ and $\mu_p(p)= \det(\id_{T_pM}) = 1$.
\end{proof}

\begin{Cor}
 Proposition \ref{PropLDomain} leads to the following relations:
\begin{itemize}
 \item For $d$ odd, we have $L_\pm^\Omega(0,p)=\delta_p$.
 \item For $\Omega$ moreover convex, i.e. starshaped w.r.t. all $p\in\Omega$, $\alpha\in\C$ and $u\in\D(\Omega\times\Omega)$, we have
 \begin{align} \label{LSymmetry}
  \int_\Omega L_\pm^\Omega(\alpha,p)[u(p,\cdot)] \df V(p) = \int_\Omega L^\Omega_\pm(\alpha,q)[u(\cdot,q)] \df V(q)
 \end{align}
 and similarly for $\widetilde{L}^\Omega_\pm(\alpha)$.
\end{itemize}
\end{Cor}

\subsection{The Hadamard series}

Let $E$ be a real vector bundle over $M$, $\Omega\subset M$ geodesically starshaped with respect to some $p\in \Omega$ and $P\colon  \CC(M,E) \rightarrow \CC(M,E)$ a wave operator. Adopting the approach pursued in section 5.2 of \cite{G1964}, we start the deduction of local expressions for Feynman and anti-Feynman parametrices for $P$ at $p$ by taking the following ansatz of a formal Hadamard series:

\begin{align} \label{HadamardSeries} 
 \L_\pm(p) := \left\{\begin{array}{cl}
                      \mathop{\sum}\limits_{k=0}^\infty U_p^k L^\Omega_\pm(2k+2,p) + \mathop{\sum}\limits_{k=\frac{d-2}{2}}^\infty W_p^k L^\Omega(2k+2,p), \quad & d \text{ odd},\\[5mm]
                      \mathop{\sum}\limits_{k=0}^{\frac{d-4}{2}} U_p^k\widetilde{L}^\Omega_\pm(2k+2,p) \pm \frac{i}{\pi} \mathop{\sum}\limits_{k=\frac{d-2}{2}}^\infty \big(U_p^k\log (\Gamma_p \pm i0) + W_p^k\big)L^\Omega(2k+2,p), \quad & d \text{ even},
                     \end{array}\right.
\end{align}

with coefficients $U_p^k,W_p^k \in \CC\left(\Omega,E_p^*\otimes E\right)$ yet to be determined. For $\f\in\D(\Omega,E^*)$, we identify $U_p^k\f, W_p^k\f$ with $E_p^*$-valued test functions (see section 2.1 of \cite{BGP2007}), so $\L_\pm(p)$ is (formally) understood as a distribution mapping $\D(\Omega,E^*)$ to the complexified fiber $E_p^*\otimes_\R\C$. Similar to the procedure in chapter 2 of \cite{BGP2007}, we determine $U_p^k,W_p^k$ by formally demanding $P\L_\pm(p) = \delta_p$, which, by imposing the initial condition $U_p^0(p)=\id_{E_p^*}$ and by means of (\ref{PCompConn}) and Corollary \ref{TechnicalLemma}, leads to the transport equations
\begin{align} 
 2k~PU_p^{k-1} & = \nabla_{\grad \Gamma_p}U_p^k - \left(\frac{1}{2}\Box\Gamma_p-d+2k\right)U_p^k, \qquad k\in\N_0, \label{TranspEqU}\\[2mm]
 2k~PW_p^{k-1} & = \left\{\begin{array}{cl}
                            \nabla_{\grad\Gamma_p}W_p^k - \left(\frac{1}{2}\Box\Gamma_p-d+2k\right)W_p^k, \qquad & k+\frac{1}{2}\in\N,~k \geq \frac{d}{2}, \\[2mm]
                            \nabla_{\grad\Gamma_p}W_p^k - \left(\frac{1}{2}\Box\Gamma_p +2k-d\right) W_p^k + \frac{2k~PU_p^{k-1}}{k-\frac{d-2}{2}} - 2 U_p^k, \quad & k\in\N,~k\geq \frac{d}{2}.
                          \end{array}\right. \label{TranspEqW}
\end{align}

For the full calculation, we refer to section A.2 in the appendix of this work.

\begin{Rmk}
 Note that there is no constraint on $W_p^{\frac{d-2}{2}}$, which is therefore free to choose. Hence, even if (\ref{HadamardSeries}) converges, the requirement $P\L_\pm(p)=\delta_p$ determines $\L_\pm(p)$ only up to smooth solutions $\mathop{\sum}\limits_{k=0}^\infty c_{k,d}(W_p^k - \widetilde{W}_p^k)~\Gamma^k$ with $W_p^k, \widetilde{W}_p^k$ arising from different choices of $W_p^{\frac{d-2}{2}}$.
\end{Rmk}

\begin{Prop} \label{HadamardCoefficients}
 Let $O\subset\Omega$ be a non-empty domain such that $\Omega$ is geodesically starshaped with respect to all $p\in O$. For any $W_{\frac{d-2}{2}}\in\CC(O\times\Omega,E^*\boxtimes E)$, there are unique and smooth solutions of (\ref{TranspEqU}) and (\ref{TranspEqW}) given by
 {\allowdisplaybreaks\begin{align}
   U_0(p,q) & = \frac{\Pi_q^p}{\sqrt{\mu(p,q)}}, \nonumber \\[2mm] 
   U_k(p,q) & = -kU_0(p,q) \int_0^1 t^{k-1} U_0\big(p,\phi_{pq}(t)\big)^{-1} \big(P_{(2)}U_{k-1}\big)\big(p,\phi_{pq}(t)\big) \df t, \qquad k\geq 1, \label{Uk} \\[2mm]
   W_k(p,q) & = -kU_0(p,q) \int_0^1 t^{k-1} U_0\big(p,\phi_{pq}(t)\big)^{-1} \widehat{W}_{k-1}\big(p,\phi_{pq}(t)\big) \df t, \qquad k \geq \frac{d}{2}. \label{Wk}
 \end{align}}%
 $\Pi^p_q\colon  E_p \rightarrow E_q$ denotes the $\nabla$-parallel transport, $\phi_{pq}\colon  [0,1] \rightarrow \Omega$ the unique geodesic connecting $p,q$ (\ref{GeodesicMap}) and
 \begin{align*} 
  \widehat{W}_{k-1} := \left\{\begin{array}{cl}
                                      P_{(2)}W_{k-1}, \qquad & k+\frac{1}{2}\in\N,~k \geq \frac{d}{2}, \\[2mm]
                                      P_{(2)}\left(W_{k-1} - \frac{U_{k-1}}{k-\frac{d-2}{2}}\right) + \frac{U_k}{k}, \quad & k\in\N,~k\geq \frac{d}{2}.
                                     \end{array}\right.
 \end{align*}
\end{Prop}

\begin{proof}
 \vspace{-1mm} The transport equations (\ref{TranspEqU}) and for half-integer $k$ also (\ref{TranspEqW}) coincide with (2.3) of \cite{BGP2007}. Therefore, $U_k$ and for $k+\frac{1}{2}\in\N$ also $W_k$ are the Hadamard coefficients given by (\ref{Uk}) and (\ref{Wk}) due to Proposition 2.3.1 of \cite{BGP2007}. For integer $k$, we can apply the same proof for $W_k$ with $PW_p^{k-1}$ replaced by $\widehat{W}_p^{k-1}$ everywhere, for which the same procedure then leads to (\ref{Wk}).
\end{proof}

\begin{Cor} \label{SymmetricHCoeffLem}
 For $\Omega$ convex, $E$ equipped with a non-degenerate inner product and $P$ formally self-adjoint, we have symmetry of $U_k$ for all $k\in\N_0$ and, in case of symmetric $W_{\frac{d-2}{2}}$, of $W_k$ for half-integer $k\geq\frac{d}{2}$ in the sense of (\ref{SymmetricHCoeff}).
\end{Cor}

\begin{Rmk}
 Note that in the odd-dimensional case, $W_p^{\frac{d-2}{2}}=0$ leads to $W_p^k=0$ for all $k$, whereas, as a consequence of the coupling with $U_p^k$, for even dimensions, we have $W_p^k\neq0$, in general.\\
 More remarkably, the $W_k$'s fail to be symmetric in the even-dimensional case, i.e. for integer $k$, even if we chose $W_{\frac{d-2}{2}}$ to be symmetric. This phenomenon is closely related to the conformal trace anomaly, which indeed does not occur in odd-dimensional spacetimes \cite{W1978, W1994, DF2008}.
\end{Rmk}

\subsection{Local parametrices and Hadamard bidistributions}

From now on, let $\Omega\subset M$ always denote a convex domain, i.e. geodesically starshaped with respect to all $p\in\Omega$. We referred to (\ref{HadamardSeries}) as formal since in general, the series do not converge, and we merely employed it in order to extract the transport equations (\ref{TranspEqU}) - (\ref{TranspEqW}). Nevertheless, it leads to left parametrices by some well-known procedure \cite{F1975, G1988, BGP2007}, which smoothly cuts off $\L_\pm(p)$ away from its singular support and leads to convergent series on relatively compact domains. Due to the derivatives arising from the cut-off, this results in left parametrices for $P$ at $p$ rather than fundamental solutions. To be more precise, for $N>\frac{d}{2}$, some sequence $\{\e_k\}_{k\geq N}\subset(0,1]$ and $\sigma \in \D\big([-1,1],[0,1]\big)$ with $\sigma\big|_{\left[-\frac{1}{2},\frac{1}{2}\right]} \equiv 1$, we define
\begin{align} \label{Parametrices}
 \widetilde{\L}_\pm(p) := \left\{\begin{array}{cl}
                                  \mathop{\sum}\limits_{k=0}^\infty \widetilde{U}_p^k L^{\Omega}_\pm(2k+2,p) + \mathop{\sum}\limits_{k=\frac{d-2}{2}}^\infty\widetilde{W}_p^k L^{\Omega}(2k+2,p), \qquad & d \text{ odd}, \\[5mm]
                                  \mathop{\sum}\limits_{k=0}^{\frac{d-4}{2}} U_p^k \widetilde{L}^\Omega_\pm(2k+2,p) \pm \frac{i}{\pi} \mathop{\sum}\limits_{k=\frac{d-2}{2}}^\infty \big(\widetilde{U}_p^k \log(\Gamma_p \pm i0) + \widetilde{W}_p^k\big) L^\Omega(2k+2,p), \quad & d \text{ even},
                                 \end{array} \right.
\end{align}
where
\begin{align} \label{HadCoeffTilde}
 \widetilde{U}_k := \left\{\begin{array}{cl} U_k, \qquad & k<N, \\[2mm] \left(\sigma \circ \frac{\Gamma}{\e_k}\right) \cdot U_k, \qquad & k\geq N, \end{array}\right. \qquad \widetilde{W}_k := \left\{\begin{array}{cl} W_k, \qquad & k<N, \\[2mm] \left(\sigma \circ \frac{\Gamma}{\e_k}\right) \cdot W_k, \qquad & k\geq N. \end{array}\right.
\end{align}

\begin{Prop} \label{ParametricesProp}
 For any relatively compact domain $O\subset\Omega$ and any smooth choice of $W_{\frac{d-2}{2}}$, there is a sequence $\{\e_k\}_{k\geq N}\subset(0,1]$ such that (\ref{Parametrices}) yield well-defined distributions for all $p\in\Obar$, and
 \begin{itemize}
  \item[(i)] $\ssupp \left(\widetilde{\L}_\pm(p)\right) \subset C(p)$,
  \item[(ii)] $P\widetilde{\L}_\pm(p) = \delta_p + K_\pm(p,\cdot)$ with $K_\pm \in \CC(\Obar\times \Obar, E^*\boxtimes E)$,
  \item[(iii)] $p\mapsto\widetilde{\L}_\pm(p)[\f] \in \CC(\Obar,E^*)$ for all $\f \in \D(O,E^*)$,
  \item[(iv)] they are of order at most $\kappa_d$,
  \item[(v)] there is a constant $C>0$ such that $\big|\widetilde{\L}_\pm(p)[\f]\big|\leq C\|\f\|_{C^{\kappa_d}(O,E^*)}$ for all $p\in\Obar$ and $\f\in\D(O,E^*)$.
 \end{itemize}
\end{Prop}

The proofs of Lemma 2.4.1 - 2.4.4 of \cite{BGP2007} only employ smoothness of the Hadamard coefficients and $\sigma\left(\frac{\Gamma(p,q)}{\e_k}\right)=0$ if $|\Gamma(p,q)|\geq\e_k$, so replacing there $R^\Omega_\pm$ by $L^\Omega_\pm$ proves the Proposition in the odd-dimensional case. Similarly, we obtain convergence in $\CC$ of the $W_k$-part in even dimensions. However, for the logarithmic terms, we have to adapt the corresponding estimates, which is of purely technical nature and therefore removed to the Appendix. Considering $\widetilde{\L}_\pm,K_\pm$ as Schwartz kernels, we extract the corresponding operators
\begin{align}
 \widetilde{\LL}_\pm\colon \qquad & \D(O,E^*)  \longrightarrow \CC(\Obar,E^*), \qquad \f \longmapsto \big(p\mapsto\widetilde{\L}_\pm(p)[\f]\big), \label{LTildeOperator}\\[2mm]
 \KK_\pm\colon            \qquad & C^0(O,E^*) \longrightarrow \CC(\Obar,E^*), \qquad u  \longmapsto \left(p\mapsto\int_{\Obar} K_\pm(p,q)u(q) \df V(q)\right), \label{KIntOp}
\end{align}

which are bounded due to compactness of $\Obar$. 

Let $E$ be equipped with some non-degenerate inner product and $P$ be formally self-adjoint, which implies symmetry of the Hadamard coefficients (Theorem \ref{HadCoeffSymm}). The aim of the rest of the section is to show that then, for all choices involved in Proposition \ref{ParametricesProp}, the corresponding operators (\ref{LTildeOperator}) represent anti-Feynman and Feynman parametrices for $P^t$ in the sense of Duistermaat-H{\" o}rmander.

\begin{Cor} \label{DiffParametricesSmoothing}
 Let $\widetilde{\LL}_\pm$ and $\widetilde{\LL}'_\pm$ be the operators (\ref{LTildeOperator}) arising from two different choices of $N,W_{\frac{d-2}{2}},O$ and $\{\e_k\}_{k\in\N}$. Then $\widetilde{\LL}_\pm - \widetilde{\LL}'_\pm$ is a smoothing operator on $O\cap O'$.
\end{Cor}

\begin{proof}
 The Schwartz kernels of these differences are given by the bidistributions
 $$(p,q) \longmapsto \big(\widetilde{\L}_\pm(p) - \widetilde{\L}'_\pm(p)\big)(q),$$
 which are smooth by Lemma 2.4.3 of \cite{BGP2007} and Lemma \ref{ParametricesReg} since $\supp(\sigma_k-\sigma_k')\cap\Gamma^{-1}(0)=\emptyset$ for all $k$.
\end{proof}

Note that in terms of the operators (\ref{LTildeOperator}), (\ref{KIntOp}), Proposition \ref{ParametricesProp} (iii) reads $\widetilde{\LL}_\pm P^t = \id + \KK_\pm$ with $\KK_\pm$ smoothing. Hence, $\widetilde{\LL}_\pm$ are left parametrices for $P^t$ and due to formal self-adjointness of $P$, they also provide right parametrices:

\begin{Prop} \label{ApprFundSolSymCor}
 For $P$ formally self-adjoint, the operators $\widetilde{\LL}_\pm$ define two-sided parametrices for $P^t$.
\end{Prop}

\begin{proof}
 We just have to show that $\widetilde{\LL}_\pm$ yield right parametrices. From the symmetry properties of $L^\Omega_\pm(\alpha)$ and $\widetilde{U}_k$ ((\ref{LSymmetry}) and Theorem \ref{HadCoeffSymm}) directly follows 
 \begin{align} \label{Lfsadj}
  \int_O L^\Omega_\pm(2k+2,p)\big[\big(\widetilde{U}_k(p,\cdot)\f\big)\big(\psi(p)\big)\big] \df V(p) = \int_O L^\Omega_\pm(2k+2,q)\big[\widetilde{U}_k(q,\cdot)\Theta\psi\big(\Theta^{-1}_q\f(q)\big)\big] \df V(q)
 \end{align}
 for all $\f\in\D(O,E^*),\psi\in\D(O,E)$ and $k\in\N_0$. This works analogously for the logarithmic and $\widetilde{L}^\Omega_\pm$-terms in (\ref{Parametrices}). Furthermore, the series involving the coefficients $\widetilde{W}_k$ are given by convergent power series $\sum_{j=0}^\infty a_j \Gamma^j$, which yield smooth sections in $E^*\boxtimes E$. Altogether, we obtain the decomposition $\widetilde{\LL}_\pm = \mathcal{U}_\pm + \mathcal{W}$ with $\mathcal{U}_\pm$ representing the symmetric $\widetilde{U}_k$-part, i.e.\ $\mathcal{U}_\pm^t = \Theta^{-1}\mathcal{U}_\pm\Theta$, and $\mathcal{W}$ the smooth $\widetilde{W}_k$-part. Hence, $\widetilde{\LL}_\pm$ is symmetric up to smoothing in the sense
 \begin{align} \label{SymmetryUpToSmooth}
  \widetilde{\LL}_\pm = \Theta\widetilde{\LL}_\pm^t\Theta^{-1} + \underbrace{\mathcal{W} - \Theta\mathcal{W}^t\Theta^{-1}}_{\text{smoothing}},
 \end{align}
 from which we directly deduce the claim:
 \begin{align*}
  P^t\widetilde{\LL}_\pm & = P^t\big(\Theta\widetilde{\LL}_\pm^t\Theta^{-1} + \mathcal{W} - \Theta\mathcal{W}^t\Theta^{-1}\big) = \Theta P \widetilde{\LL}_\pm^t\Theta^{-1} + P^t\mathcal{W}^t - P^t\Theta\mathcal{W}^t\Theta^{-1} \\[2mm]
                        & = \Theta \big(\widetilde{\LL}_\pm P^t\big)^t\Theta^{-1} + \Theta P \big(\Theta^{-1}\mathcal{W} - \mathcal{W}^t\Theta^{-1}\big) = \id + \underbrace{\Theta \KK_\pm^t\Theta^{-1} + \Theta P \big(\Theta^{-1}\mathcal{W} - \mathcal{W}^t\Theta^{-1}\big)}_{\text{smoothing}}. \qedhere
 \end{align*}
\end{proof}

If $P$ is formally self-adjoint, then so is the operator 
\begin{align} \label{LDiff}
 \widetilde{\LL} := \frac{i}{2}(\widetilde{\LL}_+ - \widetilde{\LL}_-)
\end{align}

due to symmetry of its Schwartz kernel $\widetilde{\L} = \frac{i}{2}\big(\widetilde{\L}_+ - \widetilde{\L}_-\big)$ by (\ref{Lfsadj}).\\
So far, we found two-sided parametrices $\widetilde{G}_\pm, \widetilde{\LL}_\pm$ given by Hadamard series (\ref{RieszParametrix}), (\ref{Parametrices}), and \cite{SV2001} actually proved equivalence of the Hadamard condition (\ref{HConditionWF}) and that the bidistribution is given by a certain Hadamard series. This latter condition therefore allows us to express (\ref{TwoPointFunctionHadState}) in terms of $\widetilde{G}_\pm, \widetilde{\LL}_\pm$ by directly comparing the corresponding Hadamard series. More precisely, we confirm that $\frac{i}{2}\big(\widetilde{\L}_+ - \widetilde{\L}_- + \widetilde{G}_+ - \widetilde{G}_-\big)$ is a Hadamard bidistribution, which, up to smooth errors, moreover is a bisolution with the right antisymmetric part. By examining a further linear combination of parametrices, it will follow that $\widetilde{\L}_\pm$ represent a Feynman and an anti-Feynman parametrix.

\begin{Prop} \label{HadamardStructure}
 Let $O\subset \Omega$ be relatively compact, $P$ formally self-adjoint and $\widetilde{\Riesz}_\pm,\widetilde{\L}_\pm$ the bidistributions given by (\ref{RieszParametrix}) and (\ref{Parametrices}). Then, for
 \begin{align} \label{ApprHadamard}
  \widetilde{H} := \frac{i}{2}\big(\widetilde{\L}_+ - \widetilde{\L}_- + \widetilde{\Riesz}_- - \widetilde{\Riesz}_+\big),
 \end{align}
 the sections $P^t_{(1)}\widetilde{H},P_{(2)}\widetilde{H}$ are smooth, the antisymmetric part of $\widetilde{H}$ is given by $\frac{i}{2}\big(\widetilde{\Riesz}_- - \widetilde{\Riesz}_+\big)$ and $\widetilde{H}$ has the Hadamard singularity structure (\ref{HConditionWF}). Furthermore, 
 \begin{align} \label{SmoothLinComb}
  \widetilde{\L}_+ + \widetilde{\L}_- - \widetilde{\Riesz}_- - \widetilde{\Riesz}_+ \in \CC(O\times O, E^*\boxtimes E).
 \end{align}
\end{Prop}

\begin{proof}
 The first two claims are immediate conclusions from $\widetilde{\Riesz}_\pm,\widetilde{\LL}_\pm$ being two-sided parametrices and $\widetilde{\L}$ being symmetric, so we proceed with the Hadamard singularity structure.\\
 Let $k,j\in\N_0$ such that $k\geq j$, and for even $d$, let either $j,k\leq\frac{d-2}{2}$ or $j,k>\frac{d-2}{2}$. Then
 \begin{align} \label{Kjk}
  \frac{R^\Omega_\pm(2k+2)}{R^\Omega_\pm(2j+2)} = \frac{L^\Omega_\pm(2k+2)}{L^\Omega_\pm(2j+2)} = \frac{\widetilde{L}^\Omega_\pm(2k+2)}{\widetilde{L}^\Omega_\pm(2j+2)} = \underbrace{\frac{C(2k+2,d)}{C(2j+2,d)}}_{=:K_{k,j,d}\neq 0} ~ \Gamma^{k-j}
 \end{align}
 
 due to (\ref{DefLTildeL}) and Proposition \ref{PropLDomain}. $R^\Omega_\pm(\alpha)$ denote the Riesz distributions (\ref{RieszDistrDomain}), which are considered as bidistributions in the canonical way. Define 
 \begin{align} 
  \begin{split}\label{HSingularParts}
   H^\Omega(2) & := \frac{i}{2}\big(\widetilde{L}^\Omega_+(2) - \widetilde{L}^\Omega_-(2) + R^\Omega_-(2) - R^\Omega_+(2)\big), \\[2mm]
   H^\Omega(d) & := -\frac{L^\Omega(d)}{2\pi}\big(\log(\Gamma+i0) + \log(\Gamma-i0)\big) + \frac{i}{2}\big(R^\Omega_-(d) - R^\Omega_+(d)\big)
  \end{split}
 \end{align}
 
 such that the Hadamard series (\ref{ApprHadamard}) takes the form
 \begin{align} \label{HTildeSeries}
  \widetilde{H} = \left\{\begin{array}{cl} H^\Omega(2) \mathop{\sum}\limits_{k=0}^\infty K_{k,0,d}\cdot\widetilde{U}_k\Gamma^k, \qquad & d \text{ odd}, \\[2mm]
                                           H^\Omega(2) \mathop{\sum}\limits_{k=0}^{\frac{d-4}{2}} K_{k,0,d}\cdot\widetilde{U}_k\Gamma^k + H^\Omega(d) \mathop{\sum}\limits_{k=\frac{d-2}{2}}^\infty K_{k,\frac{d-2}{2},d}\cdot\widetilde{U}_k\Gamma^{k-\frac{d-2}{2}}, \qquad & d \text{ even}.
                         \end{array} \right.
 \end{align}
 We show that this is of Hadamard form in the sense of Definition 5.1 in \cite{SV2001}, where mostly the notations and conventions of \cite{G1988} are adopted. In particular, the Hadamard coefficients $U_{(k)}$ used in \cite{SV2001} are related with $U_k$ via $2^k k! \cdot U_{(k)} = U_k$ (see Remark 2.3.2 of \cite{BGP2007}). In addition, with the notation $(\alpha,k) := 2^k\cdot\frac{\Gamma\left(\frac{\alpha}{2}+k\right)}{\Gamma\left(\frac{\alpha}{2}\right)}$, we find $(2j+2,k-j)\cdot(2j+4-d,k-j)=K_{k,j,d}^{-1}$ as defined in (\ref{RecursionRiesz}), and hence,
 \begin{align*}
  K_{k,0,d}\cdot U_k             & = \frac{2^k\cdot k! \cdot U_{(k)}}{2^k\cdot k!\cdot (4-d,k)} = \frac{U_{(k)}}{(4-d,k)}, \\[2mm] 
  K_{k,\frac{d-2}{2},d}\cdot U_k & = \frac{2^k\cdot k! \cdot \Gamma\left(\frac{d}{2}\right) \cdot U_{(k)}}{4^{k-\frac{d-2}{2}}\cdot k! \cdot \Gamma\left(k+2-\frac{d}{2}\right)} = \frac{\left(2,\frac{d-2}{2}\right)}{2^{k+d-2}\cdot \left(k-\frac{d-2}{2}\right)!} \cdot U_{(k)}.
 \end{align*}
 
 For all $n\in\N$, we choose $N\geq n+\left\lceil\frac{d}{2}\right\rceil$ in (\ref{HadCoeffTilde}), so the series in (\ref{HTildeSeries}) truncated at $k=n+\left\lceil\frac{d}{2}\right\rceil$ coincide with $U,V^{(n)},T^{(n)}$ given in Appendix A.1 of \cite{SV2001}. The remainder term is then of regularity $C^n$ and corresponds to $H^{(n)}$ in Definition 5.1 of \cite{SV2001}. \\
 It remains to identify the singular terms (\ref{HSingularParts}) with $G^{(1)},G^{(2)}$ given by (5.3) in \cite{SV2001} up to some global factor, which is $-2$ in the odd- and $2\cdot(-1)^{\frac{d}{2}}$ in the even-dimensional case. Moreover, note that for the squared Lorentzian distance in the definition of $G^{(1)},G^{(2)}$, the convention $s=-\Gamma$ is used, whereas in Appendix A.1, it is $s=\Gamma$. \\

 Let $p,q\in O$. By definition of $\widetilde{L}^\Omega_\pm(\alpha,p)$ and $R^\Omega_\pm(\alpha,p)$ as well as (\ref{ExtractFaF}), we have $H^\Omega(2,p) = \big(\exp_p\big)_*W$ with $W$ Wightman's solution for $(\R^d_{\mathrm{Mink}},\Box)$. Then (\ref{Deltapm}) leads to
 $$H^\Omega(2,p)(q) = \frac{\Gamma\left(\frac{d-2}{2}\right)}{4\pi^{\frac{d}{2}}}\lm{\e}{0} \big(-\Gamma(p,q)+2i\e \cdot q^0 + \e^2\big)^{\frac{2-d}{2}}$$
 
 in the distributional sense with $q^0=\big(\exp_p^{-1}(q)\big)^0$. Since $\Gamma\left(\frac{d-2}{2}\right)\Gamma\left(1-\frac{d-2}{2}\right) = (-1)^{\frac{d+1}{2}}\pi$ for odd $d$, this coincides with $G^{(1)}$. Furthermore, one directly calculates $H^\Omega(d)-G^{(2)}=0$ away from $\Gamma^{-1}(0)$. For $\phi_{pq}(t):=\exp_p\big(t\,\exp_p^{-1}(q)\big)$ with $t\in[0,1]$, we obtain $\Gamma\big(p,\phi_{pq}(t)\big) = t^2\cdot\Gamma(p,q)$. Similar to (\ref{GammaEpsilon}), we set \mbox{$\Gamma_\e^\pm(p,\cdot) := \gamma_\e^\pm \circ \exp_p^{-1}$}, which yields $\Gamma_\e^\pm\big(p,\phi_{pq}(t)\big) = t^2 \cdot \Gamma_{\frac{\e}{t}}^\pm\big(p,q\big)$ for all $t\in(0,1]$, and hence,
 $$G^{(2)}\big(p,\phi_{pq}(t)) = -\frac{L^\Omega(d)}{\pi} \lm{\e}{0}\log\big(-\Gamma_\e^\pm\big(p,\phi_{pq}(t)\big)\big) = -\frac{2L^\Omega(d)}{\pi} \cdot \log t + G^{(2)}(p,q).$$
 On the other hand, since $R_\pm^\Omega(d), L^\Omega(d)$ are homogeneous distributions of degree $0$, (\ref{HSingularParts}) provides 
 $$H^\Omega(d,p)\big(\phi_{pq}(t)\big) = -\frac{L^\Omega(d)}{2\pi}\big(2\log t + 2\log t\big) + H^\Omega(d,p)(q) = -\frac{2L^\Omega(d)}{\pi} \cdot \log t + H^\Omega(d,p)(q).$$
 Of course, both expressions have to be understood in the distributional sense. Since $\Omega$ is diffeomorphic to \mbox{$\exp_p^{-1}(\Omega)\subset T_pM$} for all $p\in\Omega$, their difference corresponds to a $\sLL$-invariant distributions on Minkowski space $\R^d$, which is supported on the light cone and homogeneous of degree $0$. Therefore, by the classification (\ref{LorentzInvariantDistributionRiesz}), it has to vanish everywhere and thus, Theorem 5.8 of \cite{SV2001} ensures that $\widetilde{H}$ is of Hadamard form in the sense of (\ref{HConditionWF}).\\ 
 It remains to show (\ref{SmoothLinComb}). According to (\ref{HSingularParts}), we define
 \begin{align} 
  \begin{split}\label{ASingularParts}
   A^\Omega(2) & := \frac{i}{2}\big(\widetilde{L}^\Omega_+(2) + \widetilde{L}^\Omega_-(2) - R^\Omega_-(2) - R^\Omega_+(2)\big), \\[2mm]
   A^\Omega(d) & := -\frac{L^\Omega(d)}{2\pi}\big(\log(\Gamma+i0) - \log(\Gamma-i0)\big) - \frac{i}{2}\big(R^\Omega_-(d) + R^\Omega_+(d)\big)
  \end{split}
 \end{align}
 such that for (\ref{SmoothLinComb}), we obtain the expression (\ref{HTildeSeries}) with $H^\Omega(2), H^\Omega(d)$ replaced by $A^\Omega(2), A^\Omega(d)$ and it suffices to show smoothness of the bidistributions (\ref{ASingularParts}). For $A^\Omega(2)$, this follows directly from the definitions of $\widetilde{L}^\Omega_\pm(2), R^\Omega_\pm(2)$ as pullbacks of $S_\pm, R^2_\pm$ along a diffeomorphism and (\ref{ExtractFaF}). On the other hand, one directly calculates that $A^\Omega(d)$ is given by the constant $-iC(d,d)$.
\end{proof}

\vspace{-2mm} Since $\widetilde{\L}_\pm,\widetilde{\Riesz}_\pm$ are determined merely up to smooth sections, without loss of generality, we regard (\ref{SmoothLinComb}) as the equality
\begin{align} \label{EqualitySums}
 \widetilde{\L}_+ + \widetilde{\L}_- = \widetilde{\Riesz}_- - \widetilde{\Riesz}_+.
\end{align}

Now we can deduce from Theorem 5.1 of \cite{R1996} and section 6.6 of \cite{DH1972} that the operators $\widetilde{\LL}_\pm$ represent anti-Feynman and Feynman parametrices for $P^t$ in the sense of (\ref{WFDistParametrices}).

\subsection{Local fundamental solutions and Hadamard bisolutions}

It remains to construct actual local bisolutions $S^O$ for $P$ from the parametrices $\widetilde{\LL}_\pm$ following the lines of section 2.4 of \cite{BGP2007}. By Proposition \ref{ParametricesProp} (ii), we have $\widetilde{\L}_\pm P^t\big|_{\D(O,E^*)} = \id + \KK_\pm$, and hence, fundamental solutions are obtained by inverting the operators $\id + \KK_\pm$. Indeed, if $\vol(\Obar) \cdot \|K_\pm\|_{C^0(\Obar\times \Obar)} < 1$, that is, for $O$ chosen "small enough", (\ref{KIntOp}) provides isomorphisms $\id+\KK_\pm\colon C^l(\Obar,E^*) \longrightarrow C^l(\Obar,E^*)$ for all $l\in\N_0$ with bounded inverses given by the Neumann series
 \begin{align} \label{NeumannSeries}
  (\id + \KK_\pm)^{-1} = \sum_{j=0}^\infty (-\KK_\pm)^j.
 \end{align}
This means that all $C^l$-norms of the series exist, which follows from compactness of $\Obar$ and smoothness of $K_\pm$. The full proof coincides with the one of Lemma 2.4.8 of \cite{BGP2007}. In the following, we restrict to such small domains: 

\begin{Def} \label{AdmissibleSubsets}
 We call a relatively compact and causal subdomain $O$ of $\Omega$ admissible if Proposition \ref{ParametricesProp} provides parametrices $\widetilde{\LL}_\pm$ via (\ref{LTildeOperator}) such that the smooth Schwartz kernel $K_\pm$ of $\widetilde{\LL}_\pm P^t - \id$ fulfills 
 \begin{align} \label{SmallDomain}
  \vol(\Obar) \cdot \|K_\pm\|_{C^0(\Obar\times \Obar)} < 1.
 \end{align}
\end{Def}

More precisely, $O$ is admissible if there is a choice of $\{\e_k\}_k$ and $W_{\frac{d-2}{2}}$ such that (\ref{SmallDomain}) holds for the corresponding $K_\pm$. Lemma 2.4.8 of \cite{BGP2007} shows that for $O$ admissible, the corresponding operators $\widetilde{\LL}_\pm P^t = \id+\KK_\pm$ are isomorphisms with bounded inverses (\ref{NeumannSeries}).

\begin{Prop}
 For any admissible $O$, the operators 
 \begin{align*} 
  \widetilde{\S}^O_\pm:=(\id+\KK_\pm)^{-1}\widetilde{\LL}_\pm\colon \qquad \D(O,E^*) \rightarrow \CC(\Obar,E^*)
 \end{align*}
 fulfill $\widetilde{\S}^O_\pm P^t\big|_{\D(O,E^*)}=\id$, and hence, the distributions $\widetilde{S}_\pm^O(p),~p\in O,$ given by
 \begin{align} \label{TrueFundSols}
  \widetilde{S}_\pm^O(p)[\f] = \big((\id+\KK_\pm)^{-1}\widetilde{\LL}_\pm\f\big)(p) ,\qquad \f\in\D(O,E^*),
 \end{align}
 yield fundamental solutions for $P$ at $p$. Furthermore, $\QQ_\pm := (\id+\KK_\pm)^{-1} - \id$ are smoothing operators.
\end{Prop}

\begin{proof}
 The first claim follows from $\widetilde{\LL}_\pm P^t=\id+\KK_\pm$. Moreover, Proposition \ref{ParametricesProp} and Lemma 2.4.10 of \cite{BGP2007}, with $\widetilde{\Riesz}_\pm(\cdot)[\f]$ and $F^\Omega_\pm(\cdot)[\f]$ replaced by $\widetilde{\LL}_\pm\f$ and $\widetilde{\S}_\pm^O\f$, show that (\ref{TrueFundSols}) yield fundamental solutions. Finally, (\ref{NeumannSeries}) directly yields $\QQ_\pm=(\id+\KK_\pm)^{-1}\circ \KK_\pm$, which is smoothing since $\KK_\pm$ is, and $(\id+\KK_\pm)^{-1}$ is a continuous map $\CC(M,E^*) \rightarrow \CC(M,E^*)$.
\end{proof}

From now on, let $E$ be always equipped with some non-degenerate inner product, $P$ formally self-adjoint and $O$ admissible.

\begin{Prop} \label{FundSolParametrixSmooth}
 The operators $\widetilde{\S}^O_\pm - \widetilde{\LL}_\pm$ are smoothing.
\end{Prop}

\begin{proof}
 Note that $\widetilde{\S}^O_\pm - \widetilde{\LL}_\pm = \QQ_\pm\widetilde{\LL}_\pm$. Since $\QQ_\pm, \widetilde{\LL}_\pm$ are bounded and $\QQ_\pm$ has a smooth Schwartz kernel, they extend to bounded maps
 $$\QQ_\pm\colon\quad \D(O,E^*)' \rightarrow \CC(\Obar,E^*), \qquad \qquad \widetilde{\LL}_\pm\colon\quad \E(O,E^*)' \rightarrow \D(O,E^*)'.$$
 Hence, $\QQ_\pm\widetilde{\LL}_\pm\colon \E(O,E^*)'\rightarrow\CC(\Obar,E^*)$ are bounded and therefore smoothing.
\end{proof}

It follows that $\widetilde{\S}_\pm^O$ yield anti-Feynman and Feynman parametrices for $P^t$ on $O$. Moreover, their Schwartz kernels determine a real-valued bidistribution via
\begin{align} \label{RealValuedBidistribution}
 \widetilde{S}^O[\psi,\f] := \frac{i}{4}\left(\widetilde{S}_+^O[\psi,\f] - \widetilde{S}_-^O[\psi,\f] - \overline{\widetilde{S}_+^O[\psi,\f]} + \overline{\widetilde{S}_-^O[\psi,\f]}\right), \qquad \psi\in\D(O,E),~\f\in\D(O,E^*),
\end{align}

which has the right singularity structure and is a solution for $P$ in the second argument, meaning $\WF\big(\widetilde{S}^O~\big) = \WF\big(\widetilde{\L}~\big)$ and $P_{(2)}\widetilde{S}^O = 0$. 

\begin{Prop} \label{ExHadBisol}
 Let $M$ be a globally hyperbolic Lorentzian manifold, $\pi\colon E\rightarrow M$ a real vector bundle with non-degenerate inner product and $P\colon \CC(M,E)\rightarrow\CC(M,E)$ a formally self-adjoint wave operator. Furthermore, let $O\subset M$ be admissible and $\widetilde{\L}$ denote the bidistribution given by Proposition \ref{ParametricesProp} and (\ref{LDiff}). Then there is a bisolution $S^O\colon \D(O,E)\times\D(O,E^*) \rightarrow\R$ for $P$ with $\WF\big(S^O\big) = \WF\big(\widetilde{\L}~\big)$.
\end{Prop}

\begin{proof}
 Since $O$ is admissible, we obtain fundamental solutions $\widetilde{S}_\pm^O(p)$ at each $p\in \Obar$, and furthermore, (\ref{RealValuedBidistribution}) provides $p\mapsto\widetilde{S}^O(p)[\f]\in\CC(\Obar,E^*)$ for all $\f\in\D(O,E^*)$. Moreover, as a causal subdomain of a globally hyperbolic Lorentzian manifold, $O$ is globally hyperbolic on its own right (Lemma A.5.8 of \cite{BGP2007}). Hence, for $\Sigma$ a Cauchy hypersurface of $O$ with unit normal field $\nu$, there is a unique smooth solution of
 \begin{align*} 
  \left\{\begin{array}{cl} P^t \big(S^O(\cdot)[\f]\big) & = 0, \\[2mm] S^O(\cdot)[\f]\big|_\Sigma & = \widetilde{S}^O(\cdot)[\f]\big|_\Sigma, \\[2mm] \nabla_\nu \big(S^O(\cdot)[\f]\big)\big|_\Sigma & = \nabla_\nu\widetilde{S}^O(\cdot)[\f]\big|_\Sigma. \end{array}\right.
 \end{align*} 
 
 By continuous dependence on the Cauchy data, $S^O(p)$ defines an $E_p^*$-valued distribution for all $p\in O$. Furthermore, $S^O(\cdot)[P^t\f]=0$ for all $\f$ since it satisfies the trivial Cauchy problem. \\
 It remains to check the wave front set, i.e. smoothness of $D^O:=\widetilde{\L}-S^O$. Since $S^O, \widetilde{S}^O$ and $\widetilde{\L}$ yield parametrices for $P$, the sections given by $P_{(2)}D^O$, $P^t_{(1)}D^O$ and $\widetilde{\L}-\widetilde{S}^O$ are smooth, and hence, $D^O$ is the solution of a Cauchy problem with smooth Cauchy data, which is smooth by Theorem \ref{CauchyProblemBisection}.
\end{proof}

Altogether, any choice of parametrices $\widetilde{\L}_\pm$ in the sense of Proposition \ref{ParametricesProp} leads to a bisolution $S^O$ with singularity structure given by $\frac{i}{2}(\widetilde{G}_{aF}-\widetilde{G}_F)$ in the sense of (\ref{WFDistParametrices}), so we constructed bisolutions with the Hadamard singularity structure on every $O\times O$:

\begin{Thm} \label{LocalHadamardBisolutionThm}
 Let $M$ be a globally hyperbolic Lorentzian manifold, $O\subset M$ an admissible domain, $\pi\colon E\rightarrow M$ a real vector bundle with non-degenerate inner product and $P\colon \CC(M,E)\rightarrow\CC(M,E)$ a formally self-adjoint wave operator. For $S^O$ the bisolution given by Proposition \ref{ExHadBisol} and $G_\pm$ the advanced and retarded Green operator on $O$, the bisolution $H^O := S^O + \frac{i}{2}\big(G_+ - G_-\big)$ is of Hadamard form.
\end{Thm}

\begin{proof}
 Due to smoothness of $S^O-\widetilde{\L}$ and $G_\pm-\widetilde{G}_\pm$, we obtain smoothness of $H^O-\widetilde{H}$, so Proposition \ref{HadamardStructure} ensures the Hadamard property of $H^O$.
\end{proof}

\begin{Rmk}
 For $l\in\N_0$ and $p\in \Omega$, let $\L_\pm^{N+l}(p)$ denote the formal Hadamard series (\ref{HadamardSeries}) with all terms $k\geq N+l$ removed and $\L^{N+l} := \frac{i}{2}\big(\L^{N+l}_+ - \L^{N+l}_-\big)$. Then Lemma 2.4.2 of \cite{BGP2007} with $R^\Omega_\pm$ replaced by $L^\Omega_\pm$ and Lemma \ref{ParametricesReg} for the logarithmic part show that 
 $(p,q) \mapsto \big(S^O(p) - \L^{N+l}(p)\big)(q)$ represents a $C^l$-section over $O\times O$. This and the corresponding result for $\widetilde{G}_\pm$ (Proposition 2.5.1 in \cite{BGP2007}) ensure that $H^O$ is given by a Hadamard series up to terms of arbitrarily high regularity.
\end{Rmk}

\section{Global Hadamard two-point-functions}

So far, we constructed Hadamard bisolutions on products of certain small patches $O\times O$, and in this final section, the construction of global bisolutions $S$, which locally coincide with those $S^O$ up to smooth bisolutions and thus inherit their singularity structure, is tackled. Assuming $E$ to be Riemannian and the validity of Theorem 6.6.2 of \cite{DH1972} for sections in $E$, we furthermore prove the existence of a smooth bisolution $u$ such that $S+u$ is symmetric and positive, i.e. a Hadamard two-point-function.

\subsection{Global construction of symmetric bisolutions} \label{GlobalConstruction}

For $M$ globally hyperbolic, we fix a Cauchy hypersurface $\Sigma\subset M$ and two locally finite covers \mbox{$\O:=\{O_i\}_{i\in I}$}, \mbox{$\O':=\big\{O_i'\big\}_{i\in I}$} of it by admissible subsets of $M$ with $\overline{O}_i\subset O_j'$ if and only if $i=j$. Without loss of generality, we assume that $O_i\cap\Sigma$ is a Cauchy hypersurface of $O_i$ for all $i$. Then $N:=\bigcup_{i\in I}O_i$ yields a causal normal neighborhood of $\Sigma$ in the sense of Lemma 2.2 of \cite{KW1991}. By paracompactness of $M$ and the Hopf-Rinow-Theorem, we find an exhaustion $\{A_m\}_{m\in\N}$ of $I$ by finite subsets such that the relatively compact sets $N_m := \bigcup_{i\in A_m}O_i$ exhaust $N$ and every compact subset of $N$ is contained in some $N_m$. Besides that, causality of $O$ implies $O\subset D(O) = D(O\cap\Sigma)$ and therefore,
\begin{align} \label{NjGlobalHyperbolic}
 N_m \subset \bigcup_{i\in A_m} D(O_i\cap\Sigma) \subset D\left(\bigcup_{i\in A_m}O_i\cap\Sigma\right) = D(N_m\cap\Sigma),
\end{align}

where $D$ stands for the Cauchy development of the respective set. It follows that every inextendible causal curve in $N_m$ meets $N_m\cap\Sigma$ exactly ones, so $N_m\cap\Sigma$ is a Cauchy hypersurface of $N_m$, i.e.\ $N_m$ is globally hyperbolic. In addition, for all $i\in I$, we choose the corresponding local bisolutions $S^{O_i'},S^{O_i}$ obtained by Theorem \ref{LocalHadamardBisolutionThm} such that $S^{O_i'}\big|_{O_i\times O_i} = S^{O_i}$.

\begin{Prop} \label{PropFirstGlob}
 For each $O\in\O$, there is a bisolution $\widehat{S}^O$ on $O\times M$ satisfying $\widehat{S}^O\big|_{O\times O} = S^O$.
\end{Prop}

\begin{proof}
 Let $O'\in\O'$ such that $\Obar\subset O'$, and $\chi\in\D(O')$ with $\chi\big|_{O}=1$. Then $\chi S^{O'}(p)$ is a well-defined distribution with spacelike compact support on $M$ for all $p\in O$, since $\chi S^{O'}(p)[\f]=S^{O'}(p)[\chi\f],$ \mbox{$\f\in\D(M,E^*)$}. With regard to Theorem \ref{CauchyProblemSingularSolution}, we define $\widehat{S}^{O}(p) \in \D(M,E,E^*_p)'$ as the unique solution of 
 \begin{align} \label{DefSHut}
  \left\{\begin{array}{cl} P\widehat{S}^{O}(p) & = 0, \\[2mm] \widehat{S}^{O}(p)\big|_\Sigma & = \chi S^{O'}(p)\big|_\Sigma, \\[2mm] \nabla_\nu \widehat{S}^{O}(p)\big|_\Sigma & = \nabla_\nu \big(\chi S^{O'}(p)\big)\big|_\Sigma, \end{array} \right.
 \end{align}
 which moreover depends smoothly on $p$ in the sense $p\mapsto\widehat{S}^{O}(p)[\f]\in\CC(O,E^*)$ for fixed $\f\in\D(M,E^*)$. Furthermore, global hyperbolicity of $O$ ensures $\widehat{S}^{O}\big|_{O\times O}=S^{O}$ by Theorem \ref{CauchyProblemBisection}, since the difference solves the trivial Cauchy problem on $O\times O$. \\
 Let $T(p)[\f]:=P^t\big(\widehat{S}^{O}(\cdot)[\f]\big)(p)$ and hence $T(p) \in \D(M,E,E^*_p)'$ for all $p\in O$. It follows that $PT(p)=0$, and $T(p)[\f] = 0 = T(p)[\nabla_\nu\f]$ if $\supp\f\subset O$, which leads to $T(p)\big|_\Sigma= \nabla_\nu T(p)\big|_\Sigma=0$. Consequently, it satisfies the trivial Cauchy problem, so we have $T(p)=0$, that is, $\widehat{S}^{O}$ represents a bisolution.
\end{proof}

This definition of $\widehat{S}^O$ is independent of the choice of $\chi$ in an appropriate sense: Let $\widetilde{\chi}\in\D(O')$ be another cut-off with $\widetilde{\chi}\big|_{O}=1$ and corresponding bisolution $\widetilde{S}^{O}$. Then $D := \widehat{S}^{O}-\widetilde{S}^{O}$ is a bisolution with Cauchy data on $(O\cap\Sigma)\times\Sigma$ given by $(\chi-\widetilde{\chi}) S^{O'}$. Recall that $\ssupp S^{O'} \subset \Gamma^{-1}(0) \cap (O'\times O')$, so causality of $O$ yields $\ssupp S^{O'}(p)\big|_\Sigma \subset \big(C^M(p)\cap O'\cap\Sigma\big) \subset O$ for all $p\in O$, and hence, $\ssupp \chi S^{O'}\big|_{O\cap\Sigma\times \Sigma}$ is contained in $O\times O$. Since $D$ satisfies the trivial Cauchy problem on $O\times O$, i.e.\ $D\big|_{O\times O}=0$, it is a smooth bisolution by Theorem \ref{CauchyProblemBisection}. Therefore, $\widehat{S}^O$ and $\widetilde{S}^O$ differ merely by some smooth bisolution.


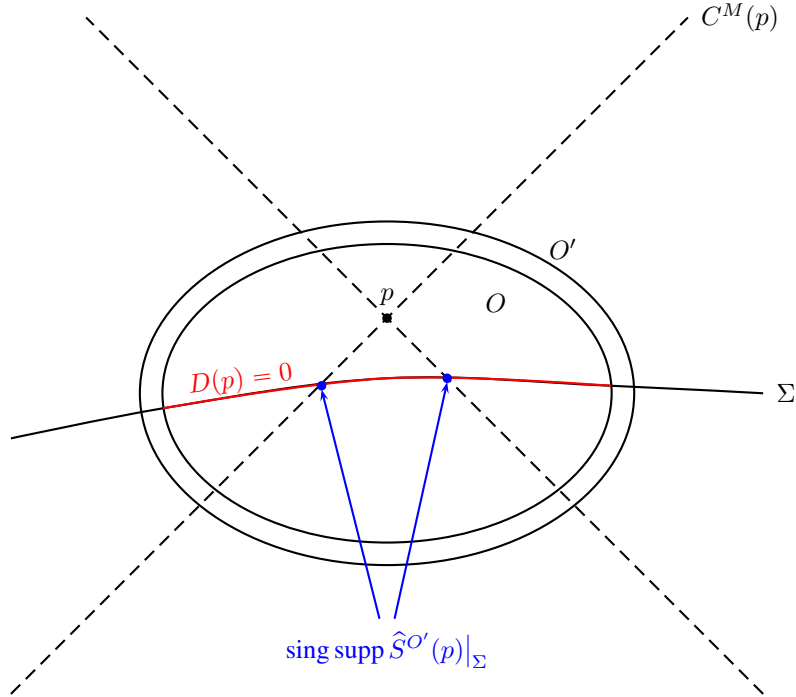
\begin{center}
 \begin{pspicture}
  \psellipse(5,5)(3.3,2.3)
  \psellipse(5,5)(3,2)
  \pscurve(0,4.4)(2.02,4.8)(5,5.2)(7.98,5.1)(10,5)
  \pscurve[linecolor=red](2.02,4.8)(5,5.2)(7.98,5.1)
  \psline[linestyle=dashed](0,1)(9,10)
  \psline[linestyle=dashed](10,1)(1,10)
  \psdot(5,6)
  \psdot[linecolor=blue](4.13,5.1)
  \psdot[linecolor=blue](5.8,5.2)
  \uput[90](5,6){$p$}
  \uput{1.7}[40](5,5){$O$}
  \uput{2.8}[40](5,5){$O'$}
  \uput[0](10,5){$\Sigma$}
  \uput{0.4}[30]{7.5}(2,4.8){\red $D(p)=0$}
  \uput[0](9,10){$C^M(p)$}
  \psline[arrows=->,linecolor=blue](4.9,2)(4.13,5.05)
  \psline[arrows=->,linecolor=blue](5.1,2)(5.8,5.15)
  \uput[-90](5,2){\blue $\ssupp \widehat{S}^{O'}(p)\big|_\Sigma$}
 \end{pspicture}
\end{center}

Next, we prove the existence of a compatible choice of bisolutions $\{\widehat{S}^{O_i}\}_{i\in I}$, meaning that they coincide on the overlaps $O_i\cap O_j,~i,j\in I$. In this way, these compatible bisolutions assemble to a well-defined object on $N\times M$. The tools for such a procedure are provided by \v{C}ech cohomology theory, for which we give a brief and purposive overview. For an introduction to this subject with details and proofs, we refer to section 5.33 of \cite{W1983}.\\
On $N\times M$, let $\CGerm$ denote the sheaf given by the germs of the smooth sections in $E^*\boxtimes E$ (see Example 5.2 in \cite{W1983}). For the open cover $\O^M:=\{O_i\times M\}_{i\in I}$ of $N\times M$, the $n$-simplices correspond to the non-empty ($n+1$)-times intersections 
$$O^M_{i_0\hdots i_n} := \big(O_{i_0} \cap\hdots\cap O_{i_n}\big)\times M, \qquad i_0,\hdots, i_n\in I,$$ 
with $n+1$ faces $\big\{O^M_{i_0\hdots\hat{i}_k \hdots i_n}\big\}_{k=0,\hdots, n}$ obtained by leaving out one $O_i$ in the intersection, respectively. An \mbox{$n$-cochain} is a map that assigns to each non-empty $O^M_{i_0\hdots i_n}$ a section of $\CGerm$ over $O^M_{i_0\hdots i_n}$, which we identify with the elements of $\CC\big(O^M_{i_0\hdots i_n},E^*\boxtimes E\big)$. The space of $n$-cochains is denoted by $\mathcal{C}^n(\O^M,\CGerm)$, where $\mathcal{C}^n :=\{0\}$ if $n<0$, and the coboundary operator is defined by
$$\partial_n\colon \quad \mathcal{C}^n(\O^M,\CGerm) \longrightarrow \mathcal{C}^{n+1}(\O^M,\CGerm), \quad (\partial_n f_n)\big(O^M_{i_0\hdots i_{n+1}}\big) := \sum_{k=0}^{n+1} (-1)^k \cdot f_n\big(O^M_{i_0\hdots\hat{i}_k \hdots i_{n+1}}\big)\big|_{O^M_{i_0\hdots i_{n+1}}}.$$
It follows that $\partial_{n+1} \circ\partial_n=0$ for all $n\in\N_0$ and we set $H^n(\O^M,\CGerm) := \frac{\ker\partial_n}{\ran\partial_{n-1}}$. These modules are trivial for all $n\in\N$ by some well-known construction (e.g. p. 202 in \cite{W1983}), employing that $\CGerm$ admits a partition of unity subordinate to the locally finite cover $\O^M$:

\begin{Lem} \label{TrivialCechCohomology}
 For all $n\in\N$, we have 
 $$H^n(\O^M,\CGerm)=\{0\}.$$
\end{Lem}

\begin{proof}
 By choice of $\O$, the cover $\O^M$ is locally finite. Let $\{\chi_i\}_{i\in I}$ denote a partition of unity subordinate to $\O^M$ and $f_n\in\mathcal{C}^n(\O^M,\CGerm)$. Then, for each $i\in I$, the smooth section $\chi_i f_n\big(O^M_i\cap O^M_{i_0\hdots i_{n-1}}\big)$ is supported in $O^M_i\cap O^M_{i_0\hdots i_{n-1}}$, and thus, via extension by zero, we consider it as an element of $\CC\big(O^M_{i_0\hdots i_{n-1}},E^*\boxtimes E\big)$. In this way, we obtain homomorphisms $h_n\colon  \mathcal{C}^n(\O^M,\CGerm) \rightarrow \mathcal{C}^{n-1}(\O^M,\CGerm)$ via
 $$h_n(f_n)\big(O^M_{i_0\hdots i_{n-1}}\big) := \sum_{i\in I} \chi_i  f_n\big(O^M_i\cap O^M_{i_0\hdots i_{n-1}}\big) \in \CC\big(O^M_{i_0\hdots i_{n-1}},E^*\boxtimes E\big),$$
 which satisfy
 \begin{align*}
  \big(h_{n+1} (\partial_n f_n)\big)\big(O^M_{i_0\hdots i_n}\big) & = \sum_{i\in I} \chi_i \cdot \partial_n f\big(O^M_i\cap O^M_{i_0\hdots i_n}\big) \\[2mm]
                                                                       & = \sum_{i\in I} \chi_i f_n\big(O^M_{i_0\hdots i_n}\big) + \sum_{i\in I}\sum_{k=0}^n (-1)^{k+1} \cdot \chi_i f_n\big(O^M_i\cap O^M_{i_0\hdots \hat{i}_k \hdots i_n}\big)\big|_{O^M_{i_0\hdots i_n}} \\[2mm]
                                                                       & = f_n\big(O^M_{i_0\hdots i_n}\big) - \big(\partial_{n-1} h_n(f_n)\big)\big(O^M_{i_0\hdots i_n}\big).
 \end{align*}
 Hence, $f_n\in \ker \partial_n$ implies $f_n = \partial_{n-1} h_n(f_n)$, that is, $f_n \in \ran \partial_{n-1}$. 
\end{proof}

\begin{Lem} \label{LemCech}
 For all $i\in I$, there is a bisolution $h_i\in\CC(O_i\times M,E^*\boxtimes E)$ such that
 \begin{align*} 
  \big(\widehat{S}^{O_i} + h_i\big)\big|_{O^M_{ij}} = \big(\widehat{S}^{O_j} + h_j\big)\big|_{O^M_{ij}}, \qquad i,j \in I.
 \end{align*}
\end{Lem}


\begin{proof}
 For $i,j\in I$, we consider the bisolution $h_{ij} := \widehat{S}^{O_i}\big|_{O^M_{ij}} - \widehat{S}^{O_j}\big|_{O^M_{ij}}$. For all $m\in \N$, Proposition \ref{ParametricesProp} provides parametrices $\widetilde{\L}_\pm^m$ on the relative compact domains $N_m$ such that for $\widetilde{\L}^m := \frac{i}{2}\big(\widetilde{\L}_+^m-\widetilde{\L}_-^m\big)$, Propositions \ref{ExHadBisol} and \ref{PropFirstGlob} yield
 \begin{align} \label{hAlphaBetaSmooth}
  h_{ij}\big|_{O_{ij}\times N_m} = \underbrace{\widehat{S}^{O_i} - \widetilde{\L}^m}_{\in\CC} - \big(\underbrace{\widehat{S}^{O_j} - \widetilde{\L}^m}_{\in\CC}\big) \in \CC\big(O_{ij}\times N_m, E^*\boxtimes E\big).
 \end{align}
 
 Such $\widetilde{\L}^m$ exist for all $m$ and $\{N_m\}_{m\in\N}$ exhausts $N$, so we have smoothness on $O_{ij}\times N$. Furthermore, as $O_{ij}$ is causal and $N$ a neighborhood of a Cauchy hypersurface, $h_{ij}$ fulfills a Cauchy problem with smooth Cauchy data and hence is smooth on all of $O^M_{ij}$ by Theorem \ref{CauchyProblemBisection}.\\
 Therefore, recalling the identification of sections of $\CGerm$ with smooth sections in $E^*\boxtimes E$, the map $f_1\colon O^M_{ij} \mapsto h_{ij}$ represents a \v{C}ech-1-cochain, which moreover is a cocycle, since
 \begin{align*}
  (\partial_1 f_1)(O^M_{ijk}) & = h_{jk}\big|_{O_{ijk}^M} - h_{ik}\big|_{O_{ijk}^M} + h_{ij}\big|_{O_{ijk}^M} \\[2mm]
                                          & = \widehat{S}^{O_j}\big|_{O_{ijk}^M} - \widehat{S}^{O_k}\big|_{O_{ijk}^M} - \widehat{S}^{O_i}\big|_{O_{ijk}^M} + \widehat{S}^{O_k}\big|_{O_{ijk}^M} + \widehat{S}^{O_i}\big|_{O_{ijk}^M} - \widehat{S}^{O_j}\big|_{O_{ijk}^M}\\[2mm] 
                                          & = 0
 \end{align*}
 
 for all $i,j,k\in I$. Thus, Lemma \ref{TrivialCechCohomology} ensures the existence of $f_0\colon  O^M_i\mapsto \widetilde{h}_i\in\CC(O^M_i, E^*\boxtimes E)$ such that $\partial_0 f_0=f_1$, and hence,
 \begin{align*} 
  h_{ij} = f_1(O^M_{ij}) = \partial_0 f_0(O^M_{ij}) = f_0(O^M_j)\big|_{O^M_{ij}} - f_0(O^M_i)\big|_{O^M_{ij}} = \widetilde{h}_j\big|_{O^M_{ij}} - \widetilde{h}_i\big|_{O^M_{ij}}, \qquad i,j\in I.
 \end{align*}
 
 Recall that $O_i \cap \Sigma$ is a Cauchy hypersurface of $O_i$ for all $i\in I$ and thus, each $\widetilde{h}_i$ determines a bisolution $h_i\in\CC\big(O^M_i,E^*\boxtimes E\big)$ via Theorem \ref{CauchyProblemBisection}. On the other hand, due to causality of $O_{ij}$, we have a well-posed Cauchy problem on $O^M_{ij}$, and consequently, $h_j\big|_{O^M_{ij}} - h_i\big|_{O^M_{ij}} = h_{ij}$, since their Cauchy data coincide. This proves the claim:
 \begin{gather*}
  \big(\widehat{S}^{O_i}+h_i\big)\big|_{O^M_{ij}} = \big(\widehat{S}^{O_i}+h_j - h_{ij}\big)\big|_{O^M_{ij}} = \big(\widehat{S}^{O_j}+h_j\big)\big|_{O^M_{ij}}. \qedhere
 \end{gather*}
\end{proof} 

For a partition of unity $\{\chi_i\}_{i\in I}$ subordinate to $\O^M$, a well-defined bisolution on $N\times M$ is given via
\begin{align} \label{BisolutionN}
 \widehat{S}^N[\psi,\f] := \sum_{i\in I} \big(\widehat{S}^{O_i} + h_i\big)[\chi_i\psi,\f], \qquad \psi\in\D(N,E),~\f\in\D(M,E^*).
\end{align}

Since $\O^M$ a locally finite cover, for each $\psi$, only finitely many summands are non-zero. Moreover, due to Lemma (\ref{LemCech}), this definition does not depend on the choice of the partition, and for all $i$, we directly read off from (\ref{BisolutionN}) that
\begin{align} \label{CompN}
 \widehat{S}^N\big|_{O_i\times M} - \widehat{S}^{O_i} \in \CC(O_i\times M,E^*\boxtimes E).
\end{align}

Hence, two different constructions of such a bisolution on $N\times M$ differ only by a smooth bisolution.

\begin{Prop} \label{PropGlobalBisolution}
 There is a bisolution $S\colon \D(M,E) \times \D(M,E^*)\rightarrow\R$ such that 
  \begin{align} \label{CompM}
  S\big|_{O_i\times O_i}-S^{O_i} \in \CC\big(O_i\times O_i,E^*\boxtimes E\big), \qquad i\in I.
 \end{align}
\end{Prop}

\begin{proof} 
 Let $\widehat{S}^N$ be the bisolution on $N\times M$ defined by (\ref{BisolutionN}) and recall that $N$ is an open neighborhood of $\Sigma$. For all $\f\in\D(M,E^*)$, we define $S(\cdot)[\f]$ as the unique solution of
 $$\left\{\begin{array}{cl} P^t\big(S(\cdot)[\f]\big) & =0, \\[2mm] S(\cdot)[\f]\big|_\Sigma & = \widehat{S}^N(\cdot)[\f]\big|_\Sigma, \\[2mm] \nabla_\nu \big(S(\cdot)[\f]\big)\big|_\Sigma & = \nabla_\nu\widehat{S}^N(\cdot)[\f]\big|_\Sigma. \end{array} \right.$$
 This yields a smooth section, which leads to a bisolution since $\widehat{S}^N(\cdot)[P^t\f]=0$, and hence, $S(\cdot)[P^t\f]$ solves the trivial Cauchy problem. Furthermore, we have $S\big|_{O_i\times O_i} = \widehat{S}^N\big|_{O_i\times O_i}$, so (\ref{CompM}) follows from (\ref{CompN}) and Proposition \ref{PropFirstGlob}.
\end{proof}

\begin{Cor} \label{CorSymGlobalBisolution}
 There is a smooth bisolution $u\in\CC(M\times M,E^*\boxtimes E)$ such that 
 $$(S-u)[\psi_1,\Theta\psi_2] = (S-u)[\psi_2,\Theta\psi_1], \qquad \psi_1,\psi_2\in\D(M,E).$$
\end{Cor}

\begin{proof}
 For $(\iota S)[\psi_1,\Theta\psi_2] := S[\psi_2,\Theta\psi_1]$, let $u := \frac{1}{2}(S - \iota S)$. It follows that $S-u = \iota(S-u)$ and we show that $u$ is smooth. For all $m\in\N$, let $\widetilde{\L}^m$ be given as in (\ref{hAlphaBetaSmooth}), i.e.\ $\widetilde{\L}^m = \iota\widetilde{\L}^m$ due to symmetry and $\widetilde{S}^N\big|_{N_m\times N_m} - \widetilde{\L}^m$ smooth by Proposition \ref{FundSolParametrixSmooth}. Therefore, $u$ is smooth on $N_m\times N_m$ for all $m$:
 $$2u\big|_{N_m\times N_m} = \widehat{S}^{N_m} - \iota\widehat{S}^{N_m} + \widetilde{\L}^m - \widetilde{\L}^m = \widehat{S}^{N_m} - \widetilde{\L}^m - \iota\big(\widehat{S}^{N_m} - \widetilde{\L}^m\big)$$
 and thus on $N\times N$. Since $u$ is a bisolution and $N$ a neighborhood of $\Sigma$, Theorem \ref{CauchyProblemBisection} ensures smoothness on all of $M\times M$.
\end{proof}

\begin{Thm} \label{SymHadBisol}
 Let $M$ be a globally hyperbolic Lorentzian manifold, $\pi\colon E\rightarrow M$ a real vector bundle with non-degenerate inner product over $M$ and $P\colon \CC(M,E)\rightarrow\CC(M,E)$ a formally self-adjoint wave operator. Furthermore, let $G_\pm$ denote the advanced and retarded Green operator for $P^t$ and $S$ the symmetric bisolution given by Proposition \ref{PropGlobalBisolution} and Corollary \ref{CorSymGlobalBisolution}. Then
 \begin{align} \label{GlobalHadamard}
  H := S + \frac{i}{2}(G_+ - G_-)
 \end{align}
 is a Hadamard bisolution, and a Feynman and an anti-Feynman Green operator for $P^t$ is determined by
 \begin{align} \label{FeynmanSchwartzKernels}
  G_F = iS + \frac{1}{2}(G_++G_-), \qquad \qquad G_{aF} = -iS + \frac{1}{2}(G_+ + G_-).
 \end{align}
\end{Thm}

\begin{proof}
 For each $m\in\N$, let $\widetilde{\L}^m$ be given as in (\ref{hAlphaBetaSmooth}). It follows that $\WF\big(\widetilde{\L}^m\big) = \WF(\widetilde{G}_{aF} - \widetilde{G}_F)$ in the sense of (\ref{WFDistParametrices}) from Proposition \ref{FundSolParametrixSmooth}, and moreover, we have
 $$S\big|_{N_m\times N_m} - \widetilde{\L}^m \in \CC(N_m\times N_m, E^*\boxtimes E)$$
 by Propositions \ref{ExHadBisol} and \ref{PropFirstGlob} as well as (\ref{CompN}). This holds for all $m$ and hence, $H$ is of Hadamard form in a causal normal neighborhood $N$ of $\Sigma$ due to Proposition \ref{HadamardStructure}. Therefore, $H$ is globally Hadamard by Theorem 5.8 of \cite{SV2001} or, to be more precise, by (i) of the subsequent Remark.\\ 
 For (\ref{FeynmanSchwartzKernels}) note that the proof in the scalar case, given by Theorem 5.1 of \cite{R1996} and section 6.6 of \cite{DH1972}, exclusively employs the singularity structure of the parametrices. In our case, this is still determined by scalar distributions $L^\Omega_\pm, R^\Omega_\pm$ and thus stays unaffected when multiplying with smooth Hadamard coefficients. Hence, a Feynman and an anti-Feynman parametrix for $P^t$ are given by
 $$\pm iH + G_\pm = \pm i S +\frac{1}{2}(G_++G_-),$$ 
 which are even Green operators since $G_\pm$ are and $S$ is a bisolution.
\end{proof}

\subsection{Positivity}

It remains to show that $S$ can be chosen as a positive bisolution meaning that there is some smooth, symmetric bisolution $u$ such that $(S+u)[\f,\Theta\f]\geq0$ for all $\f\in\D(M,E)$. For this, we need a bundle-valued version of Theorem 6.6.2 of \cite{DH1972}, so we restrict to a certain class of operators, for which this result holds.\\
Let $M$ be a smooth manifold, $\pi\colon E\rightarrow M$ a real or complex vector bundle over $M$ with non-degenerate inner product and $P\colon \CC(M,E)\rightarrow\CC(M,E)$ a properly supported pseudodifferential operator. For the definitions of $P$ being of real principal type in $M$, pseudo-convexity of $M$ with respect to $P$ and the bicharacteristic relation $C_P$ of $P$, we adopt Definition 3.1 of \cite{D1982} as well as Definition 6.3.2 and (6.5.2) of \cite{DH1972}, respectively. Assuming those properties for $M$ and $P$, according to Theorem 6.5.3 of \cite{DH1972}, there are distinguished parametrices $\widetilde{Q}_{C_P\backslash\Delta},\widetilde{Q}_\emptyset$ associated to the respective components of $C_P\backslash\Delta$, where $\Delta$ denotes the diagonal in $\Char P \times \Char P$. For $P$ a wave operator, they correspond to Feynman and anti-Feynman parametrices, respectively.

\begin{Def} \label{PosPropType}
 Let $M$ be a smooth manifold, $\pi\colon E\rightarrow M$ a real or complex vector bundle with non-degenerate inner product and $P\colon \CC(M,E)\rightarrow\CC(M,E)$ a formally self-adjoint, properly supported pseudodifferential operator of real principal type in $M$ such that $M$ is pseudo-convex with respect to $P$. Then $P$ is called of positive propagator type if there exists some $f\in\CC(M\times M, E^*\boxtimes E)$ such that the bidistribution $T:=\frac{i}{2}\big(\widetilde{Q}_{C_P\backslash\Delta} - \widetilde{Q}_\emptyset\big)+f$ satisfies
 $$T[\psi,\Theta\psi]\geq 0, \qquad \psi\in\D(M,E).$$
\end{Def}

Note that $f$ is not demanded to be unique and in general, a positive propagator type operator will have many such sections. Theorem 6.6.2 of \cite{DH1972} states that every such $P$ is of positive propagator type for $E$ the trivial line bundle $M\times\R$. Note that the proof of this theorem employs positivity of $\frac{i}{2}(\widetilde{G}_{aF} - \widetilde{G}_F)$ for the directional derivatives $D_n:=-i\pd{}{x_n},~n=0,\hdots,d-1,$ on $\CC(\R^d)$, and by applying certain operators, allowing one to keep track of the singularity structure of the corresponding parametrices, the general case is reduced to $D_n$. Eventually, positivity holds up to smooth functions since there is no way to control this smooth part in terms of the singularity structure. However, in the setting of Definition \ref{GlobHypGreen} with $E$ assumed to be Riemannian, we can choose the same ansatz and basically the same procedure. This strongly suggests the assumption that wave operators acting on smooth sections in some general Riemannian vector bundle over a globally hyperbolic Lorentzian manifold are of positive propagator type. On the other hand, by Proposition 5.6 of \cite{SV2001}, the Hadamard bisolutions fail to be positive if the inner product on $E$ is not positive definite. Hence, anticipating the result of this section, wave operators acting on sections in a non-Riemannian vector bundle over a globally hyperbolic Lorentzian manifold are not of positive propagator type.\\
Assuming $P$ to be of positive propagator type, we need to show that $f$ can be actually chosen as a symmetric bisolution. It turns out that the existence of a pair $G_F,G_{aF}$, a well-posed Cauchy problem and
\begin{align} \label{CharBiCharP}
 \begin{split}
  \Char(P) & = \big\{(p,\xi)\in T^*M\backslash\{0\}~\big|~g_p(\xi^\sharp,\xi^\sharp)=0\}, \\[2mm]
  C_P      & = \big\{(p,\xi;q,\eta)\in \big(T^*M\times T^*M\big)\backslash\{0\}~\big|~(p,\xi)\sim(q,\zeta)\big\}.
 \end{split}
\end{align}

is sufficient. For $\Sigma$ some Cauchy hypersurface of $M$, the idea is to use $f$ as initial data on $\Sigma\times\Sigma$ in order to determine a smooth bisolution $u$ via Theorem \ref{CauchyProblemBisection} and then following the lines of section 3.3 of \cite{GW2015}.\\
Let $\iota\colon  \Sigma \hookrightarrow M$ be the embedding map and $\rho:= \big(\iota^*, \iota^* \circ \nabla_\nu\big)$ the corresponding pullback to the initial data on $\Sigma$, i.e.\
\begin{align} \label{Rho}
 \rho\colon \qquad \CC(M,E) \rightarrow \CC(\Sigma,E\oplus E), \qquad u\longmapsto\big(u\big|_\Sigma, \nabla_\nu u\big|_\Sigma\big).
\end{align}

Clearly, $\rho$ is surjective and we have $\rho\big(\CCsc(M,E)\big)=\D(M,E\oplus E)$. Furthermore, for any differential operator $P$ with well-posed Cauchy problem, $\rho$ yields a bijection $\ker P \rightarrow \CC(\Sigma,E\oplus E)$. The transposed map $\rho^t$ is related to the pushforward along the embedding, which creates singular directions orthogonal to the embedded (spacelike) hypersurface. More precisely, according to Proposition 10.21 of \cite{DK2010}, $\iota_*\f$ corresponds to $\f\delta_\Sigma$ for any \mbox{$\f\in\CC(\Sigma,E)$}, and hence, $\rho^t$ is a map 
\begin{align*} 
 \rho^t\colon \qquad \CC(\Sigma,E^*\oplus E^*) \longrightarrow \D_{N^*\Sigma}(M,E^*)'.
\end{align*}

$\D_\Gamma'$ denotes the distributions with wave front set contained in the closed cone $\Gamma\subset T^*M\backslash\{0\}$, and we refer to section 8.2 of \cite{H1990} for precise definitions and properties of these spaces. Due to H\"ormander's criterion \big((8.2.3) of \cite{H1990}\big), we can pull back a distribution along $\iota$ if its wave front set does not contain the orthogonal directions mentioned above. Hence, for all closed cones $\Gamma\subset T^*M\backslash\{0\}$ with $\Gamma \cap N^*\Sigma=\emptyset$, (\ref{Rho}) extends to a map
\begin{align*} 
 \rho\colon \qquad \D_\Gamma(M,E^*)' \longrightarrow \D_{\iota^*\Gamma}(\Sigma,E^*\oplus E^*)', \qquad u \longmapsto (\chi\mapsto u[\rho^t\chi]),
\end{align*}

where $\iota^*\Gamma := \big\{\big(\sigma,\df\iota|_\sigma^t(\xi)\big)~\big|~\big(\iota(\sigma),\xi\big)\in\Gamma\big\}\subset T^*\Sigma\backslash\{0\}$ contains the projections of $\xi\in\Gamma$ onto $T^*\Sigma$. Let 
\begin{align} \label{L2ProductSigmaE}
 (\chi,\zeta)_\Sigma := \int_\Sigma \big(\braket{\chi_0}{\zeta_0} + \braket{\chi_1}{\zeta_1}\big) \df V_\Sigma, \qquad \chi,\zeta\in\D(\Sigma,E\oplus E),
\end{align}

denote the inner product on $\D(\Sigma,E\oplus E)$ with $\df V_\Sigma$ the induced volume density and $\widetilde{\Theta}:=(\Theta,\Theta)$ the corresponding isomorphism $E\oplus E \rightarrow E^*\oplus E^*$. For $P$ Green-hyperbolic, the exact sequence (\ref{ExactSequence}) provides $\ran G = \ker P\big|_{\CCsc}$ and thus a further bijection $\rho G\colon  \D(M,E)\slash\ker G \rightarrow \D(\Sigma,E\oplus E)$, which transfers $G$ to a Green operator $G_\Sigma$ on the space of initial data  $\D(\Sigma,E\oplus E)$ via
\begin{align} \label{GreenSigma}
 (\rho G\psi_1, G_\Sigma\rho G\psi_2)_\Sigma := (\psi_1,G\psi_2)_M, \qquad \psi_1,\psi_2\in\D(M,E).
\end{align}

It is not hard to deduce the Cauchy evolution operator $U_\Sigma := -G\rho^*G_\Sigma$ mapping initial data $(u_0,u_1)\in\CC(\Sigma,E\oplus E)$ to the solution $u\in\CC(M,E)$ of the corresponding homogeneous Cauchy problem. Moreover, \cite{D1980} and, for the vector-valued case, \cite{BS2019} provide the expression
$u=G^*\big(\rho^*_1 u_0 - \rho^*_0 u_1\big)$, from which uniqueness and surjectivity of $\rho G\colon \D(M,E) \rightarrow \D(\Sigma,E\oplus E)$ lead to the particular expression $G_\Sigma(u_0,u_1) = (-u_1,u_0)$.

\begin{Thm} \label{PositivityThm}
 Let $M$ be a globally hyperbolic Lorentzian manifold, $\pi\colon E\rightarrow M$ a Riemannian vector bundle and $P\colon \CC(M,E) \rightarrow \CC(M,E)$ a linear first- or second-order differential operator, which is of positive propagator type and admits a well-posed Cauchy problem. Assume that the characteristic set and the bicharacteristic relation of $P$ are given by (\ref{CharBiCharP}) and that $\widetilde{Q}_{C_P\backslash\Delta},\widetilde{Q}_\emptyset$ can be chosen as actual Green operators $Q_{C_P\backslash\Delta},Q_\emptyset$.\\ 
 Then there is a real-valued and symmetric bisolution $S$ such that $S-\frac{i}{2}\big(Q_{C_P\backslash\Delta} - Q_\emptyset\big)$ is smooth and
 \begin{align*} 
  S[\psi,\Theta\psi]\geq0, \qquad \psi\in\D(M,E).
 \end{align*}
\end{Thm}

\begin{proof}
 The desired real-valued bisolution is given by
 \begin{align} \label{S}
  S[\psi,\f]:=\frac{i}{4}\big(Q_{C_P\backslash\Delta} - Q_\emptyset\big)[\psi,\f] + \frac{i}{4}\overline{\big(Q_{C_P\backslash\Delta} - Q_\emptyset\big)[\psi,\f]}, \qquad \psi,\in\D(M,E),~\f\in\D(M,E^*),
 \end{align}
 and we show the claimed properties. With regard to Corollary \ref{CorSymGlobalBisolution} and without loss of generality, we assume $S$ to be symmetric, and furthermore, there is some $f\in\CC(M\times M, E^*\boxtimes E)$ such that $(S + f)[\psi,\Theta\psi]\geq0, ~\psi\in\D(M,E),$ by assumption on $P$. Because this is also true for $\widetilde{f}[\psi_1,\Theta\psi_2] := \frac{1}{2}\big(f[\psi_1,\Theta\psi_2] + f[\psi_2,\Theta\psi_1]\big)$, we assume symmetry of $f$ as well, that is, $f[\psi_1,\Theta\psi_2] = f[\psi_2,\Theta\psi_1]$ for all $\psi_1,\psi_2\in\D(M,E)$.\\
 Recall that Green operators map $\D(M,E^*)$ to $\CC(M,E^*)$, so for fixed $\f\in\CC(M,E^*)$, (\ref{S}) provides a smooth section \mbox{$p\mapsto S(p)[\f]$} in $E^*$. It follows that for each $p\in M$, we obtain a well-defined $E_p^*$-valued distribution $S(p)$, which solves $PS(p)=0$, and hence, $\WF\big(S(p)\big)$ exclusively contains lightlike directions by assumption on $P$. By $\WF\big(S(p)\big)\cap N^*\Sigma = \emptyset$, the restriction of $S(p)$ to $\Sigma$ yields a well-defined distribution $\rho\big(S(p)\big)$ on $\D(\Sigma,E^*\oplus E^*)$ for any Cauchy hypersurface $\Sigma$. This means that, due to Theorem 8.2.13 of \cite{H1990}, the operator $\S$ associated to (\ref{S}) can be applied to $\rho^t\chi \in \D_{N^*\Sigma}(M,E^*)'$, $\chi\in\D(\Sigma,E^*\oplus E^*)$, and for the result, we obtain
 $$\WF\big(\S\rho^t\chi\big) \subset \big\{(p,\xi)~\big|~(p,\xi;q,0)\in\WF(S)\big\} \cup \big\{(p,\xi)~\big|~\exists (q,\zeta)\in\WF(\rho^t\chi):\enspace (p,\xi;q,-\zeta)\in\WF(S)\big\}.$$
 Since $\WF(S) \subset C_P = \big\{(p,\xi)\sim(q,\zeta)\big\}$ and $\WF(\rho^t\chi)\subset N^*\Sigma$, both contributions on the right hand side are empty. Hence, $\S\rho^t$ represents a map $\D(\Sigma,E^*\oplus E^*) \rightarrow\CC(M,E^*)$, so it follows that \mbox{$p\mapsto \rho\big(S(p)\big)[\chi] = \big(\S\rho^t\chi\big)(p)$} is smooth for fixed $\chi$. With the adjoint operator $\rho^* = \Theta^{-1}\rho^t\widetilde{\Theta}$, we eventually obtain a well-defined bidistribution $S^\Sigma\colon \D(\Sigma,E\oplus E)\times\D(\Sigma,E^*\oplus E^*) \rightarrow \R$ via
 \begin{align} \label{InducedBidistributionSigma}
  S^\Sigma[\lambda,\chi] := S\big[\rho^*\lambda,\rho^t\chi\big] = \int_\Sigma \widetilde{\Theta}\rho\Theta^{-1}\Big(S(\cdot)\big[\rho^t\chi\big]\Big)(\sigma) ~ \big(\lambda(\sigma)\big) \df V_\Sigma(\sigma).
 \end{align}
 The bisection $f$ determines smooth and symmetric Cauchy data on $\Sigma\times\Sigma$ and thus a smooth and symmetric bisolution $u$ by Theorem \ref{CauchyProblemBisection} (the first order analogon works completely similarly). Using the short-hand notation $S_f := S+f$ and $S_u := S+u$, this yields $S^\Sigma_u=S^\Sigma_f$ for the corresponding bidistributions (\ref{InducedBidistributionSigma}), and we show that positivity is preserved under the restriction to $\Sigma$, i.e.\ $S^\Sigma_f[\lambda,\widetilde{\Theta}\lambda]\geq 0$ for all $\lambda\in\D(\Sigma,E\oplus E)$. Theorem 8.2.3 of \cite{H1990} provides a sequence $(\psi_n)_{n\in\N}\subset\D(M,E)$ such that $\psi_n \rightarrow \rho^*\lambda$ in $\D_{N^*\Sigma}(M,E)'$, and consequently, $\Theta\psi_n \rightarrow \Theta\rho^*\lambda = \rho^t\widetilde{\Theta}\lambda$, so continuity of $S_f$ ensures
 \begin{align} \label{RestrPositive}
  S^\Sigma_f[\lambda,\widetilde{\Theta}\lambda] = S_f\big[\rho^*\lambda, \rho^t\widetilde{\Theta}\lambda\big] = \lm{n}{\infty} \underbrace{S_f\big[\psi_n,\Theta\psi_n\big]}_{\geq 0} \geq 0.
 \end{align}
 The proof of Theorem 3.3.1 and Proposition 3.4.2 of \cite{BGP2007} show that well-posedness of the Cauchy problem implies the existence of a unique advanced and retarded Green operator and hence exactness of the sequence (\ref{ExactSequence}). Thus, due to $\ker P = \ran G$, $S_u$ does not only descend to a well-defined bilinear form on $\D(M,E)\slash \ker P$ by being a bisolution, but also to $\ran G$ via
 $$S'_u[G\psi_1,\Theta G\psi_2] := S_u[\psi_1,\Theta\psi_2], \qquad \psi_1,\psi_2\in\D(M,E).$$
 By following the lines of Proposition 3.9 of \cite{GW2015} and employing $G=-G\rho^*G_\Sigma\rho G$, this allows us to trace back the claimed positivity property to (\ref{RestrPositive}). More precisely, for all $\psi_1,\psi_2\in\D(M,E)$, we have
 \begin{align*}
  S_u[\psi_1,\Theta\psi_2] & = S'_u[G\psi_1,\Theta G\psi_2] = S'_u[G\rho^*G_\Sigma\rho G\psi_1,\Theta G\rho^*G_\Sigma\rho G\psi_2] \\[2mm] 
                        & = S_u[\rho^*G_\Sigma\rho G\psi_1, \Theta \rho^*G_\Sigma\rho G\psi_2] = S^\Sigma_u[G_\Sigma\rho G\psi_1,\widetilde{\Theta}G_\Sigma\rho G\psi_2] = S^\Sigma_f[G_\Sigma\rho G\psi_1,\widetilde{\Theta}G_\Sigma\rho G\psi_2],
 \end{align*}
 which proves the theorem.
\end{proof}

In the case of formally self-adjoint wave operators, the existence of $G_F$ and $G_{aF}$ is ensured by Theorem \ref{SymHadBisol}, so Theorem \ref{PositivityThm} leads to the final result:

\begin{Thm} \label{MainTheorem}
 Let $M$ be a globally hyperbolic Lorentzian manifold, $\pi\colon E\rightarrow M$ a Riemannian vector bundle and $P\colon \CC(M,E) \rightarrow \CC(M,E)$ a formally self-adjoint wave operator of positive propagator type. Then there exists a bidistribution $S\colon \D(M,E)\times\D(M,E^*) \rightarrow \R$ such that
 \begin{align*} 
  H := S + \frac{i}{2}(G_+-G_-)
 \end{align*}
 yields a Hadamard two-point-function, where $G_\pm$ denotes the advanced and retarded Green operator for $P^t$. This means that $\WF(H)$ has the Hadamard singularity structure (\ref{HConditionWF}) and satisfies
 \begin{align*}
  H[P\psi,\f]=0=H[\psi,P^t\f], \qquad H[\psi,\f]-H[\Theta^{-1}\f,\Theta\psi] = \frac{i}{2}(G_+-G_-)[\psi,\f], \qquad H[\psi,\Theta\psi]\geq 0
 \end{align*}
 for all $\psi\in\D(M,E), ~\f\in\D(M,E^*)$.\\[2mm] 
 Moreover, a Feynman and an anti-Feynman Green operator $G_F,G_{aF}$ are given by (\ref{FeynmanSchwartzKernels}).
\end{Thm}

Note that, in general, $S$ is far from being unique, i.e. there may be many bidistributions with the required properties. Clearly, this is related to the non-uniqueness of the many choices of smooth sections during the construction, and in most cases, it is not at all obvious, how to find these sections practically. This particularly concerns the choice of the $h_i$'s in Lemma \ref{LemCech} and the $f$ for operators of positive propagator type.\\ 
However, the overall reasoning provides a comparatively constructive alternative to the existence proofs, which are already present in the literature (\cite{BF2014}, \cite{FNW1981}, \cite{GOW2017}). It starts most naturally with the Hadamard condition, so the form of the bidistributions is, up to smooth terms, determined right from the start. It therefore might provide a promising starting point for a possible classification of these states up to unitary equivalence of their respective GNS-representations. This and the identification of pure states in particular would require to investigate the choices of the said smooth sections. \\
Furthermore, the methods used here provide an alternative procedure to the classic deformation arguments since they rely on the ability to make modifications to the metric confined to certain spacetime regions. There are situations, where this is not applicable, for instance, in the case of linearized gravity, where the background spacetime must solve the Einstein equation, or similarly for linearizations of Yang-Mills theories. They also occur if one is restricted to analytic metrics.

\section*{Appendix}
\addcontentsline{toc}{section}{Appendix}
\renewcommand{\thesubsection}{\Alph{subsection}}
\renewcommand{\theequation}{\thesubsection.\arabic{equation}}
\renewcommand{\theDef}{\thesubsection.\arabic{Def}}
\setcounter{subsection}{0}
\setcounter{equation}{0}

\setcounter{Def}{0}

\subsection{Symmetry of the Hadamard coefficients}

In this section, we prove the symmetry of the Hadamard coefficients $U_k$ in a setting $(M,E,P)$ as in Definition \ref{GlobHypGreen} with $P$ a wave operator. This represents an alternative path to the more general case treated in \cite{K2019}, which is more geared to the specific situation of wave operators, and generalizes the well-known approach of Moretti \cite{M1999, M2000} it to sections in $E$.\\
Let $\Omega\subset M$ be a non-empty and convex domain, that is time-orientable, and let $\nabla$ denote the $P$-compatible connection on $E$, meaning that
\begin{align} \label{PCompConn}
 \nabla_{\grad f} s = \frac{1}{2}\big(f\cdot Ps - P(f\cdot s) + \Box f \cdot s\big), \qquad s\in\CC(M,E), ~ f\in\CC(M).
\end{align}

It follows that $P=\Box^\nabla+B$ for some uniquely determined endomorphism field $B$ and \mbox{$\Box^\nabla=\left(\mathrm{tr} \otimes \id_E\right) \circ \nabla^{T^*M \otimes E} \circ \nabla$} the connection-d'Alembert operator (see section 1.5 of \cite{BGP2007}). Then the Hadamard coefficients \mbox{$U_k\in\CC(\Omega\times\Omega,E^*\boxtimes E),~k\in\N_0,$} for $P$ are defined as the unique solutions of the transport equations
\begin{align} \label{TransportEquations}
 \nabla_{\grad\Gamma_p} U_p^k - \left(\frac{1}{2}\Box\Gamma_p - d + 2k\right)U_p^k = 2k\cdot PU_p^{k-1}, \qquad k\in \N,
\end{align}

with $U_0(p,p)=\id_{E_p^*}$ and $U^k_p := U_k(p,\cdot)$ for all $p\in\Omega$ (Proposition 2.3.1 of \cite{BGP2007}). Recall the identification (\ref{FibersBoxBundle}) of $U_k(p,q)$ as homomorphisms $E_q^* \rightarrow E_p^*$ with fiberwise transposed operator \mbox{$U_k(p,q)^t\in\Hom(E_p,E_q)$}. We are going to show symmetry in the sense
\begin{align} \label{SymmetricHCoeff}
 U_k(p,q) = \Theta_p U_k(q,p)^t \Theta_q^{-1}, \qquad p,q\in \Omega,~k\in\N_0.
\end{align}

One can see the Hadamard coefficients as a measure of the deviation of $(M,g,E,P)$ from $(\R^d_{\mathrm{Mink}},\Box)$, and indeed, "adding" $\big(\R, \df s^2, \{0\}, \pd{^2}{s^2}\big)$ does not change them: Let $(\widehat{M},\widehat{g}):=(M\times \R, g+\df s^2\big)$, over which we consider the same vector bundle $E$, and $\widehat{P}:=P - \pd{^2}{s^2}$. For the corresponding Hadamard coefficients $\widehat{U}_k$, we obtain the initial condition $\widehat{U}_0\big((p,s),(p,s)\big)=\id_{E_p^*}$ and the transport equations
\begin{align*}
  & 0 = \widehat{\nabla}_{\grad_{\widehat{g}}\widehat{\Gamma}_{(p,s)}} \widehat{U}_{(p,s)}^k - \left(\frac{1}{2}\Box_{\widehat{g}}\widehat{\Gamma}_{(p,s)} - (d+1) + 2k\right)\widehat{U}_{(p,s)}^k - 2k\widehat{P}\widehat{U}_{(p,s)}^{k-1} \\
  & = \nabla_{\grad_g\Gamma_p} \widehat{U}_{(p,s)}^k + 2(s'-s) \widehat{\nabla}_{\pd{}{s}}\widehat{U}_{(p,s)}^k - \left(\frac{1}{2}\big(\Box_g \Gamma_p + 2\big) - (d+1) +2k\right)\widehat{U}_{(p,s)}^k - 2kP\widehat{U}_{(p,s)}^{k-1} + 2k \pd{^2}{s^2}\widehat{U}_{(p,s)}^{k-1} \\
  & = \nabla_{\grad_g\Gamma_p} \widehat{U}_{(p,s)}^k - \left(\frac{1}{2}\Box_g \Gamma_p - d +2k\right)\widehat{U}_{(p,s)}^k - 2kP\widehat{U}_{(p,s)}^{k-1} + 2(s'-s) \widehat{\nabla}_{\pd{}{s}}\widehat{U}_{(p,s)}^k + 2k \pd{^2}{s^2}\widehat{U}_{(p,s)}^{k-1},
\end{align*}%
which are clearly solved by $U_p^k$ leading to
\begin{align} \label{HCoeffProduct}
 \widehat{U}_k\big((p,s),(q,s')\big) = U_k(p,q), \qquad p,q\in\Omega, ~ s,s'\in\R, ~ k\in\N_0.
\end{align}

We adopt Moretti's approach insofar as we start by considering wave operators with analytical coefficients, so in particular the metric as the principal symbol is assumed to be analytic, and deduce analyticity of $(p,q)\mapsto U_k(p,q)$. We directly conclude that, as a function of $g$, the Levi-Civita connection on $TM$ and, due to (\ref{PCompConn}), the $P$-compatible connection $\nabla$ on $E$ are analytic as well, so the corresponding Christoffel symbols are. Applying basic ODE-theory and the analytic inverse function theorem (Theorem 1.4.3 of \cite{KP1992}) ensure analyticity of the Lorentzian distance $(p,q) \mapsto \Gamma(p,q) = g_p\big(\exp_p^{-1}(q)\big)$, the distortion function $(p,q)\mapsto\mu(p,q)$ and the geodesic considered as a map
\begin{align} \label{GeodesicMap}
 \phi: \qquad [0,1]\times \Omega \times \Omega \longrightarrow \Omega, \qquad (t,p,q) \longmapsto \exp_p\big(t\,\exp_p^{-1}(q)\big) =: \phi_{pq}(t).
\end{align}

\begin{Lem}
 The $\nabla$-parallel transport along $\phi_{pq}$ is analytic as a map
 \begin{align} \label{ParallelTransportAnalytic}
  [0,1] \times \Omega\times \Omega \longrightarrow E^*\boxtimes E, \qquad (t,p,q) \longmapsto \Pi_{\phi_{pq}(t)}^p.
 \end{align}
\end{Lem}

\begin{proof}
 For fixed $p\in \Omega$ and $e\in E_p$, consider the parallel section $s_p(t,q) := \Pi^p_{\phi_{pq}(t)}e$ in $E$ along $\phi_{pq}$ for all $q\in\Omega$, which therefore satisfies the ODE's
\begin{align} \label{ParallelTransportODE}
 \dot{s}_p^\beta(t,q) = \underbrace{-\Gamma^\beta_{i\alpha}\big(\phi_{pq}(t)\big) \, \dot{\phi}^i_{pq}(t)}_{=:A_p(t,q)^\beta_\alpha} \, s_p^\alpha(t,q), \qquad s_p(0,q) = e.
\end{align}%

The columns of the corresponding fundamental matrix $\Phi_p(t,q)$ are given by $\mathrm{rk}(E)$ linearly independent solutions of (\ref{ParallelTransportODE}), so we have $\dot{\Phi}_p(t,q) = A_p(t,q)\Phi_p(t,q)$ and the solution of (\ref{ParallelTransportODE}) takes the form 
$$s_p(t,q) = \Phi_p(t,q)\Phi_p(0,q)^{-1}s_p(0,q).$$ 

From the definition of $s_p$, we read off $\Pi^p_{\phi_{pq}(t)} = \Phi_p(t,q)\Phi_p(0,q)^{-1}$, and hence, the map $(t,q) \longmapsto \Pi^p_{\phi_{pq}(t)}$ is analytic. Moreover, $\Pi^r_p = \big(\Pi_r^p\big)^{-1}$ implies
$$\Phi_p(1,r)\Phi_p(0,r)^{-1} = \left(\Phi_r(1,p)\Phi_r(0,p)^{-1}\right)^{-1} = \Phi_r(0,p)\Phi_r(1,p)^{-1},$$%
so $p\mapsto\Pi^p_r$ is analytic for fixed $r\in\Omega$ and Osgood's Lemma \cite{O1898} proves the claim.
\end{proof}

\begin{Prop} \label{AnalyticHadamardCoefficients}
 The map $(p,q)\mapsto U_k(p,q)$ is analytic on $\Omega\times\Omega$ for all $k\in \N_0$.
\end{Prop}

\begin{proof}
 Analyticity of the zeroth Hadamard coefficient can be directly read off from
 $$(p,q) \longmapsto U_0(p,q) = \frac{\Pi^p_q}{\sqrt{\mu(p,q)}},$$%
 and we proceed via induction. By analyticity of $P$, clearly $(p,q) \mapsto P_{(2)}U_{k-1}\big(p,\Phi_{pq}(t)\big)$ is analytic if the prior coefficient $(p,q) \mapsto U_{k-1}(p,q)$ is.
 Similarly, $(t,p,q)\mapsto U_0\big(p,\phi_{pq}(t)\big)^{-1} = \sqrt{\mu\big(p,\phi_{pq}(t)\big)} \cdot \Pi_p^{\phi_{pq}(t)}$ is analytic as a composition of analytic maps (recall that $\mu$ is positive). Therefore, the integrand of
 $$U_k(p,q) = -kU_0(p,q) \int_0^1 t^{k-1} \cdot U_0\big(p,\phi_{pq}(t)\big)^{-1} P_{(2)}U_{k-1}\big(p,\phi_{pq}(t)\big)\df t$$
 is analytic in $(p,q)$ and uniformly continuous in $t$ on $[0,1]$. Hence, taking the power series expression of the integrand, the sum and the integral can be swapped, which results in a uniformly converging power series for $U_k$.
\end{proof}

Now the general case of smooth $P$ is tackled by analytic approximation of the coefficients, for which we quote

\begin{Prop}[Proposition 2.1 of \cite{M1999}] \label{MetricApproxProp}
 Let $M$ be a real, smooth and connected manifold with non-singular metric $g$.
 \begin{itemize}
  \item[(a)] For any local chart $(x,V)$ of $M$ and any relatively compact domain $O$ with $\Obar \subset V$, there is a sequence $\{g^n\}_{n\in\N}$ of real and analytic (with respect to $x$) metrics with the same signature as $g$, which are defined on some neighborhood of $\overline{O}$ such that $g^n \rightarrow g$ in $\CC$, that is, all derivatives of $g^n$ converge uniformly on $\overline{O}$:
  $$\forall i,j=1,\hdots,D, \quad \alpha \in \N_0^D: \qquad \max_{v \in x(\overline{O})} \Big|\big(D^\alpha (g^n\circ x^{-1})_{ij}\big)(v) - \big(D^\alpha (g\circ x^{-1})_{ij}\big)(v)\Big| \longrightarrow 0.$$%
  \item[(b)] For any $(x,V), O, \{g^n\}_{n\in\N}$ as in (a) and additionally any $z \in O$, there is an $n_0\in\N$ and a family $\{N_z^i\}_{i\in\R}$ of open neighborhoods of $z$ such that $N_z^i \subset \overline{N}_z^j \subset O$ for any $j>i$, and $\{N_z^i\}_{i\in\R}$ is a local base of the topology of $M$. Moreover, for all $i \in \R$, both $N_z^i$ and $\overline{N}_z^i$ are common convex neighborhoods of $z$ for all metrics $\{g^n\}_{n>n_0}$ and $g$.
 \end{itemize}
\end{Prop}

\begin{Prop} \label{AnalytAppr}
 Let $O\subset\Omega$ be relatively compact and $\{g^n\}_{n\in\N}$ a sequence of real and analytic metrics defined in a neighborhood of $\Obar$ with the same signature as $g$ such that $O$ and $\Obar$ are convex with respect to all $g^n,~n\in\N,$ and $g$ and $g^n \rightarrow g$ in $\CC$. For $\{U_k^n\}_{n\in\N}$ the corresponding Hadamard coefficients, we obtain $U_k^n(p,q) \rightarrow U_k(p,q)$ for all $k\in\N_0$ and $p,q\in O$.
\end{Prop}

\begin{proof}
 The assumption directly provides $\Gamma^{k,n}_{ij} \rightarrow \Gamma^k_{ij}$, and with regard to the geodesic equation with converging right hand side, we similarly obtain $\exp^n \rightarrow \exp$ as smooth maps $(t,\xi,p) \mapsto  \exp_p(t\xi)$ on their domain of existence. Then the inverse function theorem provides $\big(\exp^n\big)^{-1} \rightarrow \exp^{-1}$ as smooth maps on $\Obar\times \Obar$ and, as a consequence, of the Lorentzian distance $\Gamma^n\rightarrow\Gamma$ and the distortion function $\mu^n\rightarrow\mu$. Eventually, we have $\phi^n \rightarrow \phi$ for the connecting geodesic (\ref{GeodesicMap}). \\
 It remains to investigate the parallel transport. For all $p\in O$, convergence of $\Gamma^{k,n}_{ij}$ and $\phi^n$ leads to $A_p^n \rightarrow A_p$ for the matrices defined in (\ref{ParallelTransportODE}), and hence,
 $$\Pi^{p,n}_{\phi^n_{pq}(t)} = \Phi^n_p(t,q)\Phi^n_p(0,q)^{-1} \longrightarrow \Phi_p(t,q)\Phi_p(0,q)^{-1} = \Pi^p_{\phi_{pq}(t)}$$
 
 as smooth maps $[0,1]\times\Obar \rightarrow E_p^*\otimes E$. Thus, we can directly conclude convergence of the zeroth Hadamard coefficient
 \begin{align} \label{SymmZeroHadCoeff}
  U_0^n(p,\cdot) = \frac{\Pi^{p,n}_\cdot}{\sqrt{\mu^n_p}} \longrightarrow \frac{\Pi^p_\cdot}{\sqrt{\mu_p}} = U_0(p,\cdot)
 \end{align}
 
 as smooth maps $\Obar \rightarrow E^*_p\otimes E$ and, in particular, $U_0^n(p,q) \rightarrow U_0(p,q)$ in $\Hom(E_q^*,E_p^*)$. \\  
 We proceed inductively. Due to $\phi_{pq}^n \rightarrow \phi_{pq}$ in $\CC\big([0,1],\Obar\big)$, (\ref{SymmZeroHadCoeff}) implies $U_0^n(p,\cdot) \circ \phi^n_{pq} \rightarrow U_0(p,\cdot) \circ \phi_{pq}$ in $\CC\big([0,1],E^*_p\otimes E\big)$ and consequently, $PU_0^n(p,\cdot) \circ \phi^n_{pq} \rightarrow PU_0(p,\cdot) \circ \phi_{pq}$. Therefore, the integrand in the expression of the first Hadamard coefficient
 $$U_1^n(p,q) = -kU^n_0(p,q) \int_0^1 U_0^n\big(p,\phi_{pq}^n(t)\big)^{-1} P_{(2)}U_0^n\big(p,\phi_{pq}^n(t)\big) \df t$$%
 converges to the one in the expression of $U_1(p,q)$, and as a smooth function in $t$, it is integrable on the compact interval $[0,1]$. Hence, due to majorized convergence, the integral converges as well, so we have
 $$-U^n_0(p,q) \int_0^1 U_0^n\big(p,\phi_{pq}^n(t)\big)^{-1} P_{(2)}U_0^n\big(p,\phi_{pq}^n(t)\big) \df t \longrightarrow -U_0(p,q) \int_0^1 U_0\big(p,\phi_{pq}(t)\big)^{-1} P_{(2)}U_0\big(p,\phi_{pq}(t)\big) \df t,$$%
 which is $U_1(p,q)$. Recursively, this implies $U_k^n(p,q) \rightarrow U_k(p,q)$ for all $k\in\N$ and $p,q\in O$.
\end{proof}

It remains to show symmetry (\ref{SymmetricHCoeff}) in the analytic case, for which we leave Moretti's path and present a novel approach using the construction in \cite{BGP2007} instead.\\
Let $V_k \in \CC(\Omega\times\Omega, E\boxtimes E^*),~k\in\N_0,$ denote the Hadamard coefficients associated to $(M,g,E^*,P^t)$. Employing formal self-adjointness of $P$ in the transport equations provides the relation
\begin{align} \label{HadCoeffSelfAdjoint}
 V_k(p,q) = \Theta_p^{-1}U_k(p,q)\Theta_q, \qquad p,q\in \Omega.
\end{align} 

According to section 1.4 of \cite{BGP2007}, we define
\begin{align} \label{RieszDistrDomain}
 R^\Omega_\pm(\alpha,p)[\f] := R^\alpha_\pm\big[(\mu_p\f)\circ\exp_p\big],\qquad \f\in\D(\Omega),
\end{align}

which ensures the identification $R^\Omega_\pm(\alpha,p)\big|_{J_\pm^\Omega(p)} = C(\alpha,d)\cdot\Gamma_p^{\frac{\alpha-d}{2}}$ for $\RT{\alpha}>d$. Due to Proposition 2.4.6 of \cite{BGP2007}, they comprise Hadamard series, which yield advanced and retarded parametrices $\widetilde{\Riesz}_\pm(p)$ for $P$ at each $p\in \Obar$ on any relatively compact domain $O\subset \Omega$. More precisely, for any integer $N>\frac{d}{2}$ and cut-off function $\sigma\in\D\big((-1,1),[0,1]\big)$ with \mbox{$\sigma\big|_{\left[-\frac{1}{2},\frac{1}{2}\right]}=1$}, there is a sequence $\{\e_k\}_{k\geq N}\subset(0,1]$ such that 
\begin{align} \label{RieszParametrix}
 \widetilde{\Riesz}_\pm(p) = \sum_{k=0}^\infty \widetilde{U}_k(p,\cdot) ~R^\Omega_\pm(2k+2,p), \qquad \widetilde{U}_k := \left\{\begin{array}{cl} U_k, \qquad & k<N, \\[2mm] \left(\sigma \circ \frac{\Gamma}{\e_k}\right) \cdot U_k, \qquad & k\geq N, \end{array}\right.
\end{align}

represent well-defined distributions and $(p,q) \mapsto \big(P\widetilde{\Riesz}_\pm(p) - \delta_p\big)(q) \in \CC(\Omega\times\Omega,E^*\boxtimes E)$. Furthermore, we have $p\mapsto\widetilde{\Riesz}_\pm(p)[\f] \in\CC(\Obar,E^*)$ for fixed $\f\in\D(O,E^*)$, so regarded as bidistributions and due to compactness of $\Obar$, they provide continuous operators
\begin{align} \label{RieszOperator}
 \widetilde{G}_\pm\colon \qquad \D(O,E^*) \rightarrow \CC(\Obar,E^*), \qquad \f\longmapsto \big(p\mapsto\widetilde{\Riesz}_\mp(p)[\f]\big).
\end{align} 

Consequently, $\widetilde{G}_\pm$ yield left parametrices for $P^t$ with $\supp\widetilde{G}_\pm\f \subset J_\pm^{\Obar}(\supp\f)$, and for $G_\pm$ the advanced and retarded Green operator for $P^t$, the differences $G_\pm - \widetilde{G}_\pm$ are smoothing (their integral kernels correspond to the last equation in the proof of Proposition 2.5.1 of \cite{BGP2007}, which is actually a smooth section). Therefore, $\widetilde{G}_\pm$ represent an advanced and a retarded parametrix for $P^t$ in the sense of Duistermaat-H{\" o}rmander.

\begin{Prop} \label{DiffRieszSmoothProp}
 For all convex and relatively compact domains $O\subset\Omega$, the maps
 \begin{align} \label{DiffRieszSmooth}
  (p,q) \longmapsto \sum_{k=0}^\infty \left(\big(\widetilde{U}_k(p,\cdot) - \widetilde{V}_k(\cdot,p)^t\big) R^\Omega_\pm(2k+2,p)\right)(q)
 \end{align} 
 define smooth sections in $E^*\boxtimes E$ over $O\times O$.
\end{Prop}

\begin{proof}
 By Lemma 3.4.4 of \cite{BGP2007}, the advanced and retarded Green operators for $P$ are given by $G^t_\mp$, so formal self-adjointness of $P$ and uniqueness of $G_\pm$ lead to $G_\pm = \Theta G_\mp^t\Theta^{-1}$ such that the operators $\widetilde{G}_\pm - \Theta\widetilde{G}^t_\mp\Theta^{-1}$ are smoothing:
 $$\widetilde{G}_\pm - \Theta\widetilde{G}^t_\mp\Theta^{-1} = \underbrace{\widetilde{G}_\pm - G_\pm}_{\text{smoothing}} - ~\Theta\big(\underbrace{\widetilde{G}_\mp - G_\mp}_{\text{smoothing}}\big)^t\Theta^{-1}.$$
 
 Recalling $\widetilde{G}^t_\pm[\f,\psi] = \widetilde{G}_\pm[\psi,\f]$, we just have to show that the Schwartz kernel of $\widetilde{G}_\pm^t$ is given by the distribution
 $$\sum_{k=0}^\infty\Theta^{-1}_p\widetilde{V}_k(\cdot,p)^t\Theta ~ R^\Omega_\pm(2j+2,p), \qquad p\in O,$$
 which can be directly deduced from Lemma 1.4.3 of \cite{BGP2007} and (\ref{HadCoeffSelfAdjoint}).
\end{proof}

\begin{Lem} \label{EqLightlikeEverywhere}
 Let $k\in\N_0$ and assume that for all quadruples $(M,g,E,P)$ as introduced in the beginning of the chapter with odd spacetime dimension and all lightlike related $p,q\in O$, we have
 $$U_k(p,q) = V_k(q,p)^t.$$ 
 Then this equality holds for all $p,q\in\Omega$ and all $(M,g,E,P)$.
\end{Lem}

\begin{proof}
 Consider the setting $(M,g,E,P)$ with odd spacetime dimension $d$, let $p,q$ be causally related and choose $a,a'\in\R^2$ such that $\|a-a'\|^2 = \Gamma(p,q)$. It follows that $(p,a),(q,a')$ are lightlike related in $\big(M\times\R^2,g+g_{\mathrm{Eucl}}\big)$ and thus
 $$\widehat{U}_k\big((p,a),(q,a')\big) = \widehat{V}_k\big((q,a'),(p,a)\big)^t$$
 by assumption. Therefore, Proposition \ref{HCoeffProduct} provides $U_k(p,q) = V_k(q,p)^t$. \\
 Let $\{g^n\}_{n\in\N}$ be an analytic approximation of $g$ and $U_k^n, V_k^n$ the corresponding Hadamard coefficients. Write $D^n_k(p,q) := U^n_k(p,q) - V^n_k(q,p)^t$, which depends analytically on $p,q$ due to Proposition \ref{AnalyticHadamardCoefficients} and vanishes on $\Gamma^{-1}\big(\R_{\geq0}\big)$, so the identity theorem for analytic maps implies $D_k^n=0$ on all of $O\times O$. Furthermore, by Proposition \ref{AnalytAppr}, we have $D^n_k(p,q) \rightarrow D_k(p,q)$ and therefore $D_k(p,q)=0$ for all $p,q$, which proves the claim in the case of odd spacetime dimension.\\ 
 For even-dimensional settings $(M,g,E,P)$, this can be deduced from $\big(M\times\R,g+\df s^2,E,P-\pd{^2}{s^2}\big)$, which is odd-dimensional, and Proposition \ref{HCoeffProduct}.\\
 Since this works for any relatively compact and convex domain $O\subset\Omega$, by uniqueness of the Hadamard coefficients, an appropriate exhaustion of $\Omega$ by such subsets proves the claim on all of $\Omega\times\Omega$.
\end{proof}

\begin{Thm} \label{HadCoeffSymm}
 Let $M$ be a Lorentzian manifold of dimension $d$, $\pi\colon  E\rightarrow M$ a real or complex vector bundle over $M$ with non-degenerate inner product, $P\colon \CC(M,E)\rightarrow\CC(M,E)$ a formally self-adjoint wave operator and $\Omega\subset M$ a convex domain. Then the Hadamard coefficients $U_k\in\CC(\Omega\times\Omega,E^*\boxtimes E)$ are symmetric in the sense
 \begin{align} \label{SymmDiag}
  U_k(p,q) = \Theta_p U_k(q,p)^t\Theta_q^{-1}, \qquad p,q\in\Omega,~k\in\N_0.
 \end{align}%
\end{Thm}

\begin{proof}
 Let $d$ be odd. For all $k,j\in\N_0$ with $j\leq k$, Lemma 1.4.2 (1) of \cite{BGP2007} provides the recursion
 \begin{align} \label{RecursionRiesz}
  R^\Omega_\pm(2k+2,p) = K_{k,j,d} \Gamma(p,\cdot)^{k-j} \cdot R^\Omega_\pm(2j+2,p),
 \end{align}
 
 with $K_{k,j,d}\in\R\backslash\{0\}$ as defined in (\ref{Kjk}), so the smooth sections (\ref{DiffRieszSmooth}) can be rewritten into
 $$R^\Omega_\pm(2,p)\sum_{k=0}^\infty K_{k,0,d} \big(\widetilde{U}_k(p,\cdot) - \widetilde{V}_k(\cdot,p)^t\big)\Gamma(p,\cdot)^k.$$%
 
 The proof of Lemma 2.4.2 in \cite{BGP2007} shows that the sum defines a smooth section, which has to vanish for lightlike related $p,q$ since $\ssupp R^\Omega_\pm(2,p) = C_\pm^\Omega(p)$. Due to $\sigma \left(\frac{\Gamma(p,q)}{\e_k}\right)=\sigma(0) =1$, this leads to
 $$0=\sum_{k=0}^\infty \big(\widetilde{U}_k(p,q) - \widetilde{V}_k(q,p)^t\big)\Gamma(p,q)^k = U_0(p,q)-V_0(q,p)^t,$$
 
 and it follows from Lemma \ref{EqLightlikeEverywhere} that for $k=0$, (\ref{SymmDiag}) is true also for even $d$ and on all of $\Omega\times\Omega$. \\
 Now let $d$ again be odd, and for some $k_0\in\N$, assume (\ref{SymmDiag}) to hold for all $k=0,\hdots,k_0-1$, i.e.\ the smooth section (\ref{DiffRieszSmooth}) is given by
 $$\sum_{k=k_0}^\infty \big(\widetilde{U}_k(p,\cdot) - \widetilde{V}_k(\cdot,p)^t\big)R^\Omega_\pm(2k+2,p) = R^\Omega_\pm(2k_0+2,p)\sum_{k=k_0}^\infty K_{k,k_0,d}\big(\widetilde{U}_k(p,\cdot) - \widetilde{V}_k(\cdot,p)^t \big)\Gamma(p,\cdot)^{k-k_0}.$$ 
 
 Analogously, we obtain $0=U_{k_0}(p,q) - V_{k_0}(q,p)^t$ if $\Gamma(p,q)=0$, so again, applying Lemma \ref{EqLightlikeEverywhere} completes the proof by induction. 
\end{proof}

\setcounter{Def}{0}

\subsection{Derivation of transport equations}

Since Definition \ref{PCompConn} of the $P$-compatible connection $\nabla$ implies a product rule for $P(f\cdot s)$, for odd $d$, a straight forward calculation provides

{\allowdisplaybreaks \begin{align*}
 L^\Omega_\pm(0,p) & = \delta_p \stackrel{!}{=} P\L_\pm(p) = \sum_{k=0}^\infty P\left(U_p^k  L^\Omega_\pm(2k+2,p)\right) + \sum_{k=\frac{d-2}{2}}^\infty P\left(W_p^k  L^\Omega(2k+2,p)\right)\\[2mm]
 & \hspace{-1cm} = U_p^0  ~\boxx L_\pm^\Omega(2,p) - 2\nabla_{\grad L_\pm^\Omega(2,p)}U_p^0 + \sum_{k=1}^\infty\left( PU_p^{k-1}  L^\Omega_\pm(2k,p) - 2\nabla_{\grad L_\pm^\Omega(2k+2,p)}U_p^k + U_p^k ~\boxx L^\Omega_\pm(2k+2,p)\right) \\[2mm]
 & \hspace{1cm} + \sum_{k=\frac{d}{2}}^\infty \left(PW_p^{k-1}  L^\Omega(2k,p) - 2\nabla_{\grad L^\Omega(2k+2,p)}W_p^k + W_p^k  ~\boxx L^\Omega(2k+2,p)\right) \\[2mm]
 & \hspace{-1cm} = U_p^0  ~\boxx L_\pm^\Omega(2,p) - 2\nabla_{\grad L_\pm^\Omega(2,p)}U_p^0 + \sum_{k=1}^\infty\frac{1}{2k}\left(2k~PU_p^{k-1} - \nabla_{\grad \Gamma_p}U_p^k + \left(\frac{1}{2}\boxx\Gamma_p-d+2k\right)U_p^k \right)L^\Omega_\pm(2k,p) \\[2mm]
 & \hspace{1cm} + \sum_{k=\frac{d}{2}}^\infty \frac{1}{2k}\left(2k~PW_p^{k-1} - \nabla_{\grad \Gamma_p}W_p^k + \left(\frac{1}{2}\boxx\Gamma_p-d+2k\right)W_p^k \right)L^\Omega(2k,p).
\end{align*}}

For even $d$, we need to relate the logarithmic and non-logarithmic parts, for which we first prove the following technical Lemma:

\begin{Lem} \label{TechnicalLemma}
 Let $d$ be even and $k\in\N$ with $k\geq\frac{d}{2}$. Then we have
 \begin{align*}
   \begin{split} 
    & \grad \left(L^\Omega(2k+2,p) \cdot \log (\Gamma_p\pm i0)\right) \\[2mm]
    & \hspace{2cm} = \frac{\grad\Gamma_p}{4k} \cdot L^\Omega(2k,p) \log (\Gamma_p\pm i0) + \frac{\grad\Gamma_p}{2k(2k+2-d)} \cdot L^\Omega(2k,p),
   \end{split}\\[2mm]
   \begin{split} 
    & \boxx \left(L^\Omega(2k+2,p) \cdot \log (\Gamma_p\pm i0)\right) \\[2mm]
    & \hspace{2cm} = \frac{\frac{\Box\Gamma_p}{2} - d+2k}{2k} \cdot L^{\Omega}(2k,p)\cdot \log(\Gamma_p \pm i0) + \frac{\frac{\Box\Gamma_p}{2}-2d+4k+2}{2k\left(k-\frac{d-2}{2}\right)}\cdot L^{\Omega}(2k,p),
   \end{split}
 \end{align*}
 and for $k=\frac{d-2}{2}$
 \begin{align*}
  \grad \left(L^\Omega(d,p) \cdot \log (\Gamma_p\pm i0)\right) & = \mp \frac{i\pi\cdot\grad\Gamma_p}{2(d-2)} \cdot \widetilde{L}^\Omega_\pm(d-2,p) \\[2mm] 
  \boxx \left(L^\Omega(d,p) \cdot \log (\Gamma_p\pm i0)\right) & = \mp \frac{i\pi\left(\boxx\Gamma_p-4\right)}{4(d-2)} \cdot \widetilde{L}^\Omega_\pm(d-2,p). 
 \end{align*}
\end{Lem}

\begin{proof}
 Proposition \ref{PropLDomain} (3) provides
 {\allowdisplaybreaks \begin{align*}
  & \grad\left(L^\Omega(2k+2,p) \cdot\log (\Gamma_p\pm i0)\right) \\[2mm]
  & \hspace{2cm} = \log (\Gamma_p\pm i0) \cdot\grad L^\Omega(2k+2,p) + L^\Omega(2k+2,p) \cdot\grad\log (\Gamma_p\pm i0) \\[2mm]
  & \hspace{2cm} = \frac{\grad\Gamma_p}{4k} \cdot L^\Omega(2k,p) \log (\Gamma_p\pm i0) + \frac{\Gamma_p\cdot L^\Omega(2k,p)}{2k(2k+2-d)} \cdot \frac{\grad\Gamma_p}{\Gamma_p} \\[2mm]
  & \hspace{2cm} = \frac{\grad\Gamma_p}{4k} \cdot L^\Omega(2k,p) \log (\Gamma_p\pm i0) + \frac{\grad\Gamma_p}{2k(2k+2-d)} \cdot L^\Omega(2k,p),
 \end{align*}}
 and hence, $\braket{\grad\Gamma_p}{\grad\Gamma_p}=-4\Gamma_p$ implies
 {\allowdisplaybreaks \begin{align*}
  & \boxx \left(L^\Omega(2k+2,p) \cdot \log (\Gamma_p\pm i0)\right) \\[2mm]
  & \hspace{1cm} = -\mydiv \left(\frac{\grad\Gamma_p}{4k} \cdot L^\Omega(2k,p) \log (\Gamma_p\pm i0) + \frac{\grad\Gamma_p}{2k(2k+2-d)} \cdot L^\Omega(2k,p)\right) \\[2mm]
  & \hspace{1cm} = \frac{1}{4k}\left(\boxx\Gamma_p \cdot L^\Omega(2k,p)\log (\Gamma_p\pm i0) - \frac{L^\Omega(2k-2,p)\log (\Gamma_p\pm i0)}{4(k-1)} \braket{\grad\Gamma_p}{\grad\Gamma_p}\right. \\[2mm]
  & \hspace{2cm} \left. -L^\Omega(2k,p)\cdot \frac{\braket{\grad\Gamma_p}{\grad\Gamma_p}}{\Gamma_p}\right) + \frac{\boxx\Gamma_p\cdot L^\Omega(2k,p) - \frac{L^\Omega(2k-2,p)}{4(k-1)}\cdot \braket{\grad\Gamma_p}{\grad\Gamma_p}}{2k(2k+2-d)}\\[2mm]
  & \hspace{1cm} = \frac{1}{4k}\big(\boxx\Gamma_p + 2(2k-d)\big)L^\Omega(2k,p) \log (\Gamma_p\pm i0) + \left(\frac{1}{k} + \frac{\boxx\Gamma_p +2(2k-d)}{2k(2k+2-d)}\right)L^\Omega(2k,p) \\[2mm]
  & \hspace{1cm} = \frac{\frac{\Box\Gamma_p}{2} - d+2k}{2k} \cdot L^{\Omega}(2k,p) \log(\Gamma_p \pm i0) + \frac{\frac{\Box\Gamma_p}{2}-2n+4k+2}{2k\left(k-\frac{d-2}{2}\right)}\cdot L^{\Omega}(2k,p).
 \end{align*}}
 By Proposition \ref{PropLDomain} (1) $L^\Omega(d,p) = C(d,d) = \mp\frac{i\pi\widetilde{C}(d-2,d)}{2(d-2)}$ is constant on $\Omega$, so again using (3) yields
 {\allowdisplaybreaks \begin{align*}
  \grad \left(L^\Omega(d,p) \cdot \log (\Gamma_p\pm i0)\right) & = \mp\frac{i\pi\widetilde{C}(d-2,d)}{2(d-2)} \cdot \frac{\grad\Gamma_p}{\Gamma_p} = \mp\frac{i\pi\cdot \grad\Gamma_p}{2(d-2)} \cdot \widetilde{L}^\Omega_\pm(d-2,p), \\[5mm]
  \boxx \left(L^\Omega(d,p) \cdot \log (\Gamma_p\pm i0)\right) & = \pm\frac{i\pi}{2(d-2)} \mydiv\left(\grad\Gamma_p \cdot \widetilde{L}^\Omega_\pm(d-2,p)\right) \\[2mm]
  & = \mp\frac{i\pi}{2(d-2)} \left(\boxx\Gamma_p \cdot \widetilde{L}^\Omega_\pm(d-2,p) - \frac{\braket{\grad\Gamma_p}{\grad\Gamma_p}}{2(d-4)} \cdot \widetilde{L}^\Omega_\pm(d-4,p)\right) \\[2mm]
  & = \mp\frac{i\pi}{2(d-2)} \left(\boxx\Gamma_p + \frac{4}{2(d-4)} \cdot (d-4)(d-4-d+2)\right)\cdot \widetilde{L}^\Omega_\pm(d-2,p) \\[2mm]
  & = \mp \frac{i\pi\left(\boxx\Gamma_p-4\right)}{2(d-2)} \cdot \widetilde{L}^\Omega_\pm(d-2,p). \qedhere
 \end{align*}}
\end{proof}

This leads to the following calculation:

{\allowdisplaybreaks \begin{align*}
 & \widetilde{L}^\Omega_\pm(0,p) = \delta_p \stackrel{!}{=} P\L_\pm(p)\\
 & \hspace{0.3cm} = \sum_{k=0}^{\frac{d-4}{2}} P\left(U_p^k  \widetilde{L}^\Omega_\pm(2k+2,p)\right) \pm \frac{i}{\pi} \sum_{k=\frac{d-2}{2}}^\infty P\left(U_p^k  L^\Omega(2k+2,p)  \log (\Gamma_p \pm i0)+W_p^k  L^\Omega(2k+2,p)\right) \\
 & \hspace{0.3cm} = U_p^0  ~\boxx \widetilde{L}^\Omega_\pm(2,p) - 2\nabla_{\grad \widetilde{L}^\Omega_\pm(2,p)} U_p^0 \pm\frac{i}{\pi} \bigg( W_p^{\frac{d-2}{2}}  \overbrace{\boxx L^\Omega(d,p)}^{=0} - \overbrace{2\nabla_{\grad L^\Omega(d,p)} W_p^{\frac{d-2}{2}}}^{=0}\bigg) \\
 & \hspace{0.6cm} + \sum_{k=1}^{\frac{d-4}{2}} \left(U_p^k  ~\boxx \widetilde{L}^\Omega_\pm(2k+2,p) - 2\nabla_{\grad \widetilde{L}^\Omega_\pm(2k+2,p)} U_p^k + \widetilde{L}^\Omega_\pm(2k,p)  PU_p^{k-1}\right)+ \widetilde{L}^\Omega_\pm(d-2,p)  PU_p^{\frac{d-4}{2}} \\
   & \hspace{0.6cm} \pm \frac{i}{\pi} \bigg\{U_p^{\frac{d-2}{2}}  \boxx\! \left(L^\Omega(d,p)  \log(\Gamma_p \pm i0)\right) - 2 \nabla_{\grad \left(L^\Omega(d,p)  \log (\Gamma_p \pm i0)\right)} U_p^{\frac{d-2}{2}} \\
   & \hspace{0.6cm} + \sum_{k=\frac{d}{2}}^\infty \left[L^\Omega(2k,p) \log (\Gamma_p \pm i0) PU_p^{k-1} - 2\nabla_{\grad \left(L^\Omega(2k+2,p) \log (\Gamma_p \pm i0)\right)} U_p^k + U_p^k ~\boxx\! \left(L^\Omega(2k+2,p) \log (\Gamma_p \pm i0)\right)\right] \\
   & \hspace{0.6cm} + \sum_{k=\frac{d}{2}}^\infty \left(W_p^k  ~\boxx L^\Omega(2k+2,p) - 2\nabla_{\grad L^\Omega(2k+2,p)} W_p^k + L^\Omega(2k,p)  PW_p^{k-1}\right)\bigg\}\\
 & \hspace{0.3cm} = U_p^0 ~\boxx \widetilde{L}^\Omega_\pm(2,p) - 2\nabla_{\grad \widetilde{L}^\Omega_\pm(2,p)} U_p^0 + \sum_{k=1}^{\frac{d-4}{2}} \left[\left(\frac{1}{2}\boxx\Gamma_p - d+2k\right)U_p^k - \nabla_{\grad \Gamma_p} U_p^k + 2k  PU_p^{k-1}\right]\frac{\widetilde{L}^\Omega_\pm(2k,p)}{2k} \\
 & \hspace{0.6cm} + \left((d-2) PU_p^{\frac{d-4}{2}} + \left(\frac{1}{2}\boxx\Gamma_p-2\right)  U_p^{\frac{d-2}{2}} - \nabla_{\grad\Gamma_p}U_p^{\frac{d-2}{2}}\right)\frac{\widetilde{L}^\Omega_\pm(d-2,p)}{d-2}\\ 
 & \hspace{0.6cm} \pm \frac{i}{\pi} \sum_{k=\frac{d}{2}}^\infty \left[\left(\frac{1}{2}\boxx\Gamma_p - d+2k\right)U_p^k - \nabla_{\grad \Gamma_p} U_p^k + 2k  PU_p^{k-1}\right]\frac{L^\Omega(2k,p)  \log(\Gamma_p \pm i0)}{2k} \\
 & \hspace{0.6cm} \pm \frac{i}{\pi} \sum_{k=\frac{d}{2}}^\infty \left[\left(\frac{1}{2}\boxx\Gamma_p +2+4k-2d\right) U_p^k - \nabla_{\grad \Gamma_p}U_p^k\right]\frac{L^\Omega(2k,p)}{2k\left(k-\frac{d-2}{2}\right)} \\ 
 & \hspace{0.6cm} \pm \frac{i}{\pi} \sum_{k=\frac{d}{2}}^\infty \left[\left(\frac{1}{2}\boxx\Gamma_p +2k-d\right) W_p^k - \nabla_{\grad \Gamma_p}W_p^k + 2k  PW_p^{k-1}\right]\frac{L^\Omega(2k,p)}{2k},
\end{align*}}
from which one reads off the transport equations (\ref{TranspEqU}) and (\ref{TranspEqW}).

\setcounter{Def}{0}

\subsection{Proof of Proposition \ref{ParametricesProp} for even dimensional spacetimes}

Note that in the proofs of all following Lemmas, $c$ denotes a generic constant, i.e. its particular value can change from one line to another.

\begin{Lem} \label{EstSigma} 
 For all $l\in\N$ and $\beta \geq l+1$, there is some $c_{l,\beta}$ such that for all $0<\e\leq 1$, we have
 \begin{align*}
  \left\|\td{^l}{t^l}\left(\sigma\left(\frac{t}{\e}\right)\cdot t^\beta \cdot \log t\right)\right\|_{C^0(\R)} \leq \e\left(\log\frac{1}{\e}+\pi+1\right) \cdot c_{l,\beta} \cdot \|\sigma\|_{C^l(\R)}.
 \end{align*}
\end{Lem}

\begin{proof}
 We start with calculating
 \begin{align*}
  & \left\|\td{^j}{t^j} \left(t^\beta \cdot \log t\right)\right\|_{C^0(\R)} \\[2mm]
  & \hspace{0.5cm} = \left\|\beta(\beta-1)\hdots (\beta-j+1)t^{\beta-j} \log t + t^{\beta-j}\sum_{i=1}^j \binom{j}{i}(-1)^{i-1}(i-1)! \beta(\beta-1)\hdots(\beta-j+i+1) \right\|_{C^0(\R)} \\[2mm]
  & \hspace{0.5cm} \leq c_{j,\beta} \cdot |t|^{\beta-j}\left(|\log t| + 1\right) \\[2mm]
  & \hspace{0.5cm} \leq c_{j,\beta} \cdot |t|^{\beta-j}\left(\big|\log |t|\big| +\pi+ 1\right).
 \end{align*}
 Since $\sigma^{(l-j)}\left(\frac{t}{\e}\right)=0$ for $|t|\geq \e$ and $\e^{\beta-l} \leq \e$ due to $\beta \geq l+1$, this yields
 \begin{align*}
  \left\|\td{^l}{t^l}\left(\sigma\left(\frac{t}{\e}\right) \cdot t^\beta \cdot \log t\right)\right\|_{C^0(\R)} & \leq \sum_{j=0}^l \binom{l}{j} \cdot c_{j,\beta}\left\|\frac{\sigma^{(l-j)}\left(\frac{t}{\e}\right)}{\e^{l-j}}  \cdot |t|^{\beta-j}\left(\big|\log |t|\big|+\pi + 1\right) \right\|_{C^0(\R)} \\[2mm]
  & \leq \sum_{j=0}^l \binom{l}{j} \cdot c_{j,\beta} \cdot \e^{\beta-l} \left(\log \frac{1}{\e}+\pi + 1\right) \cdot \left\|\sigma^{(l-j)}\right\|_{C^0(\R)} \\[2mm]
  & \leq c_{l,\beta}\cdot \e \left(\log \frac{1}{\e}+\pi + 1\right) \cdot \left\|\sigma\right\|_{C^l(\R)}.
 \end{align*}
\end{proof}

\begin{Lem} \label{ParametricesReg} 
 For any open and relatively compact domain $O \subset \Omega$ and $l \in \N_0$, there is a sequence $\{\e_k\}_{k\geq N} \subset (0,1]$ such that for all $l\geq 0$ the series
 \begin{align} \label{ClSum}
  (p,q) \longmapsto \sum_{k=N+l}^\infty \widetilde{U}_k(p,q) \log \big(\Gamma(p,q) \pm i0\big) ~ L^\Omega_\pm(2k+2,p)(q)
 \end{align}
 converges in $C^l\left(\Obar \times \Obar, E^* \boxtimes E\right)$. In particular, for $l=0$, this defines a continuous section over $\Obar\times\Obar$ and a smooth section over $\big(\Obar\times\Obar\big)\backslash\Gamma^{-1}(0)$.
\end{Lem}

\begin{proof}
 Since $k \geq N \geq \frac{d}{2}$ and $d$ even, $\Gamma(p,q)^{k-\frac{d-2}{2}}$ is a smooth and $\Gamma(p,q)^{k-\frac{d-2}{2}} \cdot \log \left(\Gamma(p,q) \pm i0\right)$ a continuous section over $\Obar \times \Obar$, so every single summand of (\ref{ClSum}) is at least continuous, individually. Due to $\supp \left(\sigma \circ \frac{\Gamma}{\e_k}\right) \subset \{\Gamma(p,q)<\e_k\}$ for all $k \geq N$ and $0\leq\sigma\leq1$ by Lemma \ref{EstSigma}, we have 
 \begin{align*}
  \left\|\widetilde{U}_k \log \big(\Gamma \pm i0\big)\cdot C(2k+2,d)\Gamma^{k-\frac{d-2}{2}}\right\|_{C^0\left(\Obar \times \Obar\right)} \leq c_{k,d}\left\|U_k\right\|_{C^0\left(\Obar \times \Obar\right)} \cdot \e_k\left(\log\frac{1}{\e_k} + \pi+1\right).
 \end{align*}
 Since $\e_k\log\frac{1}{\e_k} \rightarrow 0$ for $\e_k \rightarrow 0$, we can choose $\e_k$ such that
 \begin{align*} 
  c_{k,d}\left\|U_k\right\|_{C^0} \cdot \e_k\left(\log\frac{1}{\e_k} + \pi+1\right) < 2^{-k},
 \end{align*}
 and (\ref{ClSum}) converges in $C^0$. Now let $l>0$ and $k \geq N+l$ such that $\Gamma(p,q)^{k-\frac{d-2}{2}} \log \left(\Gamma(p,q) \pm i0\right)$ is of $C^l$-regularity. Set $\rho_k(t) := \sigma\left(\frac{t}{\e_k}\right) t^{k-\frac{d-2}{2}}$, so by Lemma \ref{EstSigma} we have
  \begin{align} \label{EstRhok} 
  \|\rho_k\|_{C^l(\R)} \leq c_{l,k,d} \cdot \e_k\|\sigma\|_{C^l(\R)},\qquad \|\rho_k \cdot \log\|_{C^l(\R)} \leq c_{l,k,d} \cdot\e_k\left(\log\frac{1}{\e_k}+\pi+1\right) \|\sigma\|_{C^l(\R)},
 \end{align}
 and Lemma 1.1.11 and 1.1.12 of \cite{BGP2007} yield
  \begin{align*}
  \left\|\widetilde{U}_k \log \big(\Gamma \pm i0\big)\cdot C(2k+2,d)\Gamma^{k-\frac{d-2}{2}}\right\|_{C^l\left(\Obar \times \Obar\right)} & \leq c_{l,k,d}  \|U_k \cdot \big((\rho_k \cdot \log) \circ \Gamma\big)\|_{C^l\left(\Obar \times \Obar\right)} \\[2mm]
  & \hspace{-5cm} \stackrel{\text{1.1.11}}{\leq} c_{l,k,d}  \|U_k\|_{C^l\left(\Obar \times \Obar\right)} \cdot \|(\rho_k \cdot \log)\circ \Gamma\|_{C^l\left(\Obar \times \Obar\right)} \\[2mm] 
  & \hspace{-5cm} \stackrel{\text{1.1.12}}{\leq} c_{l,k,d} \|U_k\|_{C^l\left(\Obar \times \Obar\right)} \cdot \max_{j=0,\hdots,l}\|\Gamma\|^j_{C^l\left(\Obar \times \Obar\right)} \cdot \|\rho_k \cdot \log\|_{C^l(\R)}\\[2mm] 
  & \hspace{-5cm} \stackrel{(\ref{EstRhok})}{\leq} c_{l,k,d}  \|\sigma\|_{C^l(\R)} \cdot \max_{j=0,\hdots,l}\|\Gamma\|^j_{C^l\left(\Obar \times \Obar\right)} \cdot \|U_k\|_{C^l\left(\Obar \times \Obar\right)} \cdot \e_k\log\left(\frac{1}{\e_k}+\pi+1\right).
 \end{align*}
 Hence, for all $k\geq N-l$, we demand
 \begin{align} \label{Demandl}
  c_{l,k,d} \|U_k\|_{C^l\left(\Obar \times \Obar\right)} \cdot \e_k\log\left(\frac{1}{\e_k}+\pi+1\right) \leq 2^{-k},
 \end{align}
 so the $k.$ summand can be estimated by $\frac{\|\sigma\|_{C^l\left(\R\right)}}{2^k} \cdot \max_{j=0, \hdots, l} \|\Gamma\|_{C^l\left(\Obar \times \Obar\right)}^j$ and (\ref{ClSum}) converges in $C^l\left(\Obar \times \Obar, E^* \boxtimes E\right)$. Note that for each $k$, we impose only finitely many conditions on $\e_k$, namely one for each $l\leq k-N$, which are satisfied by some positive number. Hence, for each $k$, there is a sufficiently small number $\e_k>0$ such that (\ref{Demandl}) is fulfilled for all $l\leq k-N$.\\
 Since all summands are smooth on $\left(\Obar \times \Obar\right)\backslash\Gamma^{-1}(0)$ and the series converges in all $C^l$-norms, it defines a smooth section on $\left(\Obar \times \Obar\right)\backslash\Gamma^{-1}(0)$.
\end{proof}

Thus, we showed that (\ref{Parametrices}) yield well-defined distributions with singular support on the light cone, i.e.\ property (i). Furthermore, (iii) follows from Proposition \ref{PropLDomain} (6). We proceed with (ii):

\begin{Lem}
 The sequence $\big(\e_k\big)_{k\geq N}$ can be chosen such that
 $$P_{(2)}\widetilde{\L}_\pm(p)=\delta_p + K_\pm(p,\cdot)$$ 
 for some $K_\pm\in \CC\left(\Obar \times \Obar, E^* \boxtimes E\right)$.
\end{Lem}

\begin{proof}
 Let $\sigma_k:=\sigma \circ \frac{\Gamma}{\e_k}\in\CC(\Omega\times\Omega)$, so $\widetilde{U}_k = \sigma_k\cdot U_k$ for all $k\geq N$ and $\supp\sigma_k\subset\big\{\Gamma(p,q)\leq\e_k\big\}$. Due to Lemma 1.1.10 of \cite{BGP2007}, we can exchange $P$ with the sum, so the transport equations (\ref{TranspEqU}) imply
 \begin{align*}
  & \sum_{k=N}^\infty P_{(2)}~\sigma_k \big(U_k \log(\Gamma\pm i0) + W_k\big) L^\Omega_\pm(2k+2) \\[2mm]
  & \hspace{3cm} = \sum_{k=N}^\infty\left(\boxx_{(2)}\sigma_k - 2\nabla_{\grad_{\!(2)}\sigma_k} + \sigma_k P_{(2)}\right)\big(U_k \log(\Gamma\pm i0) + W_k\big) L^\Omega_\pm(2k+2) \\[2mm]
  & \hspace{3cm} =: \Sigma_1 + \Sigma_2 + \sum_{k=N}^\infty\sigma_k P_{(2)}\big(U_k \log(\Gamma\pm i0) + W_k\big) L^\Omega_\pm(2k+2).
 \end{align*}
 Recall that the transport equations (\ref{TranspEqU}) and (\ref{TranspEqW}) are derived from the requirement that $P_{(2)}$ applied to (\ref{HadamardSeries}) is a telescoping series, that is,
 \begin{align*} 
  \begin{split}
   & P_{(2)}\big(U_k \log(\Gamma\pm i0) + W_k\big) L^\Omega_\pm(2k+2) \\[2mm]
   & \hspace{0.5cm} = \big(\log(\Gamma\pm i0)\cdot P_{(2)}U_k + P_{(2)}W_k\big)L^\Omega_\pm(2k+2)  - \big(\log(\Gamma\pm i0)\cdot P_{(2)}U_{k-1} + P_{(2)}W_{k-1}\big)L^\Omega_\pm(2k).
  \end{split}
 \end{align*}
 Hence, the right hand side becomes
 \begin{align*}
  \Sigma_1 + \Sigma_2 - \sigma_N \big(\log(\Gamma \pm i0)\cdot P_{(2)}U_{N-1} + P_{(2)}W_{N-1}\big)L^{\Omega}_\pm(2N) + \Sigma_3
 \end{align*}
 with 
 $$\Sigma_3 := \sum_{k=N}^\infty (\sigma_k-\sigma_{k+1})\big(\log(\Gamma \pm i0) \cdot P_{(2)}U_k  + P_{(2)}W_k\big) L^{\Omega}_\pm(2k+2).$$
 Then, for $K_\pm(p,\cdot) := P_{(2)}\widetilde{\L}_\pm(p) - \delta_p$, the transport equations (\ref{TranspEqU}) for $U_k$ yield
 
 \begin{align} \label{DefKpm}
  K_\pm = (1-\sigma_N)\big(\log(\Gamma \pm i0)\cdot P_{(2)}U_{N-1} + P_{(2)}W_{N-1}\big)L^{\Omega}_\pm(2N) + \Sigma_1+\Sigma_2+\Sigma_3.
 \end{align}
 
 On the right hand side, every summand individually yields a smooth section, since both $1-\sigma_N$ and $\sigma_k - \sigma_{k+1}$ as well as all derivatives of $\sigma_k$ vanish in a neighborhood of $\Gamma^{-1}(0)$, which contains the singular support of $\widetilde{\L}_\pm$. Thus, $K_\pm$ vanishes on $\Gamma^{-1}(0)$ to arbitrary order and it remains to show convergence in all $C^l$-norms, which again for the $W$-part is provided by the proof of Lemma 2.4.3 of \cite{BGP2007}. Therefore, we concentrate on
 $$\sum_{k=N}^\infty\left(\boxx_{(2)}\sigma_k - 2\nabla_{\grad_{\!(2)}\sigma_k} + \sigma_k P_{(2)}\right)U_k \log(\Gamma\pm i0)~L^\Omega_\pm(2k+2) =: \Sigma_1'+\Sigma_2'+\Sigma_3'.$$
 For fixed $l\in \N_0$, let $k \geq 2(l+1)+N$ and $S_k := \left\{\frac{\e_k}{2} \leq \left|\Gamma(p,q)\right| \leq \e_k\right\}$. Then Lemma 1.1.12 of \cite{BGP2007} implies for the $k.$ summand of $\Sigma_2'$
 {\allowdisplaybreaks \begin{align*}
  & \left\|\nabla_{\grad_{\!(2)}\sigma_k} U_k\log\big(\Gamma\pm i0\big) ~ L^{\Omega}_\pm(2k+2)\right\|_{C^l\left(\Obar \times \Obar\right)} = \left\|\nabla_{\grad_{\!(2)}\sigma_k} U_k\log\big(\Gamma\pm i0\big) ~ L^{\Omega}_\pm(2k+2)\right\|_{C^l\left(\Obar \times \Obar \cap S_k\right)} \\[2mm]
  & \hspace{0.5cm} \leq c_{l,d}\|U_k\|_{C^{l+1}\left(\Obar \times \Obar \cap S_k\right)} \cdot \|\sigma_k\|_{C^{l+1}\left(\Obar \times \Obar \cap S_k\right)} \cdot \left\|\Gamma^{k-\frac{d-2}{2}} \log\big(\Gamma\pm i0\big)\right\|_{C^{l+1}\left(\Obar \times \Obar \cap S_k\right)} \\[2mm]
  & \hspace{0.5cm} \leq c_{l,d} \|U_k\|_{C^{l+1}\left(\Obar \times \Obar \cap S_k\right)} \cdot \frac{\|\sigma\|_{C^{l+1}(\R)}}{\e_k^{l+1}} \cdot \left\|t \mapsto t^{k-\frac{d-2}{2}}\cdot \log t\right\|_{C^{l+1}\left(\left[\frac{\e_k}{2},\e_k\right]\right)}\cdot \max_{j=0,\cdots, l+1} \|\Gamma\|^{2j}_{C^{l+1}\left(\Obar \times \Obar \cap S_k\right)} \\[2mm]
  & \hspace{0.5cm} \leq c_{l,k,d}\|U_k\|_{C^{l+1}\left(\Obar \times \Obar \cap S_k\right)} \frac{\|\sigma\|_{C^{l+1}(\R)}}{\e_k^{l+1}} \max_{j=0,\cdots, l+1} \|\Gamma\|^{2j}_{C^{l+1}\left(\Obar \times \Obar \cap S_k\right)}  \max_{t\in\left(\left[\frac{\e_k}{2},\e_k\right]\right)} |t|^{k-\frac{d-2}{2}-(l+1)}\left(|\log t|+1\right) \\[2mm]
  & \hspace{0.5cm} \leq c_{l,k,d}\|U_k\|_{C^{l+1}\left(\Obar \times \Obar\right)}\cdot \e_k\left(\log\frac{1}{\e_k}+\pi+1\right)\|\sigma\|_{C^{l+1}(\R)} \cdot \max_{j=0,\cdots, l+1} \|\Gamma\|^{2j}_{C^{l+1}\left(\Obar \times \Obar\right)},
 \end{align*}}
 so we additionally demand
 $$c_{l,k,d}\cdot \left\|U_k\right\|_{C^{l+1}\left(\Obar \times \Obar\right)} \cdot \e_k\left(\log\frac{1}{\e_k}+\pi+1\right) \leq 2^{-k}.$$
 Then, for all $l$, the $C^l$-norm of almost all summands (without the first $2(l+1)+N$) of $\Sigma_2'$ is bounded by $2^{-k}\cdot \|\sigma\|_{C^{l+1}(\R)} \cdot \max_{j=0,\cdots, l+1} \|\Gamma\|^{2j}_{C^{l+1}\left(\Obar \times \Obar\right)}$ and thus, we have convergence in $C^l$ for all $l$, i.e.\ $\Sigma_2'$ defines a smooth section in $E^* \boxtimes E$ over $\Obar \times \Obar$ .\\
 The treatment for $\Sigma_1'$ is completely identical, so we directly turn to the $k.$ summand of $\Sigma_3'$:
 \begin{align*}
  & \left\|(\sigma_k-\sigma_{k+1}) L^{\Omega}_\pm(2k+2,\cdot)\log\left(\Gamma \pm i0\right) \cdot P_{(2)}U_k\right\|_{C^l\left(\Obar \times \Obar\right)}\\[2mm]
  & \hspace{1cm} \leq c_{l,k,d} \left(\left\|\sigma_k~\Gamma^{l+1} \log\left(\Gamma \pm i0\right)\right\|_{C^l} + \left\|\sigma_{k+1}~ \Gamma^{l+1} \log\left(\Gamma \pm i0\right)\right\|_{C^l}\right)\left\|\Gamma^{k-N-l}\right\|_{C^l} \cdot \left\|P_{(2)}U_k\right\|_{C^l}. 
 \end{align*}
 Set $\rho_{kl}(t) := \sigma\left(\frac{t}{\e_k}\right) \cdot t^{l+1} \log t$, so we have $\sigma_k ~ \Gamma^{l+1} \log\left(\Gamma \pm i0\right) = \rho_{kl} \circ \Gamma$. Then again Lemma 1.1.12 of \cite{BGP2007} and Lemma \ref{EstRhok} yield
 \begin{align*}
  \left\|\sigma_k~\Gamma^{l+1} \log\left(\Gamma \pm i0\right)\right\|_{C^l\left(\Obar \times \Obar\right)} & \leq c_l \cdot \|\rho_{kl}\|_{C^l(\R)} \cdot \max_{j=0, \hdots, l} \|\Gamma\|_{C^l\left(\Obar \times \Obar\right)} \\[2mm]
  & \leq c_{k,l} \cdot \e_k\left(\log\frac{1}{\e_k}+\pi+1\right)\|\sigma\|_{C^l(\R)} \cdot \max_{j=0, \hdots, l} \|\Gamma\|_{C^l\left(\Obar \times \Obar\right)},
 \end{align*}
 so we obtain
 \begin{align*}
  & \|\Sigma_3'\|_{C^l\left(\Obar \times \Obar\right)} \leq c_{l,k,d} \left(\e_k\left(\log\frac{1}{\e_k}+\pi+1\right) + \e_{k+1}\left(\log\frac{1}{\e_{k+1}}+\pi+1\right)\right)\|\sigma\|_{C^l(\R)} \\[2mm]
  & \hspace{6cm} \cdot \max_{j=0, \hdots, l} \|\Gamma\|_{C^l\left(\Obar \times \Obar\right)} \cdot \left\|\Gamma^{k-N-l}\right\|_{C^l\left(\Obar \times \Obar\right)} \cdot \left\|P_{(2)}U_k\right\|_{C^l\left(\Obar \times \Obar\right)}.
 \end{align*}
 Hence, for all $k\geq N+l$ we demand
 $$c_{l,k,d} \cdot \e_k\left(\log\frac{1}{\e_k}+\pi+1\right)\cdot \left\|P_{(2)}U_k\right\|_{C^l\left(\Obar \times \Obar\right)} \cdot \left\|\Gamma^{k-N-l}\right\|_{C^l\left(\Obar \times \Obar\right)} \leq 2^{-k-1}$$
 as well as for all $k \geq N+l+1$ that
 $$c_{l,k-1,d} \cdot \e_k\left(\log\frac{1}{\e_k}+\pi+1\right)\cdot \left\|P_{(2)}U_{k-1}\right\|_{C^l\left(\Obar \times \Obar\right)} \cdot \left\|\Gamma^{k-N-l-1}\right\|_{C^l\left(\Obar \times \Obar\right)} \leq 2^{-k-2}.$$
 Then the $C^l$-norm of almost all summands of $\Sigma_3'$ (without the first $l+N$) is bounded by $2^{-k} \cdot \|\sigma\|_{C^l(\R)} \cdot \max_{j=0, \hdots, l} \|\Gamma\|_{C^l\left(\Obar \times \Obar\right)}$, so the series converges in all $C^l$-norms and is therefore smooth. Note that for each $k$ we again added only finitely many conditions.
\end{proof}

Finally, we show that the $\e_k$'s can be chosen such that for all $p\in\Obar$, the parametrices $\widetilde{\L}_\pm(p)$ are distributions of degree at most $\kappa_d$.

\begin{Lem} \label{ApproxFundSolSmooth}
 There is a sequence $(\e_k)_{k\geq N}\subset(0,1]$, for which we find some $C>0$ such that
 \begin{align*} 
  \left|\widetilde{\L}_\pm(p)[\f]\right| \leq C \cdot \|\f\|_{C^{\kappa_d}(\Omega)}, \qquad p \in \Obar,~\f \in \D\left(\Omega,E^*\right).
 \end{align*}
 Furthermore, for fixed $\f \in \D\left(\Omega,E^*\right)$, the map $p \mapsto \widetilde{\L}_\pm(p)[\f]$ is smooth.
\end{Lem}

\begin{proof}
 We show the claim only for the logarithmic part, i.e.\ $f:=\sum^\infty_{k=\frac{d-2}{2}}\widetilde{U}_k \log(\Gamma\pm i0) L^\Omega_\pm(2k+2)$, since for the other two sums the proof of Lemma 2.4.4 of \cite{BGP2007} applies identically. By Lemma \ref{ParametricesReg}, we have $f\in C^0(\Obar\times\Obar,E^*\boxtimes E)$ and thus,
 $$|f(p)[\f]| \leq \|f\|_{C^0(\Obar\times\Obar)} \cdot \vol(\Obar) \cdot\|\f\|_{C^0(\Obar\times\Obar)} \leq \|f\|_{C^{\kappa_d}(\Obar\times\Obar)} \cdot \vol(\Obar) \cdot\|\f\|_{C^{\kappa_d}(\Obar\times\Obar)}$$
 for all $p\in O$ and $\f\in\D(O,E^*)$, so the constant can be chosen via $C:=\|f\|_{C^{\kappa_d}(\Obar\times\Obar)} \cdot \vol(\Obar)$.\\ 
 Since Proposition \ref{PropLDomain} (6) directly applies also to $\log(\Gamma_p\pm i0) L^\Omega_\pm(2k+2,p)$ and $U_k\f$ is smooth on $O\times O$ with $\supp (U_p^k)\f$ compact, for every $k\geq \frac{d-2}{2}$, the map
 $$p\longmapsto \widetilde{U}_p^k \log(\Gamma_p\pm i0) L^\Omega_\pm(2k+2,p)[\f], \qquad \f\in\D(O,E^*),$$
 is smooth. Therefore, also $\sum_{k=\frac{d-2}{2}}^{l-1}\widetilde{U}_p^k \log(\Gamma_p\pm i0) L^\Omega_\pm(2k+2,p)[\f]$ is smooth for all $l>\frac{d}{2}$ and the remaining term $\sum_{k=l}^\infty\widetilde{U}_p^k \log(\Gamma_p\pm i0) L^\Omega_\pm(2k+2,p)[\f]$ is $C^l$ by Lemma \ref{ParametricesReg}. This holds for all $l$ and hence, $p\mapsto f(p)[\f]$ is smooth.
\end{proof}

Lemma \ref{ApproxFundSolSmooth} shows properties (iv) and (v) of $\widetilde{\L}_\pm$, so Proposition \ref{ParametricesProp} is proved also for even $d$.

\end{spacing}

\end{document}